\documentclass[11pt]{article}

\usepackage[margin=1in]{geometry}
\usepackage{amsmath,amssymb,amsthm,mathtools}
\usepackage[T1]{fontenc}
\usepackage{lmodern}
\usepackage{graphicx}
\usepackage{authblk}
\usepackage{braket}
\usepackage{xcolor}
\usepackage{framed}
\usepackage{hyperref}
\usepackage{enumitem}
\usepackage{colortbl}
\usepackage[backend=biber,style=alphabetic,maxbibnames=99,giveninits=true]{biblatex}
\usepackage{complexity}
\renewcommand{\P}{\mathrm{P}}
\renewcommand{\BPP}{\mathrm{BPP}}
\renewcommand{\BQP}{\mathrm{BQP}}
\renewcommand{\PH}{\mathrm{PH}}
\renewcommand{\poly}{\mathrm{poly}}

\hypersetup{
  colorlinks=true,
  linkcolor=blue!60!black,
  citecolor=blue!60!black,
  urlcolor=blue!60!black,
  pdftitle={Oracle Separations in the Fourier Hierarchy},
    pdfpagemode=FullScreen,
}

\newenvironment{boxedstatement}[1]{%
  \par\medskip\noindent\begin{minipage}{\linewidth}%
  \begin{framed}\noindent\textbf{#1}\par\medskip
}{%
  \end{framed}\end{minipage}\par\medskip
}

\newtheorem{theorem}{Theorem}[section]
\newtheorem{lemma}[theorem]{Lemma}
\newtheorem{proposition}[theorem]{Proposition}
\newtheorem{corollary}[theorem]{Corollary}
\theoremstyle{definition}
\newtheorem{definition}[theorem]{Definition}
\newtheorem{remark}[theorem]{Remark}
\newtheorem{conjecture}[theorem]{Conjecture}

\renewcommand{\FH}{\mathrm{FH}}
\newcommand{\Yes}{\mathsf{Yes}}
\newcommand{\No}{\mathsf{No}}
\newcommand{\bits}{\{0,1\}}
\newcommand{\pmone}{\{\pm 1\}}
\newcommand{\Ftwo}{\mathbb{F}_2}
\renewcommand{\E}{\mathbb{E}}
\newcommand{\Prb}{\Pr}
\newcommand{\Had}{H^{\otimes n}}
\newcommand{\Norm}[1]{\left\lVert #1 \right\rVert}
\newcommand{\abs}[1]{\left| #1 \right|}

\newcommand{\sgnstar}{\operatorname{sgn}_*}
\newcommand{\sgn}{\operatorname{sgn}}

\title{Oracle Separations in the Fourier Hierarchy}
\newif\ifanon
\anonfalse
\ifanon
  \author{}
\else
  \author{Atul Mantri~\thanks{atulmantri@vt.edu}}
  \affil{Department of Computer Science, Virginia Tech, USA 24061}
\fi
\date{}

\begin{document}
\maketitle

\begin{abstract}
The Fourier hierarchy $\FH_0\subseteq\FH_1\subseteq\FH_2\subseteq\cdots$, introduced by Shi (TCS 2005), measures a quantum computation by the number of Hadamard layers it uses. Between two layers the circuit may permute basis states and attach phases, but it may not create superposition; the layers are its only source of interference. The first level is exactly $\BPP$, while the second already solves Simon's problem and, through phase estimation, factors integers. Shi conjectured that every additional layer strictly increases computational power, and asked, as a first step, for oracle separations between consecutive levels. To our knowledge, the question was open at every level $k\ge2$.

We prove that for every constant $k\ge2$ there is an oracle relative to which $\FH_k\subsetneq\FH_{k+1}$. The separating problem is built from Forrelation (Aaronson and Ambainis, STOC 2015): the level above solves it with a constant number of queries, whereas at level $k$ it stays hard even for circuits making exponentially many queries. This holds for both of the usual ways of giving a circuit access to an oracle, the phase oracle and the standard oracle, which writes its answer into a register. The two are not interchangeable: relative to an oracle, the standard oracle is strictly more powerful at the same number of layers. We also separate the union of all the levels from $\BQP$ relative to an oracle.

The lower bounds rest on a structural property of the hierarchy: the number of Hadamard layers limits how adaptively a circuit can query its oracle. With a phase oracle, a circuit with $k$ layers is reproduced exactly by an algorithm making only $k-1$ rounds of parallel queries, which brings known lower bounds for such algorithms to bear. The standard oracle lets a circuit branch on earlier answers, and that case needs a separate argument.
\end{abstract}

\newpage
\begingroup
\small
\tableofcontents
\endgroup

\newpage
\section{Introduction}

The Fourier hierarchy, introduced by Shi~\cite{Shi_2005}, measures a quantum computation by the
number of global Hadamard layers it uses.  Between two consecutive layers the computation may carry
out classical reversible operations, which permute computational basis states, and attach phases to
them, but it may not create new superpositions; such a computation is called
\emph{basis-preserving}.  The Hadamard layers are therefore the only source of interference, and
their number is the \emph{Fourier depth} of the computation.  An $\FH_k$ computation has Fourier
depth at most $k$, and
\[
\FH_0\subseteq \FH_1\subseteq \FH_2\subseteq\FH_3\subseteq\cdots.
\]
The hierarchy is already expressive at its lowest levels.  With no Hadamard layer a computation is
deterministic, and with one it can do no more than toss coins and compute classically, so
$\FH_0=\mathrm P$ and $\FH_1=\BPP$, as Shi observes~\cite{Shi_2005} (Lemmas~\ref{lem:fh0-std}
and~\ref{lem:fh1-std} prove the relativized forms in the standard-query model).  A single pair
of layers already suffices for Simon's algorithm and, through Kitaev's phase estimation, for
factoring~\cite{Shi_2005,365701,kitaev1995quantum,cleve2000fast,shor1999polynomial}.  Commuting
circuits of IQP type are the computations of this second level in which the stage between the two
layers is diagonal, and Buzet and Chailloux recently showed that constantly many of them, with
efficient classical pre- and post-processing, already decide a language outside the polynomial
hierarchy relative to an oracle~\cite{buzet2026iqp}.

Shi conjectured that the hierarchy is strict at every level~\cite[Conjecture~4.1]{Shi_2005} and
asked, as a first step, for oracle separations: whether each additional Hadamard layer strictly
increases computational power relative to some oracle~\cite{Shi_2005}.  Without an oracle the
conjecture is already out of reach at $k=0$, where it asserts $\P\ne\BPP$.  Simon's problem
separates the first two levels relative to an
oracle, and Shi proposed iterated Simon and recursive Fourier
sampling~\cite{bernstein1997quantum} as candidates for the higher levels.  To our knowledge, no
separation at any level $k\ge2$ has been proved since Shi raised the question.

\subsection{Main results}\label{subsec:main-results}

We resolve Shi's adjacent-level question at every level $k\ge2$, and we do so in both of the usual
ways of giving a circuit access to an oracle: the \emph{phase-query} model, in which a query
multiplies a basis state by a sign, and the \emph{standard-query} model, in which a query
writes its answer into a register (Definitions~\ref{def:FHk-rel} and~\ref{def:standard-fh}).  The
standard query is the oracle gate used in Shi's paper, and the phase query is the diagonal oracle of
the Fourier-growth literature.

All the separations below are witnessed by problems built from \emph{Forrelation}, which appeared in
Aaronson's Fourier Checking problem~\cite{aaronson2010bqp} and which Aaronson and
Ambainis~\cite{aaronson2014forrelationproblemoptimallyseparates} developed into the problem that
optimally separates quantum from classical query complexity.  Forrelation asks how strongly one Boolean function is correlated with
the Fourier transform of another, and its $K$-fold version asks the same of a chain of $K$ functions
alternating with Fourier transforms.  It is the natural problem here because that alternating chain
is exactly the pattern of Hadamard layers and oracle calls that the Fourier hierarchy counts.
The oracle functions are defined on $n$-bit strings, and we write $N=2^n$.

\begin{boxedstatement}{Theorem A (phase-query hierarchy; informal, Theorem~\ref{thm:all-level-separation})}
For every constant $k\ge 2$ there is an oracle $O$ such that
\[
\FH_k^{O}\ \subsetneq\ \FH_{k+1}^{O}
\]
in the phase-query Fourier hierarchy.  The separating problem is an OR of $(2k-1)$-fold
Forrelation instances.  It is decided at Fourier depth $k+1$ with $O_k(1)$ queries, a number
depending only on $k$, while every circuit of Fourier depth $k$ making at most $2^{n/(16k^2)}$
queries fails on some promised instance, whatever its size and whether or not it is uniform.
\end{boxedstatement}

The phase and standard access modes are equivalent in ordinary query complexity, but they behave
differently with respect to Fourier depth: a standard query leaves its answer in an ancilla, on which a
basis-preserving block may then branch, whereas a phase query only multiplies each branch by a sign
(Remark~\ref{rem:phase-vs-bit}).  The standard-query hierarchy therefore needs a lower bound of its own,
and that is our second main result.

\begin{boxedstatement}{Theorem B (standard-query hierarchy; informal, Theorem~\ref{thm:standard-query-adjacent})}
For every constant $k\ge 2$ there is an oracle $O$ such that
\[
\FH_{k,\mathrm{std}}^{O}\ \subsetneq\ \FH_{k+1,\mathrm{std}}^{O}
\]
in the standard-query hierarchy.  The separating problem is an OR of $(2k+1)$-fold
Forrelation instances: it is decided at standard Fourier depth $k+1$ with $O_k(1)$ queries, while
every standard-query circuit of Fourier depth $k$ making at most $N^{\eta_k}$ queries fails on some
promised instance, where $\eta_k=\tfrac1{2(2k-1)(2k+1)}$.
\end{boxedstatement}

The upper bounds are the same test in the two models: a control-qubit interference test comparing
the two halves of an odd-order Forrelation instance.  What differs is where the comparison is read
out.  A standard query writes the middle answer into a register, so the comparison becomes a
predicate in the computational basis; a phase query leaves it as a relative phase on the control
qubit, and one further Hadamard layer is needed to convert that phase into a measurable bit.  That
single layer is the gap Theorem~C measures.  It also fixes the two Forrelation orders.  Standard
depth $k$ already decides $(2k-1)$-fold Forrelation, the order that separates the phase hierarchy at
$k$, so the standard hierarchy has to be separated at the next odd order, which is $2k+1$.

The lower bounds differ as well.  For phase queries, a depth-$k$ circuit is
simulated exactly by a $(k-1)$-round
parallel-query algorithm, so the round-based Fourier growth bound of Girish, Sinha, Tal, and
Wu~\cite{girish2024power} rules out Fourier depth $k$.  For standard queries no comparably efficient
simulation is available, and we instead bound the Fourier growth directly, by factoring each
basis-preserving block through the leaves of its decision tree, where the entire queried path is
determined (Theorem~\ref{thm:standard-query-growth}).

Combining the phase-query lower bound with the standard-query upper bound gives a third separation.
At Fourier depth $k$, phase access decides $(2k-2)$-fold Forrelation but not $(2k-1)$-fold
Forrelation, whereas standard access decides the latter, since with a written answer bit the same
test needs only depth $k$.  The
two models therefore differ at a fixed depth, and not only in the level at which each hierarchy
separates.

\begin{boxedstatement}{Theorem C (phase versus standard queries; informal, Theorem~\ref{thm:phase-vs-standard})}
For every constant $k\ge 2$ there is an oracle $O$ with
$\FH_k^{O}\subsetneq\FH_{k,\mathrm{std}}^{O}$, separated by the OR of $(2k-1)$-fold Forrelation.
Standard queries are therefore strictly more powerful than phase queries at the same Fourier depth.
\end{boxedstatement}

Throughout, an oracle is an indexed family of Boolean functions, queried only at the input's own
length (Definition~\ref{def:FHk-rel}).  Proposition~\ref{rem:standard-convention-separations}
transfers Theorems~A, B and~C to a single language oracle under the usual convention.

Table~\ref{tab:fh-separations} summarizes the separations established in this paper.

\begin{table}[ht]
\centering
\footnotesize
\renewcommand{\arraystretch}{1.3}
\setlength{\tabcolsep}{4pt}
\resizebox{\textwidth}{!}{%
\begin{tabular}{@{}lllll@{}}
\hline
Levels & Separating problem & Solved at (queries) & Fails at (queries) & Reference \\
\hline
\multicolumn{5}{@{}l}{\emph{This paper}}\\
$\FH_k$ vs $\FH_{k+1}$ (phase), $k\ge2$ & OR of $(2k{-}1)$-fold Forrelation & $\FH_{k+1}$, $O_k(1)$  & $\FH_k$, $\le N^{c_k}$ & Thm~\ref{thm:all-level-separation}\\
\rowcolor[gray]{.9}$\FH_{k,\mathrm{std}}$ vs $\FH_{k+1,\mathrm{std}}$ (std), $k\ge2$ & OR of $(2k{+}1)$-fold Forrelation & $\FH_{k+1,\mathrm{std}}$, $O_k(1)$  & $\FH_{k,\mathrm{std}}$, $\le N^{\eta_k}$ & Thm~\ref{thm:standard-query-adjacent}\\
$\FH_k$ (phase) vs $\FH_{k,\mathrm{std}}$, $k\ge2$ & OR of $(2k{-}1)$-fold Forrelation & $\FH_{k,\mathrm{std}}$, $O_k(1)$  & $\FH_k$ (phase), $\le N^{c_k}$ & Thm~\ref{thm:phase-vs-standard}\\
\rowcolor[gray]{.9}$\bigcup_k\FH_k$ vs $\BQP$ & $K$-fold Forrelation, $K=\lceil\log_2 n\rceil$ & $\BQP$, $\poly(n)$ & every depth $\le\frac14\log_2 n$, $\poly(n)$ & Thm~\ref{thm:fh-vs-bqp}\\
\hline
\multicolumn{5}{@{}l}{\emph{Known, for comparison}}\\
$\FH_1$ vs $\FH_2$ (standard) & Simon's problem & $\FH_2$, $\poly(n)$  & $\FH_1=\BPP$ & \cite{365701,Shi_2005}\\
\rowcolor[gray]{.9}$\FH_2$ vs $\PH$ & Raz--Tal twofold Forrelation & $\FH_2$, $\poly(n)$ & $\PH$, quasipoly.\ size & \cite{buzet2026iqp}, Cor.~\ref{thm:fh2-not-in-ph}\\
\hline
\end{tabular}%
}
\caption[Fourier-depth separations in the Fourier hierarchy.]{Fourier-depth separations in the
Fourier hierarchy.  ``Solved at'' gives the Fourier depth and query count of the deciding circuit;
``Fails at'' gives the Fourier depth and the query budget up to which every circuit still fails on
some instance of the promise.  Here $c_k=\tfrac{1}{4(2k-1)(2k-2)}$ and
$\eta_k=\tfrac{1}{2(2k-1)(2k+1)}$, both at least $\tfrac1{16k^2}$, so the lower-level circuits fail
even with $2^{n/(16k^2)}$ queries; at the second phase-query level $c_2$ improves from $\tfrac1{24}$ to $\tfrac16$
(Remark~\ref{rem:fh2-sharpened}).  All rows are unconditional.  The explicit candidate ensemble of
Appendix~\ref{sec:appendix-sign}, whose $\FH_2$ lower bound is conditional, is not listed.}
\label{tab:fh-separations}
\end{table}

Theorem A relies on two structural facts about the hierarchy, which may be of independent
interest, together with Fourier-growth and distributional estimates from prior work.  The first fact is a simulation theorem:

\begin{quote}
\emph{Fourier depth bounds query adaptivity.}  Every $\FH_k$ circuit making $q$ phase queries is
simulated exactly by a quantum query algorithm with $k-1$ rounds of adaptivity and at most $q$
parallel queries per round (Lemma~\ref{lem:round-embedding}).  In particular, every $\FH_2$
computation is simulated exactly by a non-adaptive quantum query algorithm.
\end{quote}

The second fact concerns the upper bound.  Aaronson and
Ambainis~\cite{aaronson2014forrelationproblemoptimallyseparates} use a control-qubit interference
test to estimate the Forrelation value with $\lceil K/2\rceil$ queries; what we add is that at odd
order that test also has a fixed Fourier depth.

\begin{quote}
\emph{Odd-order Forrelation splits into two halves of equal depth.}  For odd $K=2k-1$, the
$K$-fold Forrelation expression splits at its middle oracle into two states each prepared by $k$
Hadamard layers, and a control-qubit interference test comparing them estimates the Forrelation
value using exactly one further Hadamard layer, the $(k+1)$st
(Proposition~\ref{prop:odd-interference-circuit}).
\end{quote}

The order has to be odd.  At even order one half of the expression is a layer deeper than the
other, so comparing them would give Fourier depth $k+2$ (Remark~\ref{rem:why-odd}).  The next lower
order, $2k-2$, is already decided at Fourier depth $k$ by the circuit of
Proposition~\ref{prop:even-forrelation}, which accepts on an inner-product parity of the two
measured registers.  It follows that $2k-1$ is the smallest Forrelation order for which this
construction separates depth $k$ from depth $k+1$, and it is the order at which the two facts combine.  Depth $k+1$ decides it, while depth $k$ is
confined to $k-1$ rounds, and $2(k-1)<2k-1$ is exactly the inequality under which the
Fourier-weight criterion of Bansal and Sinha~\cite{bansal2021k} rules those rounds out.  This is the
trade-off between rounds and parallel queries studied in~\cite{girish2024power}, reached here
through the $K$-fold Forrelation framework of~\cite{bansal2021k}.

The Fourier growth bounds and the analysis of the structured distributions are taken
from~\cite{girish2024power,bansal2021k}.  New here are the simulation of Fourier depth by query
adaptivity, the connection between odd Forrelation and the hierarchy through the control-qubit
interference test, the
decision-tree factorization that gives the corresponding growth bound for standard queries, and
the resulting answer to the oracle question in both models.  We also prove the imported growth bound
with a constant singly exponential in the number of rounds and the Fourier level
(Appendix~\ref{app:gstw-growth-constant}), which is what allows the Fourier depth to grow with the
input length in Corollary~\ref{cor:growing-k}; the standard-query growth bound has such a constant
from the outset, which gives the corresponding Corollary~\ref{cor:growing-k-std}.

\subsection{Technical overview}\label{subsec:technical-overview}

We prove Theorem~A (formally, Theorem~\ref{thm:all-level-separation}) by exhibiting a single problem
that is decided at Fourier depth $k+1$ and on which every circuit of Fourier depth $k$ fails.  Both bounds are
based on Forrelation, and they match because a single combinatorial parameter, its order, controls
both the depth of the upper-bound circuit and the Fourier-weight lower bound.  We describe the ideas informally here; the formal treatment is in
Section~\ref{sec:higher-level-conjecture}.

Oracles here are Boolean functions $F:\bits^n\to\pmone$ accessed by phase queries
$P_F\ket x=F(x)\ket x$, and an $\FH_k$ circuit has the form
\[
U_k\,(H^{\otimes m}\otimes I)\,U_{k-1}\cdots U_1\,(H^{\otimes m}\otimes I)\,U_0,
\]
each $U_j$ basis-preserving and possibly containing queries, followed by a computational-basis
measurement with a designated accepting set.  The lower bounds exploit this structure.

\subsubsection{The Forrelation problem}\label{subsubsec:overview-forrelation}

The \emph{Hadamard transform} of a function $f:\bits^n\to\mathbb R$ is
$\widetilde f(y)=2^{-n/2}\sum_{x}(-1)^{x\cdot y}f(x)$, which measures the correlation of $f$ with the
character $x\mapsto(-1)^{x\cdot y}$.  We use $\widetilde{\phantom{f}}$ for this normalization
throughout, and reserve $\widehat f$ for the ordinary Boolean Fourier coefficient, normalized as an
average rather than a sum (Definition~\ref{def:fourier-weight}); the two differ by a factor
$2^{n/2}$.  On a quantum computer this transform is exactly the $n$-qubit
Hadamard gate $H^{\otimes n}$ acting on amplitudes, which is how the Fourier transform enters the
Fourier hierarchy.

We now make Forrelation precise.  For $f,g:\bits^n\to\pmone$,
\[
\mathsf{forr}(f,g)=\frac{1}{2^{3n/2}}\sum_{x,y\in\bits^n}f(x)\,(-1)^{x\cdot y}\,g(y)\ \in\ [-1,1],
\]
that is, $\mathsf{forr}(f,g)=N^{-1}\langle g,\widetilde f\rangle$ for the Hadamard transform $\widetilde f$ above.  When $\mathsf{forr}(f,g)$ is
close to $\pm1$ the two functions are said to be ``forrelated''; when it is close to $0$ they are
weakly correlated.  A quantum
algorithm can estimate this quantity with one query to each function: prepare
$\ket u=H^{\otimes n}\ket{0^n}$, apply $P_f$, apply $H^{\otimes n}$, apply $P_g$, apply
$H^{\otimes n}$, and measure; the amplitude on $\ket{0^n}$ is exactly $\mathsf{forr}(f,g)$, so the outcome $0^n$
appears with probability $\mathsf{forr}(f,g)^2$, and the control-qubit interference test of
Aaronson and Ambainis estimates $\mathsf{forr}(f,g)$ itself.  A
classical algorithm, on the other hand, requires $2^{\Omega(n)}$ queries even to distinguish
$\abs{\mathsf{forr}}\approx1$ from $\mathsf{forr}\approx0$.

The generalization we need alternates $K$ oracles with $K+1$ Hadamard layers,
\[
\mathsf{forr}_K(F_0,\dots,F_{K-1})
=\bra{0^n}H^{\otimes n}P_{F_{K-1}}H^{\otimes n}\cdots
P_{F_1}H^{\otimes n}P_{F_0}H^{\otimes n}\ket{0^n}\ \in\ [-1,1].
\]
This \emph{$K$-fold Forrelation} is the amplitude of a computation that alternately applies an
oracle and takes a Fourier transform; it measures the correlation of $F_0,\dots,F_{K-1}$ under
\emph{repeated} Fourier transforms.  We use it for two reasons.  First, since it consists of
alternating Hadamard layers and oracle applications, its Fourier depth follows directly.
Second, Bansal and Sinha~\cite{bansal2021k}, building on Raz and
Tal~\cite{raz2019oracle}, showed that no bounded function of low \emph{Fourier weight} can
distinguish the structured Forrelation distribution from the uniform distribution.  As we will show
that low Fourier depth
implies low Fourier weight, this is exactly the form of hardness that we need.

\subsubsection{Fourier depth as query adaptivity}\label{subsubsec:overview-simulation}

The lower bound uses the following structural property of the hierarchy: an $\FH_k$ circuit may
create superposition only at its $k$ Hadamard layers.  Consider one \emph{basis-preserving block}
between two consecutive layers.  Every gate in it is basis-preserving, so a basis
state entering the block remains a single basis state throughout, and a phase query only
multiplies it by a sign.  Thus within a block each computational-basis branch follows a
deterministic sequence of basis states and accumulates only oracle-dependent phases.  This has three
consequences, which together allow a block to be replaced by one round of parallel queries.
\begin{enumerate}[leftmargin=2em]
\item \emph{The query points are determined in advance.}  On the branch entering the block in the
state $\ket z$, the register passes through a fixed sequence of basis states, since a phase query
does not change the register.  Hence the address queried at each step is a function of $z$ alone, and
can be computed without the oracle.
\item \emph{The block acts by a product of signs.}  Its effect on the branch $\ket z$ is to
multiply it by the product of the oracle values at these precomputed addresses, and then to apply a
fixed, oracle-independent permutation and phase.
\item \emph{One parallel batch suffices.}  A bounded-round query algorithm can therefore reproduce
the block in a single round: compute all the addresses of the branch $z$ into fresh registers,
query them in one parallel batch, uncompute the addresses, and apply the fixed permutation.
\end{enumerate}
Queries before the first Hadamard layer act on the single state $\ket{0^n}$ and contribute only a
global phase, and queries after the last layer commute with the measurement.  The queries that matter
therefore lie in the $k-1$ interior blocks, each of which is replaced by one round.  This is the
simulation lemma of Section~\ref{subsec:main-results} (Lemma~\ref{lem:round-embedding}): the $k-1$
interior blocks become $k-1$ \emph{rounds} of adaptivity, each of at most $q$ parallel queries, and
the acceptance probability is reproduced exactly.

Two features of this simulation will be relevant later.  First, it holds in one direction only: an
arbitrary unitary between two rounds is not a single Hadamard layer, so a bounded-round algorithm
need not be an $\FH_k$ circuit (Remark~\ref{rem:strict-embedding}).  Second, it is specific to phase
queries.  A standard query records $f(x)$ in an ancilla that remains in superposition, so
a block can choose its later queries according to earlier answers, and the argument above no longer
reproduces a block by a single round of comparably many parallel queries
(Remark~\ref{rem:erasure-obstruction}; the simulation of
Proposition~\ref{prop:three-round-bit-embedding} uses exponentially many parallel queries).  The
standard-query lower bound therefore requires a different argument, which we describe in
Section~\ref{subsubsec:overview-standard}.

For the lower bound we use the Fourier growth theorem of~\cite{girish2024power}, which bounds the
Fourier weight of a bounded-round algorithm at every level.  Writing $L_{1,\ell}$ for the sum of the
absolute values of the level-$\ell$ Fourier coefficients, an $r$-round, $t$-parallel-query algorithm
on $M$ oracle bits satisfies
\[
L_{1,\ell}\ \lesssim\ t^{\ell}\,M^{\frac{\ell}{2}\left(1-\frac{1}{2r}\right)},
\]
and the same bound continues to hold after fixing any subset of the input bits.  Taking $r=k-1$, the
simulation turns this into a statement about the entire class: every $\FH_k$ circuit has level-$\ell$
Fourier weight at most $q^{\ell}M^{\frac{\ell}{2}(1-\frac{1}{2k-2})}$, up to a constant depending
on $k$ and $\ell$ (Corollary~\ref{cor:FHk-growth}).

\subsubsection{Odd Forrelation and the interference test}\label{subsubsec:overview-odd}

We now compare this weight bound with Forrelation.  The result of~\cite{bansal2021k} has two parts.
First, under the uniform distribution $\abs{\mathsf{forr}_K}$ is rarely large, while under their
\emph{structured distribution} $\mathcal F_K$ we have $\mathsf{forr}_K\ge\delta$ with probability
$\Omega(\delta)$, where $\delta=2^{-5K}$.  Second, \emph{any} bounded function $f$ distinguishes
$\mathcal F_K$ from the uniform distribution $\mathcal U$ only in proportion to its Fourier weight,
\[
\bigl|\E_{\mathcal F_K}[f]-\E_{\mathcal U}[f]\bigr|
\ \le\
\sum_{\ell=K}^{K(K-1)}
N^{-\frac{\ell}{2}\left(1-\frac1K\right)}\,(8K)^{14\ell}\,
\sup_\mu \mathsf{wt}^{\mu}_{\ell}(f),
\]
where $\mathsf{wt}^{\mu}_{\ell}$ is the level-$\ell$ weight measured under the product measure with
biases $\mu$, and the supremum is over all bias vectors with entries in $[-\tfrac12,\tfrac12]$.  We
call this bound the
\emph{Bansal--Sinha criterion}.

The lower bound now reduces to a comparison of exponents.  The instances used below have
$M=O_k(N)$ oracle bits in all, so both sides can be read as powers of $N$.  Fourier depth $k$ gives
Fourier weight growing like $N^{\frac{\ell}{2}(1-\frac{1}{2k-2})}$, while the criterion bites as long
as the weight stays below $N^{\frac{\ell}{2}(1-\frac1{K})}$.  The first exponent is the smaller
exactly when $2k-2<K$, that is, when twice the number of rounds available at depth $k$ is less than
the Forrelation order; the smallest odd $K$ for which this holds is $K=2k-1$.  At that order the two
exponents differ by $\ell/(2(2k-1)(2k-2))$ at every level $\ell$, so a depth-$k$ circuit cannot
distinguish the two distributions in the stated range of query counts.  As we explain next, odd
order is also what the upper bound requires.

We now turn to the upper bound at depth $k+1$.  The $K$-fold Forrelation expression of
Section~\ref{subsubsec:overview-forrelation} consists of $K+1=2k$ Hadamard layers, so preparing its
amplitude directly would use $2k$ layers, and we are allowed $k+1$.  Instead we \emph{split the expression at
its middle oracle} $F_{k-1}$.  The left half prepares
\[
\ket{\psi_L}=H^{\otimes n}P_{F_{k-2}}\cdots P_{F_0}H^{\otimes n}\ket{0^n}
\]
with exactly $k$ Hadamard layers, the right half prepares an analogous $\ket{\psi_R}$ with $k$ layers,
and a direct expansion gives
\[
\mathsf{forr}_K=\bra{\psi_L}P_{F_{k-1}}\ket{\psi_R}.
\]
Since $K$ is odd, the two halves have \emph{equal} depth $k$; at even order one half would be a
layer deeper.  We prepare them on the two branches of a control qubit, so that they share their $k$
Hadamard layers.  The oracle of Definition~\ref{def:FHk-rel} is indexed, so a single phase query can
choose which function it applies according to that control qubit, and the distinguished index
carries the constant-zero function, so the middle query can be made to act on one branch and not the
other.  A single further Hadamard, on the control qubit, then converts the resulting relative phase
into an acceptance probability of exactly $\tfrac12(1+\mathsf{forr}_K)$.  The circuit has Fourier
depth $k+1$ and makes $k$ queries (Proposition~\ref{prop:odd-interference-circuit}); a SWAP test
comparing the two halves also has depth $k+1$, but estimates $\mathsf{forr}_K^2$ and uses $2k-1$
queries (Proposition~\ref{prop:odd-swap-circuit}).  Thus $(2k-1)$-fold
Forrelation is decided at Fourier depth $k+1$, while by the simulation lemma it is subject to the
round-based lower bound at $k-1$ rounds.  The same Forrelation instance is used in both bounds.

\subsubsection{From a distributional gap to an oracle}\label{subsubsec:overview-oracle}

The two bounds so far compare \emph{distributions}, namely $\mathcal F_K$ against the uniform
distribution, whereas an oracle separation requires a \emph{worst-case} promise problem that depth
$k+1$ decides and depth $k$ does not.  We use two standard steps to pass from the distributional
statement to a worst-case one.

The first step is amplification by an OR.  The structured distribution gives a large Forrelation
value only with probability $\Omega(\delta)$, which is too small for a worst-case statement, so we
take an OR of $m=\delta^{-1}$ independent copies: a \emph{yes} instance has some copy with
$\mathsf{forr}_K\ge\delta$, and a \emph{no} instance has $\abs{\mathsf{forr}_K}\le\delta/2$ for every
copy.  The two product distributions are then supported almost entirely on the two sides of the
promise, and depth $k+1$ still decides every instance, by running the interference tests for the $m$ copies
in parallel and thresholding.

The second step is a hybrid argument followed by a diagonalization.  A depth-$k$ circuit that solved
the promise on every instance would distinguish the two product distributions with constant
advantage, and a hybrid argument then isolates a single copy accounting for that advantage, with the
remaining copies \emph{fixed}.  Fixing the other copies is permissible because the growth bound is
stable under restrictions, and the gap between the exponents gives a contradiction
(Theorem~\ref{thm:lower-worst-case}).  A standard diagonalization over $\FH_k$ machines then
combines the hard instances at the various input lengths into a single oracle $O$, together with a
unary language in $\FH_{k+1}^{O}\setminus\FH_k^{O}$ (Theorem~\ref{thm:all-level-separation}).  For
$k=2$ this says, via the simulation, that the third Hadamard layer adds one round of
adaptivity (Corollary~\ref{cor:fh2-fh3-separation}).  The argument bounds the advantage of
\emph{every} depth-$k$ circuit, not only that of a hypothetical decider: in that range of query
counts depth $k$ has distinguishing advantage $o_k(1)$, while depth $k+1$ decides with error
$2^{-n}$ (Corollary~\ref{cor:error-robust}).

\subsubsection{The standard-query separation}\label{subsubsec:overview-standard}

Theorem~B concerns the standard-query model, in which a query writes its answer bit,
$O_f\ket y\ket b=\ket y\ket{b\oplus f(y)}$.  The upper bound carries over: the same interference
test, run at standard depth $k+1$, decides $(2k+1)$-fold Forrelation, because the middle answer is
read from a register rather than from a phase.  The lower bound is where the
two models differ.  A phase query multiplies a basis state by a sign and does not change the
register, which is what allowed us to replace a block by a single parallel round.  A standard query
instead writes $f(y)$ into an ancilla, and later gates may read it and branch on it, so a block is
a coherent \emph{decision tree} whose query locations depend on earlier answers.  A single round of
parallel queries, followed by a fixed unitary, no longer reproduces such a block, so long as the
round is not much wider than the block itself.  The qualification matters: a round of $M$ queries
reads the whole truth table and can then reproduce any block, so what fails is a simulation whose
width is comparable to the block's own query count.  The round simulation is therefore unavailable
in the form used above.

We instead bound the Fourier growth of a standard-query circuit directly, by factoring each block
through the \emph{leaves} of its decision tree.  Conditioned on a leaf, the output basis map and the
phase of the block are fixed, and whether that leaf is reached is a product of affine literals in the
answers along its path.  Summing the leaves back together would destroy the norm control that the
argument needs, since different branches query different variables.  We therefore record the leaf in
an auxiliary register instead, which factors the block into an isometry followed by a contraction and
forces the two neighbouring factors of the operator product to agree on the branch.  Counting how
many distinct oracle bits a level-$\ell$ character can involve across that product then gives a
restriction-stable growth bound
(Theorem~\ref{thm:standard-query-growth}): the acceptance probability $f$ of a standard-query
$\FH_k$ circuit making $q$ queries satisfies, for every restriction $f_\rho$ of it,
\[
L_{1,\ell}(f_\rho)\ \le\ (q+1)^{\ell}\,M^{\frac{\ell}{2}\left(1-\frac1{2k-1}\right)},
\]
up to a factor $2^{O(k\ell)}$.  This exponent is weaker than the phase-query exponent
$\tfrac\ell2(1-\tfrac1{2k-2})$, and the reason is a count of factors.  The operator product here has
$2k-1$ of them: the $k-1$ interior blocks, their adjoints, and the acceptance operator of the final
block.  The round simulation produced only $2k-2$, and each further factor adds one to the
denominator.  We therefore run the argument at the larger order $K=2k+1$ in place of $2k-1$, since
the Bansal--Sinha criterion needs the growth exponent to fall below $\tfrac\ell2(1-\tfrac1K)$, and
$1-\tfrac1{2k-1}<1-\tfrac1K$ holds exactly when $K>2k-1$.  Given this bound the remainder of the argument is
as before: an OR amplifies the distributional gap, a hybrid step reduces to a single instance, and a
diagonalization gives the oracle of Theorem~B (Theorem~\ref{thm:standard-query-adjacent}).

\subsubsection{Further results and connections}\label{subsubsec:overview-further}

\paragraph{Bounded Fourier depth versus $\BQP$.}
A single problem separates the entire hierarchy from unrestricted quantum computation.  Forrelation
with order growing as $\Theta(\log n)$ is decided in $\BQP$ with polynomially many
queries, but at no Fourier depth up to $\tfrac14\log_2 n$, so
$(\bigcup_k\FH_k)^{O}\subsetneq\BQP^{O}$ relative to an oracle
(Theorem~\ref{thm:fh-vs-bqp}).  Together with the level separations this gives
$\FH_2^{O}\subsetneq\cdots\subsetneq(\bigcup_k\FH_k)^{O}\subsetneq\BQP^{O}$, and by tuning
the order one obtains a $\BQP$ problem of any prescribed Fourier depth up to $c\log n$
(Corollary~\ref{cor:prescribed-fdepth}).  Both statements hold in the standard-query model as well:
the growing-depth level separation is Corollary~\ref{cor:growing-k-std}, and the separation from
$\BQP$ is Remark~\ref{rem:tower}.  The proof is the round-based lower bound applied at a
growing Forrelation order, using the form of the growth theorem proved in
Appendix~\ref{app:gstw-growth-constant}, whose constant is singly exponential in the
number of rounds; the ceiling $c\log n$ comes from the deciding circuit, whose $2^{\Theta(k)}$
repetitions must remain polynomial in size.  In the
opposite direction, relative to a suitably encoded $\mathrm{PSPACE}$-complete oracle every level
from the second on collapses to $\mathrm{PSPACE}$, in both query models
(Proposition~\ref{prop:pspace-collapse}), so the adjacent-level question does not relativize
(Corollary~\ref{cor:nonrelativizing}): resolving Shi's conjecture without an oracle requires
non-relativizing techniques.

\paragraph{The hierarchy and the polynomial hierarchy.}
The separations above are internal to the hierarchy.  A result of Buzet and
Chailloux~\cite{buzet2026iqp} locates it against a classical class, and we record it as
Corollary~\ref{thm:fh2-not-in-ph} without claiming it as new.  In Shi's standard-query model the
first level is exactly $\BPP$ relative to every oracle (Lemma~\ref{lem:fh1-std}), hence always
inside $\PH$, whereas their IQP computation for twofold Forrelation puts the second level outside
$\PH$ relative to an oracle.  One further Hadamard layer therefore takes the hierarchy outside the
polynomial hierarchy, and by the simulation lemma the algorithm that achieves this is non-adaptive:
one batch of parallel quantum queries, followed by an oracle-independent unitary and a measurement,
suffices.  Remark~\ref{rem:fh2-ph-alternative} records an alternative route at a weaker gap, which
uses only the depth-$k$ circuit for even-order Forrelation of
Proposition~\ref{prop:even-forrelation}: on the Raz--Tal distribution~\cite{raz2019oracle} its
acceptance probability has a mean gap of $\Omega(1/n^2)$, far above the
$\mathrm{polylog}(N)/\sqrt N$ advantage to which they bound every quasipolynomial-size
constant-depth circuit.

\paragraph{On the nature of the oracle.}
The hard instances above are the correlated-Gaussian structured distributions
of~\cite{raz2019oracle,bansal2021k}, which are designed for the Fourier-weight analysis: their
moments are not available in closed form, and they are not known to be samplable together with a
certificate of the promise.  Since the adjacent-level question does not relativize
(Corollary~\ref{cor:nonrelativizing}), no oracle construction will settle Shi's conjecture, and an
explicitly defined hard distribution is a first step toward instances that could be evaluated
without an oracle.
Appendix~\ref{sec:appendix-sign} proposes a candidate at the second level: a threefold Forrelation
instance whose middle oracle is the product of the sign patterns of the Hadamard transforms of two
random oracles,
so that its moments factor exactly into Fourier coefficients of one explicit function.  Its analysis reduces to a single quantity: one Fourier
coefficient of a product of threshold functions of pairwise orthogonal linear forms, each of them
regular to the greatest possible degree, every coefficient having the same modulus, $N^{-1/2}$ of
the $\ell_2$ norm.  We
prove the conjectured rate of decay for that coefficient, with a weaker prefactor, in one of the two
parameter ranges (Theorem~\ref{thm:revealed-decay}), and conjecture it in the other
(Conjecture~\ref{conj:sharp-decay}); the conjecture would make the ensemble an explicitly defined hard distribution for the
oracle separation $\FH_2\subsetneq\FH_3$.  This is a program rather than a theorem: the class separations proved in the
main text are unconditional and do not depend on it.

\subsection{Related work}\label{sec:related-work}

\paragraph{The Fourier hierarchy.}
Shi introduced the hierarchy as part of a program to measure a computation by the number of
basis-changing gates it uses~\cite{Shi_2005,shi2002toffolicontrollednotneedlittle}.  The paper that
defines it~\cite{Shi_2005} also proves a tradeoff $T\cdot\ell=\Omega(n)$ between queries and layers
of basis-changing gates for Grover search, states the hierarchy conjecture
$\FH_k\subsetneq\FH_{k+1}$, and poses the oracle question, noting that Simon's
problem~\cite{365701} separates the first two levels and proposing iterated Simon and recursive
Fourier sampling~\cite{bernstein1997quantum} as candidates for the higher levels.

That $\FH_2$ already contains factoring follows from Kitaev's phase
estimation~\cite{kitaev1995quantum}, an observation Shi attributes to Mosca.  The parallelized
approximate quantum Fourier transform of Cleve and Watrous~\cite{cleve2000fast} is the standard
route to implementations of Shor's algorithm~\cite{shor1999polynomial,365700} with few Hadamard
layers.  Subsequent work on the low levels includes verification of circuits with few basis
changes~\cite{Demarie_2018} and Merlin--Arthur protocols and hardness of simulation statements for
the second level~\cite{Morimae_2018}; related restricted models include commuting circuits of IQP
type~\cite{shepherd2009temporally,bremner2011classical,shepherd2010quantum,bremner2016average}
and quantum Fourier sampling~\cite{fefferman2015powerquantumfouriersampling}.

Buzet
and Chailloux~\cite{buzet2026iqp} recently showed that the commuting model decides $2$-Forrelation, so that
$(\BPP^{\mathrm{IQP}})^{O}$ is not contained in $\PH^{O}$, thereby answering a question of
Girish~\cite{girish2025fourierspectrumnoisyquantum}; they also prove Fourier growth bounds for IQP
circuits in terms of the size of the accepting set.  Since an IQP computation is an $\FH_2$ circuit,
their separation gives Corollary~\ref{thm:fh2-not-in-ph}, which we state for the hierarchy and do not
claim as new.  All of these works concern the power of the low levels, and we are not aware of any
prior separation between adjacent levels beyond $\FH_1$ versus $\FH_2$.

This hierarchy should not be confused with the higher-order Fourier analysis of the Clifford hierarchy of Bu, Gu, and
Jaffe~\cite{bu2025quantumhigherorderfourier}, which classifies gates by the degree of their phase
polynomials rather than a computation by its number of Hadamard layers.

\paragraph{Rounds of adaptivity and parallel queries.}
Our separation theorem builds on the study of depth versus query trade-offs in the query model.
Girish, Sinha, Tal, and Wu~\cite{girish2024power} proved that $2r$-fold Forrelation can be solved by
$r$ rounds with one query per round, while any algorithm with $r-1$ rounds requires
$\widetilde\Omega(N^{1/r^2})$ parallel queries.  Their main technical component is a bound on the level-$\ell$ Fourier weights of algorithms with $r$
rounds of $t$ parallel queries, and it continues to hold after fixing input bits; we use it as
Theorem~\ref{thm:gstw-growth}.
Appendix~\ref{app:gstw-growth-constant} evaluates the assignment count in that proof in closed form
by M\"obius inversion, which gives the theorem with a constant singly exponential in $r\ell$ and
uniform over the Fourier levels.
For the standard-query setting we instead prove a Fourier growth theorem for coherent classical
decision trees (Section~\ref{app:standard-growth}).  The closest precedent is the work of Girish,
Tal, and Wu~\cite{girish2021parity}, who bound the level-$\ell$ Fourier growth of \emph{classical}
parity decision trees and, using the same Bansal--Sinha criterion, obtain the randomized parity-tree
lower bound $\widetilde\Omega(N^{1-1/k})$ for $k$-fold Forrelation.  Our basis-preserving blocks are
coherent rather
than classical, and we bound their growth by the operator-norm argument of~\cite{girish2024power}
rather than by their random-walk approach.  Carolan, Gilani, and Vempati~\cite{carolan2025parallel}
study parallel query complexity for total functions, exhibiting unbounded parallel quantum advantages
and a method for proving parallel quantum lower bounds.
Lemma~\ref{lem:round-embedding} relates this line of work to the Fourier hierarchy: Fourier depth $k$
is simulated exactly by $k-1$ rounds, so lower bounds for rounds give lower bounds for depth in the
phase-query model.

\paragraph{Forrelation and its lower bounds.}
Twofold Forrelation appeared first in Aaronson's Fourier Checking problem~\cite{aaronson2010bqp},
and was developed and optimized by Aaronson and Ambainis as the problem that optimally separates
quantum from classical query complexity~\cite{aaronson2014forrelationproblemoptimallyseparates};
its $k$-fold version is a standard test problem for such separations.  Raz and Tal's oracle
separation of $\BQP$ and $\PH$~\cite{raz2019oracle} introduced the Fourier weight method against
structured distributions, and Tal~\cite{tal2020towards} developed it further.  Sherstov,
Storozhenko, and Wu~\cite{sherstov2023optimal} proved the essentially tight level-$\ell$ bound
$c^{\ell}\sqrt{\binom d\ell(1+\log N)^{\ell-1}}$ for depth-$d$ decision trees, settling a conjecture
of Tal, and obtained the optimal quantum--classical query separation for a Gaussian variant of
$k$-fold Forrelation.  Bansal and Sinha~\cite{bansal2021k} proved the optimal randomized lower bound
$\widetilde\Omega(N^{1-1/k})$ for $k$-Forrelation itself, for every $k\ge2$, odd or even, using the
criterion that we import as Theorems~\ref{thm:bs-distributions}--\ref{thm:bs-transfer}.

On the classical simulation side, Aaronson and Ambainis~\cite{aaronson2014forrelationproblemoptimallyseparates}
stated an $O(N^{1-\frac1{2q}})$-query classical simulation for every $q$-query quantum algorithm;
Bravyi, Gosset, Grier, and Schaeffer~\cite{bravyi2021classical} subsequently repaired a gap in
the proof and established the general result.  Girish and
Servedio~\cite{girish2025forrelationextremallyhard} show that the \emph{extremal} Forrelation
problem, distinguishing $\mathsf{forr}=1$ from $\mathsf{forr}=-1$, requires
$\widetilde\Omega(2^{n/4})$ classical randomized queries, by a connection with bent functions over
$\Ftwo$ rather than by the analytic Gaussian rounding argument.  Their proof in fact bounds the
advantage of a $D$-query classical algorithm by $O(nD^2\,2^{-n/2})$, so the problem remains hard at
vanishing advantage; this is the regime in which we state our own separation
(Corollary~\ref{cor:error-robust}).  We use these lower bounds as black boxes: none of their
analyses is modified here, and what we add is a way to relate their parameters to Fourier depth,
through the decomposition of odd-order Forrelation.

\paragraph{Hierarchies of hybrid quantum depth.}
A distinct line of work separates hybrid models in which shallow quantum circuits are
interleaved with polynomial-time classical computation: Coudron and
Menda~\cite{coudron2020computations} and Chia, Chung, and Lai~\cite{chia2020need} proved that
quantum depth $d$ and $2d+1$ differ relative to oracles, with subsequent optimal $d$ versus
$d+1$ separations by Arora, Gheorghiu, and Singh~\cite{arora2022oracle} and by Hasegawa and Le
Gall~\cite{hasegawa2022optimal}.  As discussed in~\cite{girish2024power}, these models are
incomparable to bounded-round parallel-query algorithms: they restrict parallel queries but allow
measurement and adaptive classical control.  They are likewise incomparable to the Fourier
hierarchy, in which all computation between layers is coherent and basis-preserving.  The
two hierarchies measure different notions of quantum depth: layers of coherent interference in our
case, and alternations of quantum and classical processing in theirs.
Theorem~\ref{thm:fh-vs-bqp} is the analogue for Fourier depth of the separations between constant
and growing depth in those models: bounded Fourier depth is strictly weaker than $\BQP$
relative to an oracle.

\paragraph{Restricted quantum models and Fourier growth.}
A related line of work uses Fourier growth to separate restricted \emph{quantum} models from one
another, rather than quantum from classical.  Closest to our setting is
Girish~\cite{girish2025fourierspectrumnoisyquantum}, who proves Fourier-growth bounds for two noisy
models: $\mathrm{DQC}_k$, in which $k$ qubits are clean and the rest maximally mixed, and
$\tfrac12\BQP$, in which the input is maximally mixed but the algorithm is given the initial
state at the end of the computation.  That work shows that threefold Forrelation is solvable in
$\BQP$ but not in $\tfrac12\BQP$, and that twofold Forrelation is not in
$\mathrm{DQC}_1$, resolving conjectures of Jacobs and Mehraban~\cite{JM24}; models intermediate
between $\mathrm{DQC}_1$ and $\BQP$ were studied earlier by Aaronson, Bouland, Kuperberg, and
Mehraban~\cite{ABKM16}.  That threefold-Forrelation lower bound and ours are analogous: both use a
Forrelation problem to separate a restricted quantum model from the unrestricted one, and both go
through the Bansal--Sinha Fourier weight criterion.

The restricted resources are different: mixedness in those models, and in ours the Fourier depth,
which bounds the number of adaptive rounds.  These problems do not
separate $\FH_2$ from $\tfrac12\BQP$: the second Fourier level, which is non-adaptive by
Lemma~\ref{lem:round-embedding}, and $\tfrac12\BQP$ both decide twofold Forrelation and
neither decides threefold Forrelation, so no containment between them follows in either direction.
For $\mathrm{DQC}_1$ the comparison goes one way at the level of query complexity: the second
Fourier level decides twofold Forrelation, by the commuting circuits above, whereas Girish shows that
every $\mathrm{DQC}_1$ algorithm for it needs $\widetilde\Omega(N^{1/6})$ queries.  Whether $\mathrm{DQC}_k$ or $\tfrac12\BQP$ is in turn
contained in a fixed Fourier level, and more generally whether some explicit problem separates
mixedness from Fourier depth, remains open.

\paragraph{Pseudorandomness.}
The sign ensemble of Appendix~\ref{sec:appendix-sign} is proposed as a pseudorandom object against
distinguishers of bounded Fourier depth, analogous to the classical study of distributions that fool
tests of low degree, bounded-depth circuits, and branching
programs~\cite{nisan1994hardness,vadhan2012pseudorandomness,odonnell2021analysisbooleanfunctions}.
It is related, though not by definition, to pseudorandom quantum states and
unitaries~\cite{ji2018pseudorandom,krestchmer_pru,metger2024simpleconstructionslineardepthtdesigns,
ma2025constructrandomunitaries,ananth2024pseudorandomnessinverselesshaarrandom}: there the
distinguisher is an arbitrary efficient quantum algorithm, whereas here it is restricted to have
bounded Fourier depth, and the pseudorandom object is a distribution over classical oracles.

\subsection{Organization}\label{subsec:organization}

The remainder of the paper is organized as follows.  Section~\ref{sec:prelim} fixes the model, the
analytic tools, and the results of previous work that we use.
Section~\ref{sec:higher-level-conjecture} proves Theorem~A, and
Section~\ref{sec:standard-hierarchy} proves Theorems~B and~C.  Section~\ref{sec:landscape} locates
the hierarchy relative to the polynomial hierarchy and to $\BQP$, and shows that the adjacent-level
question does not relativize.  Section~\ref{sec:open-problems} presents the open problems.
Theorem~A is proved in Sections~\ref{subsec:odd-forrelation}--\ref{subsec:separation-theorem}, and
Theorems~B and~C in Sections~\ref{subsec:standard-leaf-factorization}--\ref{subsec:standard-separation},
which build on it; Sections~\ref{subsec:even-forrelation}, \ref{subsec:fh2-classical}
and~\ref{subsec:pointer-chasing} are not used in those proofs and can be skipped on a first reading.
Appendix~\ref{sec:appendix-sign} proposes an explicit candidate for the hard distribution at the
second level, Appendix~\ref{app:gstw-growth-constant} proves the Fourier growth bound that we
import with a single-exponential constant,
and Appendix~\ref{app:nonadaptive-proof} gives a self-contained proof of the non-adaptive
Fourier-growth bound used at the second level.

\section{Preliminaries}\label{sec:prelim}

Throughout, $\bits^n$ denotes the Boolean cube, identified with the vector space $\Ftwo^n$; for
$x,y \in \bits^n$ the dot product is $x \cdot y = \sum_{i=1}^n x_i y_i \pmod 2$.  We write
$N=2^n$.  All circuit families are uniform unless stated otherwise.  We use the computational
basis $\{\ket{x}:x \in \bits^n\}$ on $n$ qubits.  The $n$-qubit Hadamard transform is denoted
$H^{\otimes n}$ and satisfies
\begin{equation}\label{eq:hadamard-action}
H^{\otimes n}\ket{x}
=
2^{-n/2}\sum_{y \in \bits^n} (-1)^{x \cdot y}\ket{y}.
\end{equation}
Equation~\eqref{eq:hadamard-action} is the identity underlying all the Forrelation expressions in
this paper.

\subsection{The Fourier hierarchy}

\begin{definition}[Phase oracle]
If $F:\bits^n \to \pmone$, the corresponding phase oracle is the diagonal unitary
\begin{equation}\label{eq:phase-oracle-def}
P_F\ket{x}=F(x)\ket{x}.
\end{equation}
\end{definition}

\begin{remark}[Phase queries versus standard queries]\label{rem:phase-vs-bit}
For a Boolean function $f:\bits^n\to\bits$ with sign form $F=(-1)^f$, the standard query
oracle is the basis-preserving unitary $O_f\ket{x}\ket{b}=\ket{x}\ket{b\oplus f(x)}$.  Under the
indexed convention of Definition~\ref{def:FHk-rel}, where the distinguished index gives a phase
reference, the two access modes are equivalent in ordinary query complexity; a single unindexed
phase oracle is weaker, since $F$ and $-F$ then differ only by a global phase.  They are in any case
\emph{not} equivalent with respect to Fourier depth.  In one direction, a phase query can be
implemented by one standard query without increasing the depth, provided that at least one Hadamard
layer is available to prepare the target.  Prepare a single ancilla in $\ket{-}$ at the first global
Hadamard layer and use it as the target of every standard query, so that the answer is obtained by
phase kickback; route it off the designated wires at each later layer by the basis-preserving SWAP
gates of Remark~\ref{rem:register-routing}, so that it remains in $\ket{-}$ throughout.  Phase queries made before the first
layer act on a computational basis state and contribute only a global phase, and phase queries made
after the last layer commute with the measurement, so both can be deleted without changing the
acceptance probability (Step~1 in the proof of Lemma~\ref{lem:round-embedding}); the target is
therefore needed only inside the interior blocks.  Hence for every $k\ge1$, every circuit in the
phase-query $\FH_k$ is also a circuit of the same Fourier depth in the standard-query $\FH_k$.

At level $0$, which has no Hadamard layer, the two models differ: a phase query acts only by a global phase on a
computational-basis branch, so a phase-query $\FH_0$ computation obtains no information from $O$,
whereas a standard query writes its answer bit, so a standard-query $\FH_0$ computation is an
ordinary classical computation with oracle access to $O$.

In the other direction, converting phase
access back into a written bit requires interference on the target, and this increases the depth.
Moreover, with standard queries the register contents inside a basis-preserving block depend on the
oracle, so such a block can choose its later queries according to earlier answers, which a
phase-query block cannot (Lemma~\ref{lem:round-embedding}).
\end{remark}

\begin{definition}[Basis-preserving unitary]
A unitary $U$ is \emph{basis-preserving} if for every computational basis vector $\ket{z}$ there
is a basis vector $\ket{\pi(z)}$ and a phase $\lambda(z)$ with $\abs{\lambda(z)}=1$ such that
\begin{equation}\label{eq:basis-preserving-def}
U\ket{z}=\lambda(z)\ket{\pi(z)}.
\end{equation}
Equivalently, $U$ is a permutation matrix times a diagonal unitary in the computational basis.
\end{definition}

In other words, a basis-preserving unitary never creates superposition from a computational basis
state; it only permutes basis states and multiplies them by phases.  This is the property that we use
throughout to analyze circuits of low Fourier depth.

\begin{definition}[Relativized Fourier hierarchy]\label{def:FHk-rel}
Let $O=\{O_n\}_{n\ge 1}$ be an oracle family, where each $O_n=(f_{n,a})_{a\in A_n}$ is a tuple of
Boolean functions $f_{n,a}:\bits^n\to\bits$ indexed by a fixed set $A_n=\bits^{s(n)}$ with
$s(n)=O(\log n)$.  A distinguished index $0^{s(n)}\in A_n$ carries the constant-zero function, and
any index not assigned a hard function in a construction is likewise the constant-zero function.  The
functions are exposed through a single indexed oracle: the phase query is
$P_{O_n}\ket{a,y}=(-1)^{f_{n,a}(y)}\ket{a,y}$ and the standard query is
$O_n\ket{a,y}\ket b=\ket{a,y}\ket{b\oplus f_{n,a}(y)}$, so a query may select the function index $a$
coherently.  The phase and standard hierarchies use the same underlying functions, with sign form
$F_{n,a}=(-1)^{f_{n,a}}$. For $k \ge 0$, the class $\FH_k^O$ consists of promise problems
solvable with bounded error by polynomial-size uniform quantum circuits with phase-query access
to $O$ and at most $k$ global Hadamard layers, with only basis-preserving computation between
consecutive Hadamard layers.  Access is length-preserving: on an input of length $n$ the circuit
queries only the length-$n$ oracle $O_n$.  The standard-query class $\FH_{k,\mathrm{std}}^O$ of
Definition~\ref{def:standard-fh} is defined under the same convention.
\end{definition}

This indexed, length-preserving convention is the one in which the separations are proved.  It is
not the only convention in use, and Proposition~\ref{rem:standard-convention-separations} shows that
the separations also hold for a single language oracle queried at polynomially related lengths.  A
query to the distinguished index does nothing, since its phase is $+1$; several constructions below
use it to pad a batch of parallel queries, or to make a query conditional on a control bit
(Remark~\ref{rem:phase-vs-bit}, Definition~\ref{def:round-model}).

Throughout, bounded error means worst-case error at most $\tfrac13$.  Since repeated runs can share
the same Hadamard layers, any fixed error below $\tfrac12$ is reduced below $\tfrac13$ without
increasing the Fourier depth.

\begin{definition}[$\FH_k$ circuits and acceptance]\label{def:FH2FH3}
Fix an oracle family. A circuit has \emph{exact Fourier depth} $j$ if it has the form
\[
U_j\,(H^{\otimes m}\otimes I)\,U_{j-1}\cdots U_1\,(H^{\otimes m}\otimes I)\,U_0
\]
with $j$ global Hadamard layers.  A uniform $\FH_k$ circuit on $m$ Fourier wires (and polynomially many
ancilla wires) is a circuit of exact Fourier depth $j$ for some $0\le j\le k$; we write the depth-$k$
case as
\begin{equation}\label{eq:FHk-form}
U_k\,(H^{\otimes m}\otimes I)\,U_{k-1}\,(H^{\otimes m}\otimes I)\cdots
U_1\,(H^{\otimes m}\otimes I)\,U_0,
\end{equation}
where each block $U_i$ is a polynomial-size circuit over a fixed finite set of basis-preserving
gates: the classical reversible gates $\textsc{not}$, $\textsc{cnot}$ and Toffoli, the
$\textsc{swap}$ and controlled-$\textsc{swap}$ gates, and a fixed finite family of diagonal phase
gates with algebraic entries.  A block may in addition prepare ancillas in computational-basis states
and make phase queries to the oracle.  Every one
of the $k$ global Hadamard layers acts on the same $m$ designated Fourier wires
$(H^{\otimes m}\otimes I)$, the identity acting on all ancilla wires; logical registers are
routed onto these wires by the basis-preserving SWAP gates of Remark~\ref{rem:register-routing}.
At the end of the circuit, all wires are measured in the computational basis, and the circuit
accepts if the outcome lies in a designated set $S$ recognizable in deterministic polynomial time.
We write $\Pi_S$ for the projector onto $\operatorname{span}\{\ket{s}:s\in S\}$.
\end{definition}

Accepting on the value of one designated output qubit is the special case in which $S$ consists of
all strings with that qubit equal to $1$.  Computing a Boolean predicate of the outcome by a
reversible circuit before measurement is also of this form, since such a circuit is
basis-preserving and can be absorbed into $U_k$.  A general acceptance set is the standard
convention, and we need it for the threshold tests over repeated runs that appear below.

In the concluding discussion of~\cite{Shi_2005}, Shi introduces $\FH_k$ as the class of languages
decided with bounded error by a polynomial-size quantum circuit using at most $k$ Fourier
transforms, where the circuit is written over Toffoli and Hadamard gates and a Fourier transform is
a composition of Hadamard gates; he observes there that $\FH_0=\P$ and $\FH_1=\BPP$.  The class is
introduced in prose rather than as a numbered definition, and no oracle model is fixed for it.  Definition~\ref{def:FH2FH3} makes this concrete in
two respects.  Each Fourier transform becomes a layer $(H^{\otimes m}\otimes I)$ acting on a fixed
set of designated wires, so that the computation between two consecutive layers is exactly a
basis-preserving block; and the gates available inside a block are enlarged from the classical
reversible gates to include the fixed diagonal phase gates.  The lower bounds
below are proved for this larger model.  The upper-bound circuits of
Sections~\ref{sec:higher-level-conjecture} and~\ref{sec:standard-hierarchy} use only classical
reversible gates, SWAP gates and oracle gates between their Hadamard layers, so they are circuits
in Shi's model as well; the compiled commuting circuit of Corollary~\ref{thm:fh2-not-in-ph} uses
diagonal gates.

We fix a finite gate set for two reasons.  First, the uniform $\FH_k$ families are then
effectively enumerable as clocked generators with explicit polynomial query bounds, and the
diagonalizations below need that enumeration; see the proof of Theorem~\ref{thm:all-level-separation}.
Second, since the entries are algebraic, the phase attached to a basis state along any branch is a
product of finitely many fixed algebraic phases and can be computed to any polynomial precision in
polynomial time, which is all that the lower-bound analysis and the path-sum argument of
Proposition~\ref{prop:pspace-collapse} use.

The cases $k=2$ and $k=3$ occur most often below; an $\FH_2$ circuit has the form
$U_2 (H^{\otimes m}\otimes I) U_1 (H^{\otimes m}\otimes I) U_0$.

\begin{definition}[Standard-query Fourier circuits]\label{def:standard-fh}
A \emph{standard-query $\FH_k$ circuit} is an $\FH_k$ circuit (Definition~\ref{def:FH2FH3}) in which the
oracle gates inside the basis-preserving blocks are the indexed standard query
\[
O_n\ket{a,y}\ket b=\ket{a,y}\ket{b\oplus f_{n,a}(y)}
\]
of Definition~\ref{def:FHk-rel}, rather than phase gates.  The query count is the number of such standard gates.  For an oracle
family $O$, the corresponding polynomial-query class, with \emph{at most} $k$ global Hadamard
layers, is denoted $\FH_{k,\mathrm{std}}^O$.  This notation keeps it distinct from the
phase-query class $\FH_k^O$ of Definition~\ref{def:FHk-rel}.
\end{definition}

In the Fourier-growth literature the diagonal oracle $P_F$ is itself often called the \emph{standard}
quantum query oracle, as in Girish, Sinha, Tal and Wu~\cite{girish2024power}.  We follow Shi
instead: throughout this paper ``standard query'' means the oracle that writes its answer into a
register, and never the diagonal one.  Shi fixes no oracle model for $\FH_k$ separately; the oracle
gate used in~\cite{Shi_2005} is $O_x\ket i\ket b=\ket i\ket{b\oplus x_i}$, which is the query of
Definition~\ref{def:standard-fh}.  The distinction matters, since the two models have different Fourier depths for the
same problem (Theorem~\ref{thm:phase-vs-standard}).

\begin{remark}[Register routing for global Hadamard layers]\label{rem:register-routing}
In the circuits of this paper we sometimes describe a global Hadamard layer as acting on
whichever logical registers should be transformed at that stage.  This is an abbreviation for
the following implementation on a fixed set of wires.  The layer always acts on the same $m$
designated wires;
before the layer, SWAP gates exchange the logical registers that should be transformed with the
contents of the designated wires, and exchange registers that should not be transformed with
ancilla registers holding computational basis states.  We use a fresh computational-basis placeholder
at each layer.  A placeholder that has already passed through a Hadamard is retained as an
oracle-independent product-state ancilla rather than reused as a basis state; this costs only
polynomially many additional wires.  Since SWAP gates, controlled SWAP gates,
permutations, classical reversible gates, and phase queries are all basis-preserving, this
convention does not change the Fourier depth.
\end{remark}

The only feature of Definition~\ref{def:FH2FH3} that the lower bounds use is that the computation
between consecutive Hadamard layers is basis-preserving; this is what gives the round simulation of
Lemma~\ref{lem:round-embedding}.

The two lowest levels of the standard-query hierarchy are classical relative to every oracle; the
phase-query hierarchy behaves differently there (Remark~\ref{rem:phase-vs-bit}).

\begin{lemma}[The bottom of the standard-query hierarchy]\label{lem:fh0-std}
For every oracle $O$ we have $\FH_{0,\mathrm{std}}^{O}=\P^{O}$, and hence
$\P^{O}\subseteq\FH_{k,\mathrm{std}}^{O}$ for every $k\ge0$.  By contrast, the acceptance probability
of an $\FH_1^{O}$ phase circuit does not depend on $O$ (Lemma~\ref{lem:round-embedding}).
\end{lemma}

\begin{proof}
A depth-zero standard-query circuit is a single basis-preserving block applied to $\ket{0}$.  Every
gate available in a block, namely the reversible classical gates and the indexed standard query of
Definition~\ref{def:standard-fh}, maps a computational-basis state to a computational-basis state,
and the diagonal phase gates contribute a phase that is global along the single branch and does not
affect the final measurement.  The circuit therefore carries out a deterministic classical
computation with adaptive oracle access, whose accepting set is recognizable in deterministic
polynomial time, so its language lies in $\P^{O}$.  Conversely, Bennett's reversible
simulation~\cite{bennett1973logical} converts any polynomial-time oracle computation into such a
block with polynomial overhead.
\end{proof}

\begin{lemma}[The first standard-query level is $\BPP$]\label{lem:fh1-std}
For every oracle family $O$ we have $\FH_{1,\mathrm{std}}^{O}=\BPP^{O}$.  Consequently
$\FH_{1,\mathrm{std}}^{O}\subseteq\PH^{O}$ for every $O$.
\end{lemma}

\begin{proof}
For the inclusion $\supseteq$, a $\BPP^{O}$ computation tosses $m=\poly(n)$ coins and runs a
deterministic polynomial-time oracle computation on the outcome.  Realize it as follows: take $U_0$
to be the identity, let the single Hadamard layer act on $m$ designated wires, so that the state
becomes the uniform superposition $2^{-m/2}\sum_{w}\ket w$ over coin strings, and let $U_1$ be
Bennett's reversible simulation~\cite{bennett1973logical} of the deterministic computation, which is
basis-preserving and uses standard queries.  Measuring gives the correct output distribution.

For the inclusion $\subseteq$, let $C=U_1\,(H^{\otimes m}\otimes I)\,U_0$ be a standard-query
$\FH_1$ circuit with accepting set $S$.  The block $U_0$ is applied to a computational-basis
state, so by the argument of Lemma~\ref{lem:fh0-std} it is a deterministic polynomial-time oracle
computation and produces $\lambda_0\ket{z_0}$ for a single basis label $z_0$, computable in
$\P^{O}$.  Writing $z_0=(z_0^{\mathrm{des}},z_0^{\mathrm{anc}})$ for its restriction to the
designated and the ancilla wires, the Hadamard layer gives
\[
2^{-m/2}\sum_{w\in\bits^{m}}(-1)^{w\cdot z_0^{\mathrm{des}}}\,\ket w\ket{z_0^{\mathrm{anc}}}.
\]
Now $U_1$ is basis-preserving and unitary, hence a permutation matrix times a diagonal unitary:
$U_1\ket{v}=\lambda(v)\ket{\pi(v)}$ with $\pi$ a permutation of the computational basis.  Distinct
branches therefore reach distinct basis labels, so no two branches interfere at the final
measurement and
\[
\Prb[C\text{ accepts}]
=2^{-m}\,\bigl|\{w\in\bits^{m}:\ \pi(w,z_0^{\mathrm{anc}})\in S\}\bigr|.
\]
For each fixed $w$ the label $\pi(w,z_0^{\mathrm{anc}})$ is obtained by simulating the
polynomial-size basis-preserving circuit $U_1$ gate by gate on a basis state, using one oracle query
per standard-query gate, and $S$ is recognizable in deterministic polynomial time.  This
probability is thus the acceptance probability of a $\BPP^{O}$ machine that samples $w$ uniformly
and runs that computation.  The final containment is the relativized Sipser--Lautemann inclusion
$\BPP^{O}\subseteq\Sigma_2^{\mathrm p,O}$.
\end{proof}

\subsection{The SWAP test}

We will use the following standard test throughout; see, for example,~\cite{nielsen2010quantum}.

\begin{lemma}[SWAP test]\label{lem:swap-test}
Let $\ket{\psi}$ and $\ket{\varphi}$ be pure states on registers of the same dimension.  The
circuit that prepares a control qubit in $\ket{+}$, applies controlled-SWAP to the two data
registers, applies a Hadamard to the control, and accepts on control outcome $0$ has acceptance
probability
\[
\frac{1+|\braket{\psi|\varphi}|^2}{2}.
\]
\end{lemma}

\begin{proof}
Immediately after controlled-SWAP, the state is
\[
\frac{1}{\sqrt2}
\bigl(\ket0\ket{\psi}\ket{\varphi}+\ket1\ket{\varphi}\ket{\psi}\bigr).
\]
The squared norm of the branch with control outcome $0$ after the final Hadamard is
\[
\frac14\left(2+2|\braket{\psi|\varphi}|^2\right)
=
\frac{1+|\braket{\psi|\varphi}|^2}{2}.
\]
\end{proof}

\subsection{Fourier weight and concentration inequalities}

We use the standard Fourier expansion on the Boolean cube; see, for
example,~\cite{odonnell2021analysisbooleanfunctions}.
\begin{definition}[Level-$\ell$ Fourier weight]\label{def:fourier-weight}
Every $f:\pmone^{M}\to\mathbb C$ has a unique multilinear expansion
$f(x)=\sum_{S\subseteq[M]}\widehat f(S)\prod_{i\in S}x_i$ with
$\widehat f(S)=\E_x[f(x)\prod_{i\in S}x_i]$; the biased weights below are defined in the same way.
The \emph{level-$\ell$ Fourier weight} is
\[
L_{1,\ell}(f)=\mathsf{wt}_\ell(f):=\sum_{\abs{S}=\ell}\abs{\widehat f(S)}.
\]
More generally, following~\cite{bansal2021k}, for a bias vector $\mu\in(-1,1)^M$ let
$p_\mu$ be the product measure with $\E_{p_\mu}[x_i]=\mu_i$, let
$\psi_i^\mu(x)=(x_i-\mu_i)/\sqrt{1-\mu_i^2}$ be the associated orthonormal characters, and
let $\widehat f^\mu(S)$ be the coefficients of $f$ in the basis
$\{\prod_{i\in S}\psi_i^\mu\}$.  The $\mu$-biased level-$\ell$ weight is
$\mathsf{wt}^\mu_\ell(f):=\sum_{\abs{S}=\ell}\abs{\widehat f^\mu(S)}$.  A restriction
$\rho\in\{-1,1,\star\}^M$ fixes the coordinates with $\rho_i\ne\star$ and leaves the others free;
$f_\rho$ denotes the restricted function on the free coordinates.
\end{definition}

\begin{lemma}[Hoeffding's inequality; see \cite{hoeffding1963probability}]\label{lem:hoeffding}
Let $\xi_1,\dots,\xi_M$ be independent random variables with $\xi_i\in[a_i,b_i]$, and let
$S=\sum_i\xi_i$.  Then for every $s\ge0$,
\[
\Prb\bigl[\abs{S-\E S}\ge s\bigr]\le 2\exp\Bigl(-\frac{2s^2}{\sum_i(b_i-a_i)^2}\Bigr).
\]
Each of the one-sided events $S-\E S\ge s$ and $S-\E S\le-s$ has probability at most half the
bound above.  In particular, for a weighted Rademacher sum $S=\sum_i w_i\xi_i$ with unbiased
signs, $\Prb[\abs{S}\ge s]\le2\exp(-s^2/(2\Norm{w}_2^2))$; and the empirical frequency of $R$ independent
Bernoulli trials deviates from its mean by more than $s$ with probability at most $2e^{-2Rs^2}$,
and exceeds (respectively, falls below) it by more than $s$ with probability at most $e^{-2Rs^2}$.
\end{lemma}

\subsection{Fourier growth of bounded-round algorithms}\label{subsec:prelim-growth}

The lower bounds in this paper bound the Fourier weight of a circuit's acceptance probability by
simulating the circuit in the bounded-round query model of Girish, Sinha, Tal and
Wu~\cite{girish2024power}, and then applying their growth theorem.  We collect here the statements
of that model and of the two growth bounds we use.

\begin{definition}[Bounded-round parallel-query algorithms~\cite{girish2024power}]\label{def:round-model}
Let $x\in\pmone^{M}$.  The query oracle $O_x$, which is a phase query in the terminology of
Section~\ref{subsec:main-results} and the \emph{standard} query oracle in that
of~\cite{girish2024power}, acts on an $(M+1)$ dimensional register by
\[
O_x\ket{0}=\ket{0},
\qquad
O_x\ket{i}=x_i\ket{i}
\quad (1\le i\le M).
\]
An \emph{$r$-round algorithm with $t$ parallel queries per round} applies, to an arbitrary
input-independent initial state on the query registers and an arbitrary workspace,
\[
U_r\,(O_x^{\otimes t}\otimes I)\,U_{r-1}\cdots U_1\,(O_x^{\otimes t}\otimes I)\,U_0,
\]
where each $U_i$ is an arbitrary input-independent unitary ($t=0$ is allowed, with
$O_x^{\otimes0}=I$, so that a round may make no queries), followed by a two-outcome
measurement.  A tuple of Boolean phase oracles $F_0,\dots,F_{J-1}:\bits^n\to\pmone$ is accessed
through the concatenated string $x\in\pmone^{J\cdot N}$ of their truth tables: a phase query to
$F_a$ at the point $y$ is the query $O_x$ at the index $i=(a,y)$, and the index $0$, on which
the oracle acts trivially, is used for unneeded parallel queries.
\end{definition}

The phase-query lower bound of Section~\ref{sec:higher-level-conjecture} rests on the Fourier growth
theorem of Girish, Sinha, Tal, and Wu for bounded-round algorithms.  The standard-query lower bound
of Section~\ref{sec:standard-hierarchy} does not; it uses instead the decision-tree growth theorem
proved there, by a method adapted from the same paper.

We use the theorem in a form whose constant is singly exponential in $r\ell$, as the
growing-depth statements below require; Appendix~\ref{app:gstw-growth-constant} proves that form.
Below, $L_{1,\ell}$
denotes the level-$\ell$ Fourier weight of Definition~\ref{def:fourier-weight}.

\begin{theorem}[{Fourier growth for bounded-round algorithms; Girish--Sinha--Tal--Wu~\cite[Corollary~4.2]{girish2024power}}]\label{thm:gstw-growth}
Let $r\ge1$, let $\mathcal A$ be a
quantum query algorithm on $M$ bit inputs with arbitrarily many auxiliary qubits, making $r$ adaptive
rounds of $t$ parallel queries per round, where $1\le t\le M$, and let $f:\pmone^{M}\to[0,1]$ be its acceptance
probability.  Then for every $\ell\ge 1$,
\[
L_{1,\ell}(f)
\le
(2^{2r}-1)^{2\ell}\cdot t^{\ell}\cdot
(M/t)^{\frac12\left\lfloor\frac{(2r-1)\ell}{2r}\right\rfloor}.
\]
The same bound holds for every restriction $f_\rho$ of $f$, with $M$ replaced by the number
$\widetilde M$ of variables that $\rho$ leaves free.
\end{theorem}

The constant $(2^{2r}-1)^{2\ell}$ is singly exponential in $r\ell$ and uniform over all levels
$\ell\ge1$; it comes from the closed-form evaluation in
Appendix~\ref{app:gstw-growth-constant}.  For constant $r$ the factor is an
absorbed constant, and every application in Sections~\ref{sec:higher-level-conjecture}
and~\ref{sec:landscape} at fixed $k$ uses it only as such.

Theorem~\ref{thm:gstw-growth} holds for every number of rounds.  In the non-adaptive case $r=1$,
which contains the case of Fourier depth two (Lemma~\ref{lem:round-embedding}), a sharper bound is
available: the dependence on the number of parallel queries improves from $t^{\ell}$ to $t^{\ell/4}$,
and the constant improves from $9^{\ell}$ to $\ell+1$.  Girish, Sinha, Tal, and
Wu~\cite[Remark~1.6]{girish2024power} state this bound and give the argument for it in their proof
overview, while the detailed proof in their Section~4 establishes only the general form
(Theorem~\ref{thm:gstw-growth}).  Since our application depends on how the constant grows with
$\ell$, and since the case of two oracle factors does not involve the assignment count of
Appendix~\ref{app:gstw-growth-constant}, we give a complete proof with explicit constants.  The proof also
gives a bound that is stable under restrictions.

\begin{theorem}[{Non-adaptive Fourier growth; \cite[Remark~1.6]{girish2024power}}]\label{thm:nonadaptive-growth}
Let $M,t,m\ge1$, let $A=(\{0\}\cup[M])^{t}\times[m]$, let $u,v\in\mathbb C^{A}$ satisfy
$\Norm u,\Norm v\le1$, and let $\mathsf B\in\mathbb C^{A\times A}$ satisfy $\Norm{\mathsf B}\le1$.
Define $f:\pmone^{M}\to\mathbb C$ by
\[
f(x)=u^{\dagger}\bigl(O_x^{\otimes t}\otimes I_m\bigr)\,\mathsf B\,
\bigl(O_x^{\otimes t}\otimes I_m\bigr)\,v,
\]
where $O_x$ is the query oracle of Definition~\ref{def:round-model}.  Then for every restriction
$\rho\in\{-1,1,\star\}^{M}$ leaving $\widetilde M$ variables free and every $\ell\ge1$,
\[
L_{1,\ell}(f_\rho)\ \le\ (\ell+1)\,\bigl(\widetilde M\,\tau\bigr)^{\ell/4},
\qquad
\tau:=\min\{t,\widetilde M\}.
\]
In particular the acceptance probability of a non-adaptive quantum query algorithm making $t$
parallel queries is of the form above, with $\mathsf B=U_1^{\dagger}\Pi U_1$ for the acceptance
projector $\Pi$, and therefore obeys this bound.
\end{theorem}

Theorem~\ref{thm:nonadaptive-growth} is what gives the sharpened threshold at the second level
(Remark~\ref{rem:fh2-sharpened}) and the growth bound used in Appendix~\ref{sec:appendix-sign}.  Its
proof is self-contained and is given in Appendix~\ref{app:nonadaptive-proof}.

\section{An oracle separation at every level}\label{sec:higher-level-conjecture}

This section proves the first of the two adjacent-level separation theorems: for every constant
$k\ge2$ there
is an oracle relative to which $\FH_k$ and $\FH_{k+1}$ differ in the phase-query hierarchy
(Theorem~\ref{thm:all-level-separation}).  Section~\ref{sec:standard-hierarchy} proves the
corresponding statement for Shi's standard-query hierarchy
(Theorem~\ref{thm:standard-query-adjacent}).

Section~\ref{subsec:odd-forrelation} constructs the algorithm.  The hard problem is
$(2k-1)$-fold Forrelation.  Its defining sequence of Hadamard transforms and oracles is the standard
Forrelation circuit of Aaronson and
Ambainis~\cite{aaronson2014forrelationproblemoptimallyseparates}; we show that, for odd order, it splits at the middle oracle into two states each prepared by exactly $k$ Hadamard
layers, and that a control-qubit interference test comparing the two states uses exactly one more
layer.  Fourier depth
$k+1$ therefore suffices to estimate the Forrelation value itself on every oracle tuple.
Section~\ref{subsec:round-embedding} proves the corresponding statement about depth $k$: every
$\FH_k$ circuit is simulated exactly by a quantum query algorithm with $k-1$ rounds of parallel
queries.  Together with the Fourier growth theorem of~\cite{girish2024power}, this bounds the Fourier
weight of the acceptance probability of any depth-$k$ circuit.
Section~\ref{subsec:hard-problem} imports the input distributions and the advantage criterion of
Bansal and Sinha~\cite{bansal2021k}, which hold for every order, odd or even, and combines the
two bounds into a worst-case statement: an explicit promise problem that depth $k+1$ decides with
constantly many queries and that depth $k$ does not decide at all.
Section~\ref{subsec:separation-theorem} carries out the diagonalization.  The standard-query proof of
Section~\ref{sec:standard-hierarchy} reuses the upper-bound circuit of
Section~\ref{subsec:odd-forrelation} and the diagonalization of
Section~\ref{subsec:separation-theorem} unchanged, and replaces the round simulation by a growth
bound for decision trees.

The Fourier growth bounds and the Gaussian analysis of the structured distributions are
from~\cite{girish2024power,bansal2021k}, the former proved with a single-exponential constant in
Appendix~\ref{app:gstw-growth-constant}; the remaining steps are new, as listed in
Section~\ref{subsec:main-results}.

Throughout Sections~\ref{sec:higher-level-conjecture}--\ref{sec:landscape} we write $N:=2^n$, and we
write $\delta:=2^{-5K}$ for the gap parameter of the Bansal--Sinha distributions for $K$-fold
Forrelation;
the value of $K$ is fixed in each statement, and is not the same in every section.  In the
present section we fix a level $k\ge2$ and take
\[
K:=2k-1,
\qquad
\delta=2^{-5K}.
\]
We treat $k$, and hence $K$ and $\delta$, as constants; all quantities denoted $C_k$, $c_k$ or
$O_k(\cdot)$ may depend on $k$ but never on $n$.

\subsection{Odd Forrelation at depth \texorpdfstring{$k+1$}{k+1}}\label{subsec:odd-forrelation}

For phase oracles
\[
F_0,\dots,F_{2k-2}:\bits^n\to\pmone,
\]
define the $(2k-1)$-fold Forrelation function by
\begin{equation}\label{eq:odd-Phi}
\Phi_{2k-1}(F_0,\dots,F_{2k-2})
=
N^{-k}
\sum_{x_0,\dots,x_{2k-2}\in\bits^n}
\left(\prod_{j=0}^{2k-2}F_j(x_j)\right)
\left(\prod_{j=0}^{2k-3}(-1)^{x_j\cdot x_{j+1}}\right).
\end{equation}
For $k=2$, this is the threefold Forrelation function studied in
Appendix~\ref{sec:appendix-sign}.  Writing $z_j\in\pmone^N$ for the truth table of $F_{j-1}$ and
$\mathsf H\in\mathbb R^{N\times N}$ for the normalized Hadamard matrix
$\mathsf H_{x,y}=N^{-1/2}(-1)^{x\cdot y}$, the quantity \eqref{eq:odd-Phi} coincides with the
$K$-fold Forrelation function
\begin{equation}\label{eq:forr-dictionary}
\mathsf{forr}_K(z)
=
\frac1N\, z_1^{\top}\bigl(\mathsf H\,\mathsf Z_2\,\mathsf H\,\mathsf Z_3\cdots
\mathsf H\,\mathsf Z_{K-1}\,\mathsf H\bigr)z_K,
\qquad
\mathsf Z_i=\operatorname{diag}(z_i),
\end{equation}
of Aaronson--Ambainis~\cite{aaronson2014forrelationproblemoptimallyseparates} in the
normalization of Bansal--Sinha~\cite[Eqs.~(1.1)--(1.2)]{bansal2021k}; in particular
$\Phi_{2k-1}\in[-1,1]$ always, because $z_1/\sqrt N$ and $z_K/\sqrt N$ are unit vectors and the
matrix in \eqref{eq:forr-dictionary} is a contraction.

The definition involves a sequence of $2k-1$ variables.  Because this number is odd, fixing the
middle variable $y=x_{k-1}$ splits the sum into a part involving only the first
$k-1$ oracles and a part involving only the last $k-1$.  Define the two half-sums by
\begin{equation}\label{eq:left-profile}
L(y)
=
N^{-(k-1)/2}
\sum_{x_0,\dots,x_{k-2}\in\bits^n}
\left(\prod_{j=0}^{k-2}F_j(x_j)\right)
\left(\prod_{j=0}^{k-3}(-1)^{x_j\cdot x_{j+1}}\right)
(-1)^{x_{k-2}\cdot y}
\end{equation}
and
\begin{equation}\label{eq:right-profile}
R(y)
=
N^{-(k-1)/2}
\sum_{x_k,\dots,x_{2k-2}\in\bits^n}
(-1)^{y\cdot x_k}
\left(\prod_{j=k}^{2k-2}F_j(x_j)\right)
\left(\prod_{j=k}^{2k-3}(-1)^{x_j\cdot x_{j+1}}\right),
\end{equation}
with empty products interpreted as $1$.  Both are functions of the middle variable $y$ alone.

\begin{proposition}[Split form of odd Forrelation]\label{prop:odd-split}
For every $F_0,\dots,F_{2k-2}:\bits^n\to\pmone$,
\[
\Phi_{2k-1}(F_0,\dots,F_{2k-2})
=
N^{-1}\sum_{y\in\bits^n}L(y)\,F_{k-1}(y)\,R(y).
\]
\end{proposition}

\begin{proof}
Start from \eqref{eq:odd-Phi} and set $y=x_{k-1}$.  The variables
$x_0,\dots,x_{k-2}$ appear only in the left part of the expression, the variables
$x_k,\dots,x_{2k-2}$ appear only in the right part, and the middle oracle contributes
the factor $F_{k-1}(y)$.  The normalizations multiply as
$N^{-1}\cdot N^{-(k-1)/2}\cdot N^{-(k-1)/2}=N^{-k}$.  Grouping the left and right sums gives
the identity.
\end{proof}

This identity has an equivalent form in terms of quantum states, which is what relates it to the
hierarchy.  Define
\[
\ket{\psi_L}
=
H^{\otimes n}P_{F_{k-2}}H^{\otimes n}\cdots
P_{F_1}H^{\otimes n}P_{F_0}H^{\otimes n}\ket{0^n},
\qquad
\ket{\psi_R}
=
H^{\otimes n}P_{F_k}H^{\otimes n}\cdots
P_{F_{2k-2}}H^{\otimes n}\ket{0^n}.
\]
Each state is prepared by exactly $k$ Hadamard layers with only phase queries in between, and a
direct expansion gives
\[
\langle y|\psi_L\rangle=N^{-1/2}L(y),
\qquad
\langle y|\psi_R\rangle=N^{-1/2}R(y).
\]
Hence Proposition~\ref{prop:odd-split} is equivalently the identity
\begin{equation}\label{eq:odd-overlap}
\Phi_{2k-1}(F_0,\dots,F_{2k-2})
=
\bra{\psi_L}P_{F_{k-1}}\ket{\psi_R}:
\end{equation}
that is, the $(2k-1)$-fold Forrelation value is the inner product of two states of Fourier depth
$k$, after one further basis-preserving query.  Aaronson and
Ambainis~\cite{aaronson2014forrelationproblemoptimallyseparates} compare the two half-states using a
control-qubit interference test, which estimates $\mathsf{forr}$ itself with $\lceil K/2\rceil$
queries.  We instead compare them with a SWAP test, whose only gate that is not basis-preserving is one
Hadamard on the control qubit.  Two of its properties are relevant here: it estimates the \emph{square}
$\mathsf{forr}^2$ rather than $\mathsf{forr}$, and, since the two halves have the same depth $k$, its
Fourier depth is exactly $k+1$.  This gives the following proposition.  We state it for an arbitrary oracle tuple, since both
the worst-case statements of this section and the statements about specific ensembles in
Appendix~\ref{sec:appendix-sign} are obtained from the same circuit.

\begin{proposition}[A SWAP test for odd Forrelation at depth $k+1$]\label{prop:odd-swap-circuit}
For every fixed $k\ge 2$ there is a uniform polynomial-size $\FH_{k+1}$ circuit family making
$2k-1$ phase queries whose acceptance probability on every oracle tuple
$(F_0,\dots,F_{2k-2})$ equals
\[
p_{\mathrm{swap}}^{(k)}(F_0,\dots,F_{2k-2})
=
\frac{1+\Phi_{2k-1}(F_0,\dots,F_{2k-2})^2}{2},
\]
where $\Phi_{2k-1}$ is the $(2k-1)$-fold Forrelation of \eqref{eq:odd-Phi}.
\end{proposition}

\begin{proof}
The circuit uses two $n$-qubit data registers, denoted $R_L$ and $R_R$, one control qubit $c$, and
ancilla registers used for routing.

\emph{The circuit.}
We first prepare $\ket{\psi_L}$ on $R_L$ and $\ket{\psi_R}$ on $R_R$ in parallel.  Each preparation
uses exactly $k$ global Hadamard layers, and the two preparations share these layers: at the $j$th
layer the designated wires are those of $R_L$ and $R_R$, together with, at the first layer only, the
control qubit $c$, which is thereby prepared in $\ket{+}$.  Before each later layer, SWAP gates
exchange $c$ with an ancilla in a computational basis state, so that the layer does not act on it
(Remark~\ref{rem:register-routing}).  Between consecutive Hadamard layers the circuit performs only
phase queries.  After the $k$th layer, apply the middle phase oracle
$P_{F_{k-1}}$ to $R_R$, then a SWAP of $R_L$ and $R_R$ controlled by $c$, then route $c$ back
among the designated wires, apply the $(k+1)$st global Hadamard layer to $c$, and measure.
The circuit accepts if the outcome of $c$ is $0$.

\emph{Depth and query count.}
All gates other than the global Hadamard layers are basis-preserving: phase queries are
diagonal in the computational basis, and SWAP and controlled SWAP gates are permutations of
computational basis states.  Hence the circuit uses $k+1$ global Hadamard layers and is a uniform
$\FH_{k+1}$ circuit.  It makes
$2(k-1)$ phase queries in the two preparations and one middle query, for a total of $2k-1$.

\emph{The acceptance probability.}
By the SWAP test identity (Lemma~\ref{lem:swap-test}) applied to the states
$\ket{\psi_L}$ and $P_{F_{k-1}}\ket{\psi_R}$, the probability of outcome $0$ is
\[
\frac{1+\abs{\bra{\psi_L}P_{F_{k-1}}\ket{\psi_R}}^2}{2}
=
\frac{1+\Phi_{2k-1}(F_0,\dots,F_{2k-2})^2}{2},
\]
where the last equality is \eqref{eq:odd-overlap} and the Forrelation value is real.
\end{proof}

The SWAP test is not the sharpest way to use the split.  Because the oracle of
Definition~\ref{def:FHk-rel} is indexed, a single phase query can select which function it applies
according to a control qubit held in superposition, and the distinguished index carries the
constant-zero function, so a query can be made to act on one branch of that control and not on the
other.  This turns the control-qubit interference test of Aaronson and
Ambainis~\cite{aaronson2014forrelationproblemoptimallyseparates} into an $\FH_{k+1}$ circuit that
estimates the Forrelation value itself, and does so with fewer queries.

\begin{proposition}[An interference test for odd Forrelation at depth $k+1$]\label{prop:odd-interference-circuit}
For every fixed $k\ge 2$ there is a uniform polynomial-size $\FH_{k+1}$ circuit family making $k$
phase queries whose acceptance probability on every oracle tuple $(F_0,\dots,F_{2k-2})$ equals
\[
p_{\mathrm{int}}^{(k)}(F_0,\dots,F_{2k-2})
=
\frac{1+\Phi_{2k-1}(F_0,\dots,F_{2k-2})}{2},
\]
where $\Phi_{2k-1}$ is the $(2k-1)$-fold Forrelation of \eqref{eq:odd-Phi}.
\end{proposition}

\begin{proof}
The circuit uses one $n$-qubit data register $D$, a control qubit $c$, an index register, and
ancillas used for routing.

\emph{The two half-states, on the two branches of the control.}
The first global Hadamard layer acts on $c$ and on $D$, leaving $\ket+_c\otimes\Had\ket{0^n}_D$.
For $j=1,\dots,k-1$, the block following the $j$th layer computes into the index register, by a
classical reversible circuit reading $c$, the index of $F_{j-1}$ when $c=0$ and the index of
$F_{2k-1-j}$ when $c=1$; it then makes one indexed phase query and uncomputes the index.  The query is diagonal and the
index register is entangled with $c$, so this single query applies $F_{j-1}$ on one branch and
$F_{2k-1-j}$ on the other.  The
$(j+1)$st layer acts on $D$ alone, with $c$ routed off the designated wires by the
basis-preserving SWAP gates of Remark~\ref{rem:register-routing}.  On the branch $c=0$ the data
register is therefore acted on by $F_0,\dots,F_{k-2}$ in order, separated by the layers
$1,\dots,k$, and on the branch $c=1$ by $F_{2k-2},F_{2k-3},\dots,F_k$; after the $k$th layer the
state is
\[
\frac{1}{\sqrt2}\bigl(\ket0_c\ket{\psi_L}+\ket1_c\ket{\psi_R}\bigr),
\]
with $\ket{\psi_L}$ and $\ket{\psi_R}$ as in \eqref{eq:odd-overlap}.

\emph{The middle query, on one branch only.}
The block following the $k$th layer makes the $k$th query.  It computes into the index register the
distinguished index $0^{s(n)}$ when $c=0$ and the index of $F_{k-1}$ when $c=1$, makes one indexed
phase query, and uncomputes the index.  The distinguished index carries the constant-zero function,
so the branch $c=0$ acquires the phase $+1$ and the branch $c=1$ acquires $F_{k-1}(y)$ at each
point $y$, leaving
\[
\frac{1}{\sqrt2}\bigl(\ket0_c\ket{\psi_L}+\ket1_c P_{F_{k-1}}\ket{\psi_R}\bigr).
\]
\emph{Reading the relative phase.}
The $(k+1)$st layer acts on $c$ alone, and $c$ is measured; the circuit accepts on outcome $0$.
That layer produces
\[
\frac12\Bigl(\ket0\bigl(\ket{\psi_L}+P_{F_{k-1}}\ket{\psi_R}\bigr)
+\ket1\bigl(\ket{\psi_L}-P_{F_{k-1}}\ket{\psi_R}\bigr)\Bigr),
\]
so, by $\Norm{a+b}^2=\Norm a^2+\Norm b^2+2\Re\langle a,b\rangle$ and since
$\ket{\psi_L}$ and $P_{F_{k-1}}\ket{\psi_R}$ are unit vectors,
\[
\Prb[c=0]
=\frac14\Norm{\ket{\psi_L}+P_{F_{k-1}}\ket{\psi_R}}^2
=\frac14\bigl(2+2\Re\bra{\psi_L}P_{F_{k-1}}\ket{\psi_R}\bigr)
=\frac{1+\Phi_{2k-1}}{2},
\]
where the last equality is \eqref{eq:odd-overlap} and the Forrelation value is real.

\emph{Depth and query count.}
The circuit makes $k-1$ queries in the blocks following the layers $1,\dots,k-1$ and one after the
$k$th layer, so $k$ in all.  Every gate outside the $k+1$ global Hadamard layers is
basis-preserving: indexed phase queries are diagonal in the computational basis, and the index
computation, its uncomputation and the register routing are classical reversible operations.  The
circuit is therefore a uniform $\FH_{k+1}$ circuit.
\end{proof}

\begin{remark}[Why odd-order Forrelation]\label{rem:why-odd}
A $2k$-fold Forrelation expression has $2k+1$ Hadamard layers when written as an amplitude, and
cutting it at any
query leaves at least $k+1$ layers on one side.  Preparing that side therefore already requires $k+1$
layers, and comparing the two halves would give Fourier depth $k+2$.  An odd-order expression, on the other
hand, splits
into two halves of depth $k$ together with one basis-preserving middle query, and this is the only
splitting that places the comparison at depth $k+1$.
\end{remark}

\subsection{Simulation of Fourier depth by bounded adaptivity}\label{subsec:round-embedding}

We now prove the structural statement used for the lower bound, in the bounded-round model of
Definition~\ref{def:round-model}.

\begin{lemma}[Round simulation]\label{lem:round-embedding}
Let $k\ge 1$ and let $C$ be an $\FH_k$ oracle circuit that makes $q$ phase queries in total to a
tuple of Boolean phase oracles with concatenated truth table length $M$.  Then there is a quantum
query algorithm $\mathcal A$ with $k-1$ adaptive rounds of at most $q$ parallel queries per round
such that, for every oracle tuple,
\[
\Prb[\mathcal A\text{ accepts}]=\Prb[C\text{ accepts}].
\]
In particular, every $\FH_2$ computation making polynomially many queries is simulated exactly
by a non-adaptive quantum query algorithm with polynomially many parallel queries, and the
acceptance probability of an $\FH_1$ computation making $q$ queries is independent of the oracle.
\end{lemma}

\begin{proof}
Let $j\le k$ be the exact Fourier depth of $C$; Definition~\ref{def:FHk-rel} permits $j<k$.  If
$j=0$, then $C$ starts in a computational-basis state and applies only basis-preserving gates and
phase queries, so the measured label is oracle-independent and each phase query contributes only a
global phase; the acceptance probability is then oracle-independent and is reproduced by a
zero-query algorithm with $k-1$ empty rounds.  Assume henceforth $j\ge1$.  The argument below
gives a $(j-1)$-round simulation; appending $k-j$ empty rounds, which make no queries, yields
the stated $(k-1)$-round algorithm.  We therefore prove the claim for exact depth $k$.  Write the
circuit as
\[
C=U_k\,(H^{\otimes m}\otimes I)\,U_{k-1}\,(H^{\otimes m}\otimes I)\cdots
U_1\,(H^{\otimes m}\otimes I)\,U_0,
\]
where each block $U_j$ is a basis-preserving circuit composed of oracle-independent
basis-preserving gates and phase queries, followed by a computational basis measurement of all
wires with accepting set $S$ (Definition~\ref{def:FH2FH3}).

\emph{Step 1: the queries outside the interior blocks can be removed.}
Consider first the initial block $U_0$.  All ancillas start in computational basis states, and
every gate of $U_0$ is basis-preserving, so the state remains a single computational basis state
throughout $U_0$.  A phase query applied to a basis state multiplies the global state by the
scalar $F_a(y)\in\pmone$, where $(a,y)$ is determined by the current basis state.  A global
scalar of modulus one does not affect any outcome probability, so all queries in $U_0$ may be
deleted.

Consider next the final block $U_k$.  Every gate of $U_k$ is basis-preserving, so there are a
permutation $\pi_k$ of the computational basis and phases $\lambda^{(x)}(w)$ of modulus one,
possibly depending on the oracle string $x$ through the queried values, such that
$U_k\ket{w}=\lambda^{(x)}(w)\ket{\pi_k(w)}$; the permutation $\pi_k$ is oracle-independent
because phase queries do not move basis states.  For the acceptance projector
$\Pi_S=\sum_{s\in S}\ket{s}\bra{s}$,
\[
\bra{w}U_k^{\dagger}\Pi_S U_k\ket{w'}
=
\overline{\lambda^{(x)}(w)}\lambda^{(x)}(w')\,
\bra{\pi_k(w)}\Pi_S\ket{\pi_k(w')}
=
\delta_{w,w'}\,\mathbf 1\{\pi_k(w)\in S\},
\]
since $\pi_k$ is injective and $\Pi_S$ is diagonal.  Hence
$U_k^{\dagger}\Pi_SU_k=\Pi_{\pi_k^{-1}(S)}$ regardless of the oracle, so $U_k$ may be deleted
altogether provided the final measurement uses the accepting set $\pi_k^{-1}(S)$ in place of $S$.
All queries in $U_k$ are removed with it.

After Step 1 the circuit makes all its queries inside the $k-1$ interior blocks
$U_1,\dots,U_{k-1}$; let $q_j$ be the number of queries in $U_j$, so $\sum_j q_j\le q$.  If
$k=1$ there are no interior blocks and the acceptance probability is oracle-independent,
which proves the last assertion.

\emph{Step 2: inside a block, the query addresses are oracle-independent.}
Fix an interior block $U_j$ with gate sequence $g_1,\dots,g_T$, and let $\ket{z}$ be a
computational basis state entering the block.  We claim that the register content after any
prefix $g_1,\dots,g_i$ is a basis state $\ket{w_i(z)}$ whose label does not depend on the oracle,
and that the state equals a product of queried oracle values and oracle-independent phases times
$\ket{w_i(z)}$.  This holds by induction on $i$: an oracle-independent basis-preserving gate maps
$\ket{w}\mapsto\lambda(w)\ket{\sigma(w)}$ with fixed $\lambda,\sigma$; a phase query multiplies
the state by the oracle value at the address currently held in the queried registers and does
not change the register content.  Consequently the address of the $i$th query of the block on
branch $z$ is a fixed function $a_j^{(i)}(z)$ of the branch label, computable by composing the
oracle-independent permutations of the block, and the block acts as
\[
U_j\ket{z}
=
\lambda_j(z)\,
\Bigl(\prod_{i=1}^{q_j}x_{a_j^{(i)}(z)}\Bigr)\,
\ket{\pi_j(z)},
\]
with $\lambda_j$, $\pi_j$, and the address maps $a_j^{(i)}$ all oracle-independent.

\emph{Step 3: one round for each block.}
Set $t:=\max_j q_j\le q$.  The simulating algorithm $\mathcal A$ uses the register of $C$
together with $t$ fresh address registers, each of dimension $M+1$ and initialized to
$\ket{0}$.  Its initial unitary prepares $(H^{\otimes m}\otimes I)\,U_0'\ket{0\cdots 0}$, where
$U_0'$ is the query-free remainder of $U_0$; this is an input-independent unitary.  Round $j$
(for $1\le j\le k-1$) is defined as follows.
\begin{enumerate}[leftmargin=2em]
\item Apply the input-independent unitary
$V_j:\ket{z}\ket{0}^{\otimes t}\mapsto
\ket{z}\ket{a_j^{(1)}(z)}\cdots\ket{a_j^{(q_j)}(z)}\ket{0}^{\otimes(t-q_j)}$,
implemented by XOR of the reversibly computed addresses into the address registers.
\item Apply the parallel batch $O_x^{\otimes t}$ to the address registers.  This multiplies
branch $z$ by $\prod_{i=1}^{q_j}x_{a_j^{(i)}(z)}$; the unneeded registers hold $\ket 0$ and
contribute no phase.
\item Apply $V_j^{\dagger}$ to restore the address registers to $\ket{0}^{\otimes t}$, then apply
the input-independent unitary $D_j:\ket z\mapsto\lambda_j(z)\ket{\pi_j(z)}$, and then the next
global Hadamard layer $(H^{\otimes m}\otimes I)$.  Thus round $j$ reproduces $U_j$ followed by the
$(j+1)$-st Hadamard layer; in particular round $k-1$ applies the $k$-th and final Hadamard layer
of $C$.  The oracle-free final block $U_k$ acts, by Step~1, only as the relabeling $\pi_k$, so
no further unitary is applied and the two-outcome measurement instead projects onto
$\pi_k^{-1}(S)$.
\end{enumerate}
Items (1) and (3) above are absorbed into the arbitrary unitaries between rounds
(Definition~\ref{def:round-model}).  By Step 2, round $j$ reproduces the action of
$U_j$ followed by the next Hadamard layer exactly, on every branch and for every oracle; together
with the initial unitary, which provides the first Hadamard layer, the $k-1$ rounds apply all $k$
Hadamard layers of $C$, and the final two-outcome measurement is the projection onto
$\pi_k^{-1}(S)$.  Hence $\mathcal A$ has
$k-1$ rounds of at most $t\le q$ parallel queries and the same acceptance probability as $C$ on
every oracle tuple.
\end{proof}

\begin{remark}[The simulation is one-way]\label{rem:strict-embedding}
Lemma~\ref{lem:round-embedding} embeds $\FH_k$ into the bounded-round model
of~\cite{girish2024power} in one direction only, and the round model is syntactically more general.
An algorithm with $k-1$ rounds interleaves its $k-1$ parallel-query batches with $k$
\emph{arbitrary} input-independent unitaries, whereas an $\FH_k$ circuit is the special case in which
each such unitary consists of a single global Hadamard layer together with basis-preserving
computation.  A general unitary is not a single Hadamard layer, so a bounded-round algorithm need not
be an $\FH_k$ circuit.  The weight bound for the larger round class therefore applies in particular
to $\FH_k$, which is all that the lower bound requires, but it does not give an \emph{exact}
query-complexity characterization of the phase-query class $\FH_k$; we leave this as an open question
(Section~\ref{sec:open-problems}).
\end{remark}

For the standard oracle model, in which a query performs
$\ket{y}\ket{b}\mapsto\ket{y}\ket{b\oplus f_a(y)}$ and is likewise basis-preserving, the
register contents inside a block are no longer oracle-independent, since query answers are written
into the register and later addresses may depend on them.  Each branch of a block then evolves as an
adaptive classical query process, and the branches remain in superposition.  The adaptivity can still be replaced by parallelism,
but only at the price of exponentially many queries and a larger number of rounds
(Proposition~\ref{prop:three-round-bit-embedding}); this is why the standard-query lower bound needs
a different argument.

Two statements make this precise.  Lemma~\ref{lem:classical-prefix-lifting} shows that an
initial stage of adaptive classical queries costs only a polynomial factor in the Fourier weight; it
is how the initial block of a standard-query circuit is handled in
Section~\ref{sec:standard-hierarchy}.  Proposition~\ref{prop:three-round-bit-embedding} is the direct
simulation of a standard-query circuit by parallel rounds, at exponential width; it is useful for
comparing the two models but is not used in the lower bound.

\begin{lemma}[Classical preprocessing lifting]\label{lem:classical-prefix-lifting}
Suppose an algorithm first makes at most $d_0$ adaptive classical queries to an $M$-bit oracle string
and then, at each leaf $z$ of the resulting decision tree, runs an arbitrary continuation
$\mathcal A_z$.  Let $f$ be the acceptance probability of the whole algorithm and let $f_z$ be the
acceptance probability of $\mathcal A_z$ after the answers on the path to $z$ have been fixed.
Then, for every restriction $\rho$ and every $\ell\ge1$, the level-$\ell$ Fourier weight
$L_{1,\ell}$ of Definition~\ref{def:fourier-weight} satisfies
\[
L_{1,\ell}(f_\rho)
\le
\sum_{a=0}^{\ell}\binom{d_0}a\max_z L_{1,\ell-a}((f_z)_\rho).
\]
(The weaker bound with $d_0^a$ in place of $\binom{d_0}a$ follows at once.)  In particular, if every
continuation $\mathcal A_z$ is an $r$-round quantum query algorithm, $r\ge1$, with $t$ parallel phase
queries per round, where $1\le t\le M$, then
\[
L_{1,\ell}(f_\rho)
\le
2^{O_{r,\ell}(1)}(d_0+t)^\ell
M^{\frac{\ell}{2}(1-\frac1{2r})}.
\]
\end{lemma}

\begin{proof}
We may assume that the decision tree is read-once on every root-to-leaf path, by retaining an answer
whenever the original preprocessing would query the same coordinate again.  For a leaf $z$, let $\chi_z$ be
its path indicator and let $d_z\le d_0$ be its depth.  Define $f_z$ to be the acceptance probability
of the continuation after all coordinates queried on the path to $z$ have been fixed to
their leaf values, viewed as a function on the full cube by extending it constantly in those path
coordinates.  Thus $f=\sum_z\chi_zf_z$.  After a restriction $\rho$ consistent with $z$, the
functions $(\chi_z)_\rho$ and $(f_z)_\rho$ depend on disjoint sets of free coordinates, so their
Fourier supports combine without cancellation of degree and
\[
L_{1,\ell}\bigl((\chi_zf_z)_\rho\bigr)
\le\sum_{a=0}^{\ell}L_{1,a}((\chi_z)_\rho)\,L_{1,\ell-a}((f_z)_\rho).
\]
If $d_z'$ path coordinates remain free, then $L_{1,a}((\chi_z)_\rho)=2^{-d_z'}\binom{d_z'}{a}$, and
the consistent leaves partition the remaining cube, so $\sum_z2^{-d_z'}=1$.  Since
$\binom{d_z'}{a}\le\binom{d_0}a$, summing over the leaves proves the first inequality.  For $a<\ell$, apply
Theorem~\ref{thm:gstw-growth} at level $\ell-a\ge1$ to each quantum continuation, using its restriction
clause; the summand $a=\ell$ is at most $\binom{d_0}\ell$, since $L_{1,0}(g)=\abs{\E g}\le1$ for every $g$ with
values in $[0,1]$.  Since $1\le t\le M$, every summand is at most
$\binom{d_0}a(2^{2r}-1)^{2\ell}t^{\ell-a}M^{\frac\ell2(1-\frac1{2r})}$, and
$\binom{d_0}a\le d_0^{\,a}\le\binom\ell a d_0^{\,a}$ gives
$\sum_{a=0}^{\ell}\binom{d_0}at^{\ell-a}\le\sum_{a=0}^{\ell}\binom\ell a d_0^{\,a}t^{\ell-a}=(d_0+t)^\ell$,
which is the second inequality.
\end{proof}

\begin{proposition}[Three-round standard-query simulation]\label{prop:three-round-bit-embedding}
Let $k\ge2$ and let $C$ be a standard-query $\FH_k$ circuit making $q$ queries in total.  Its
acceptance probability is also that of an algorithm consisting of classical preprocessing of depth at
most $q$, followed by
\[
r=3k-2
\]
rounds of at most $t\le2^{q+1}$ parallel phase queries.  The equality holds on every oracle
tuple, and hence also after every restriction of the oracle bits.
\end{proposition}

\begin{proof}
\emph{Reduction to exact depth $k$.}
The statement covers every circuit of Fourier depth at most $k$.  If the exact depth is $0$, the whole
circuit is one basis-preserving block applied to a computational basis state, hence a classical
decision tree of depth at most $q$, and the claim holds with $3k-2$ empty rounds appended.  If the
exact depth is $j_0$ with $1\le j_0<k$, the construction below gives $3j_0-2$ rounds, and appending
$3(k-j_0)$ empty rounds, which make no queries, gives $3k-2$.  Assume therefore exact depth $k$.

\emph{The initial block becomes classical preprocessing.}
Write $C$ as in \eqref{eq:FHk-form}, now with standard queries, and let $q_j$ be the number of
queries in $U_j$.  The initial block $U_0$ starts in one basis state.  It is therefore a
classical decision tree of depth at most $q_0$, whose leaf specifies the basis state entering the
first Hadamard layer; its phase is global and can be discarded.  This is the classical preprocessing in the statement.

\emph{Notation for one interior block.}
Fix an interior block $U_j$ and a basis state $z$ entering it.  Its standard queries define a
decision tree $T_{j,z}$ of depth at most $q_j$.  Let $\mathsf{Addr}_j(z)$ be the canonically ordered
list of all query locations occurring at nodes of this tree; it has at most $2^{q_j}-1$ entries, and
we write $A$ for the tuple of oracle answers at those locations.  From $z$ and $A$ an
oracle-independent reversible circuit computes both the output permutation $\pi_j(z;A)$ and the
phase $\lambda_j(z;A)$ of the block.  Apply the same notation to the inverse block $U_j^{-1}$: from
an output label $w$ it computes the list $\mathsf{Addr}'_j(w)$ of query locations, and from $w$ and
the tuple $A'$ of answers at those locations it computes the unique input label $\sigma_j(w;A')$.
For the actual oracle,
\[
\sigma_j(\pi_j(z;A);A')=z.
\]

\emph{Three rounds per interior block.}
We simulate $U_j$ in three phase-query rounds.  In the first round, coherently compute
$\mathsf{Addr}_j(z)$, query its entries into a fresh answer bank $A$, uncompute the addresses, and compute
$\lambda_j(z;A)\ket{\pi_j(z;A)}$ while retaining $z$ and $A$.  In the second round, compute both
the forward list from $z$ and the inverse list from $\pi_j(z;A)$.  One parallel batch erases the
forward answer bank and fills a fresh inverse bank $A'$; then uncompute the two address lists and
XOR $\sigma_j(\pi_j(z;A);A')=z$ into the retained copy of $z$.  The state is now
$\lambda_j(z;A)\ket{A'}\ket{\pi_j(z;A)}$.  In the third round, query the inverse list once more
to erase $A'$ and uncompute its addresses.  Thus all work registers are restored to $\ket0$ and
the exact action of $U_j$ remains.  The second round uses fewer than $2(2^{q_j}-1)$ phase queries,
and the other two use fewer than $2^{q_j}$.

\emph{Standard-query batches from phase-query batches.}
Each standard-query batch is implemented from one phase-query batch as follows: apply a Hadamard
gate to each answer target $b$, route the corresponding address register to the queried index when
$b=1$ and to the dummy index $0$ when $b=0$, which is a basis-preserving operation controlled on
$b$, make the phase query, undo the routing, and apply a Hadamard gate to $b$ again.  Between the
two Hadamard gates the phase $(-1)^{bf(i)}$ is applied, which implements $b\mapsto b\oplus f(i)$.
The Hadamard gates and the routing are part of the arbitrary inter-round unitaries.

\emph{The final block, and the totals.}
After the final Hadamard layer, simulate $U_k$ in one additional round: compute its full
list of possible query locations from the entering basis label, query it into a bank, and use the answers to
compute whether the final basis label lies in the accepting set.  No erasure is needed because
no Hadamard layer follows.  This produces the exact acceptance probability.  There are
$3(k-1)+1=3k-2$ rounds after the preprocessing, and padding every batch by dummy index-$0$ queries gives
$t\le2^{q+1}$ throughout.
\end{proof}

\begin{remark}[Why the round simulation fails for standard queries]\label{rem:erasure-obstruction}
We explain here why the one-round simulation of Lemma~\ref{lem:round-embedding} does not extend
directly to standard-query blocks, that is, why a block cannot be replaced by one round of parallel
queries of width comparable to its query count.  A round of $M$ parallel queries queries every
coordinate and coherently stores the full truth table, after which it simulates any block without
further queries, so the statement concerns width.  With
standard queries a block can write an answer, branch on it,
and then erase it with a second query.  Notice that an answer bit discarded by the block is not a function of
the output of the block, so the one-round simulation of Lemma~\ref{lem:round-embedding} does not
extend to it:
the discarded bit either survives, in which case it spoils the interference at the next Hadamard layer,
or it requires a further sequential interaction with the oracle to be removed.  This is the
feature of standard queries that the phase-query round simulation does not capture, and it is why
that simulation does not by itself give the standard-query lower bound.  The
leaf-space factorization of Section~\ref{app:standard-growth} keeps this decision-tree structure
rather than attempting to parallelize it.
\end{remark}

Combining Lemma~\ref{lem:round-embedding} with the Fourier growth theorem of Girish, Sinha, Tal and
Wu, in the form of Theorem~\ref{thm:gstw-growth}, bounds the Fourier weight of the entire
class $\FH_k$.  Here $L_{1,\ell}$ denotes the level-$\ell$ Fourier weight of
Definition~\ref{def:fourier-weight}, and every application below uses only levels $\ell\ge K\ge2$.

\begin{corollary}[Fourier growth of the Fourier hierarchy]\label{cor:FHk-growth}
Let $k\ge2$.  Let $C$ be an $\FH_k$ circuit making $q$ phase queries to oracles with concatenated
truth table length $M$, and let $f:\pmone^{M}\to[0,1]$ be its acceptance probability.  Then for
every $\ell\ge 1$ and every restriction $f_\rho$ of $f$,
\[
L_{1,\ell}(f_\rho)
\le
(2^{2k-2}-1)^{2\ell}\cdot q^{\ell}\cdot
M^{\frac12\left\lfloor\frac{(2k-3)\ell}{2k-2}\right\rfloor}
\le
(2^{2k-2}-1)^{2\ell}\,q^{\ell}\,
M^{\frac{\ell}{2}\left(1-\frac{1}{2k-2}\right)}.
\]
\end{corollary}

\begin{proof}
If $q=0$ then $f$ is oracle-independent, so $L_{1,\ell}(f_\rho)=0$ for every $\ell\ge1$ and the
bound holds.  Assume $q\ge1$.  If the algorithm of Lemma~\ref{lem:round-embedding} makes no
queries at all, which happens when every query of $C$ lies in $U_0$ or $U_k$, then $f$ is again
oracle-independent and the bound holds.  Otherwise the claim is immediate from
Lemma~\ref{lem:round-embedding} and Theorem~\ref{thm:gstw-growth} with $r=k-1\ge1$ and
$1\le t\le\min(q,M)$ (after coherently cancelling
repeated queries via $x_i^2=1$, one round requires at most $M$ distinct parallel queries), using
$t\le q$, $(M/t)^{a}\le M^{a}$ and
$\frac12\lfloor(2r-1)\ell/(2r)\rfloor\le\frac{\ell}{2}(1-\frac1{2r})$.
\end{proof}

Corollary~\ref{cor:FHk-growth} is the bound on the Fourier weight that the lower bound against
$\FH_k$ uses: the $\ell_1$ weight is polynomial at the lowest levels and grows at the rate
$M^{\frac12(1-\frac1{2k-2})}$ per level thereafter.  Notice that a direct bound on the $\ell_1$ weight of the
coefficients of the quadratic form computed by the two layers would not suffice: already at
$k=2$ that weight can be linear in the number of oracle bits.  It is the round simulation that gives
the correct weight bound, and one that is stable under restrictions.

At $k=2$ the round simulation produces a non-adaptive algorithm, for which the sharper growth bound of
Theorem~\ref{thm:nonadaptive-growth} applies.

\begin{corollary}[Fourier growth at the second level]\label{cor:fh2-growth}
Let $C$ be an $\FH_2$ circuit making $q\ge1$ phase queries to oracles with concatenated truth table
length $M$, and let $f:\pmone^{M}\to[0,1]$ be its acceptance probability.  Then for every $\ell\ge1$
and every restriction $f_\rho$ leaving $\widetilde M$ variables free,
\[
L_{1,\ell}(f_\rho)\ \le\ (\ell+1)\,\bigl(\widetilde M\,q\bigr)^{\ell/4}.
\]
\end{corollary}

\begin{proof}
By Lemma~\ref{lem:round-embedding} with $k=2$, $C$ is reproduced exactly by a non-adaptive
algorithm making $t\le q$ parallel queries.  If $t=0$ the acceptance probability is
oracle-independent and $L_{1,\ell}(f_\rho)=0$ for $\ell\ge1$.  Otherwise $1\le t\le q$ and the
acceptance probability has the form required by Theorem~\ref{thm:nonadaptive-growth}; apply that
theorem and bound $\tau=\min\{t,\widetilde M\}\le q$.
\end{proof}

\begin{remark}[The second level, sharpened]\label{rem:fh2-sharpened}
At $k=2$ the query threshold of Theorem~\ref{thm:lower-worst-case} can be improved.  Running Step~4
of its proof with Corollary~\ref{cor:fh2-growth} in place of Corollary~\ref{cor:FHk-growth}, the
level-$\ell$ term contains $q^{\ell/4}N^{\ell/4}$ together with the Bansal--Sinha factor
$N^{-\frac\ell2(1-\frac13)}=N^{-\ell/3}$, so it equals $O(1)\,q^{\ell/4}N^{-\ell/12}$ and decays
whenever $q\le N^{c}$ with $c<\tfrac13$.  Taking $c=\tfrac16$, every term is at most
$N^{-\ell/24}\le N^{-1/8}$, since $\ell\ge K=3$.  Hence Theorem~\ref{thm:lower-worst-case} holds at
$k=2$ with $c_2=\tfrac16$ in place of $c_2=\tfrac1{24}$.  This improvement applies only to the second
level, since Theorem~\ref{thm:nonadaptive-growth} concerns a single round of queries.
\end{remark}

\subsection{A promise problem separating depth \texorpdfstring{$k$}{k} from depth \texorpdfstring{$k+1$}{k+1}}\label{subsec:hard-problem}

We now construct the hard problem.  We use the input distribution of Bansal and
Sinha~\cite{bansal2021k}, which is defined for every order $K\ge2$.  We identify
$\Phi_{2k-1}=\mathsf{forr}_K$ using \eqref{eq:forr-dictionary}, and recall that $\delta=2^{-5K}$.  We
need three of their statements.  The first gives the distributions.  Note that the structured
distribution guarantees a large Forrelation value only with probability of order $\delta$; this is
what makes the construction below necessary.

\begin{theorem}[{Bansal--Sinha~\cite[Section~3, Theorem~3.1]{bansal2021k}}]\label{thm:bs-distributions}
For every $K\ge 2$ there is a distribution $\mathcal F_K$ on $\pmone^{KN}$, the
\emph{structured distribution}, such that, with $\mathcal U$ the uniform distribution on
$\pmone^{KN}$:
\begin{enumerate}[leftmargin=2em]
\item $\displaystyle \Prb_{z\sim\mathcal U}\bigl[\,\abs{\mathsf{forr}_K(z)}\ge\delta/2\,\bigr]\le\frac{4}{\delta^2N}$;
\item $\displaystyle \Prb_{z\sim\mathcal F_K}\bigl[\,\mathsf{forr}_K(z)\ge\delta\,\bigr]\ge 6\delta$.
\end{enumerate}
\end{theorem}

The second statement is the main input to the lower bound.  It bounds the advantage of an
\emph{arbitrary} bounded function in distinguishing $\mathcal F_K$ from the uniform distribution,
purely in terms of its Fourier weight at levels between $K$ and $K(K-1)$, measured with respect to
biased product measures.

\begin{theorem}[{Bansal--Sinha~\cite[Theorem~3.2]{bansal2021k}}]\label{thm:bs-bias}
For every $f:\pmone^{KN}\to[0,1]$,
\[
\bigl|\E_{\mathcal F_K}[f]-\E_{\mathcal U}[f]\bigr|
\le
\sup_{\mu\in[-\frac12,\frac12]^{KN}}\;
\sum_{\ell=K}^{K(K-1)}
\left(\frac{1}{\sqrt N}\right)^{\ell\left(1-\frac1K\right)}
(8K)^{14\ell}\;
\mathsf{wt}^{\mu}_{\ell}(f),
\]
where $\mathsf{wt}^{\mu}_{\ell}$ is the level $\ell$ Fourier weight with respect to the
$\mu$ biased product measure (Definition~\ref{def:fourier-weight}).
\end{theorem}

The bound is stated in this two-sided form in~\cite{bansal2021k}; it is in any case equivalent to
the corresponding one-sided bound, since $f$ and $1-f$ have the same level-$\ell$ Fourier weights
for $\ell\ge1$.

Theorem~\ref{thm:gstw-growth} bounds weight with respect to the uniform measure; the third
statement converts such a bound into one for the biased measures appearing above.
That conversion needs a bound for all restrictions, which is exactly what the last clause of
Theorem~\ref{thm:gstw-growth} provides.

\begin{theorem}[{Bansal--Sinha~\cite[Theorem~3.4]{bansal2021k}}]\label{thm:bs-transfer}
Let $L\ge 1$, let $f:\pmone^{L}\to\mathbb R$, and let $\ell\in[L]$.  If $w$ is such that
$\mathsf{wt}_{\ell}(f_\rho)\le w$ for every restriction $\rho\in\{-1,1,\star\}^{L}$, where the
weight is with respect to the uniform measure, then
$\mathsf{wt}^{\mu}_{\ell}(f)\le 4^{\ell}w$ for every $\mu\in[-\tfrac12,\tfrac12]^{L}$.
\end{theorem}

For a worst-case lower bound we need distributions that are concentrated with sufficiently high
probability on the two sides of a promise, whereas Theorem~\ref{thm:bs-distributions}(2) gives only
mass $6\delta$ on the yes side.  Taking a direct product of constantly many independent copies amplifies this: some
copy has a large Forrelation value with probability exponentially close to one in the number of
copies, and the OR of the corresponding tests is still decidable one level up.

\begin{definition}[The OR of Forrelation instances]\label{def:hard-problem}
Set $m:=\delta^{-1}=2^{5K}$, a constant (an integer since $\delta=2^{-5K}$).  An \emph{instance} at input length $n$ is a
tuple
\[
Z=(z^{(1)},\dots,z^{(m)})\in\bigl(\pmone^{KN}\bigr)^{m},
\]
i.e., a tuple of $mK$ Boolean phase oracles on $n$ bits, with combined truth table length
$M=mKN=O_k(N)$.  The promise sets are
\[
\Pi^{\Yes}_n=\Bigl\{Z:\ \max_{j\in[m]}\mathsf{forr}_K\bigl(z^{(j)}\bigr)\ge\delta\Bigr\},
\qquad
\Pi^{\No}_n=\Bigl\{Z:\ \max_{j\in[m]}\abs{\mathsf{forr}_K\bigl(z^{(j)}\bigr)}\le\delta/2\Bigr\},
\]
and $\mathrm{ORF}_{K,n}$ is the promise problem of deciding, given phase-query access to $Z$
promised to lie in $\Pi^{\Yes}_n\cup\Pi^{\No}_n$, which case holds.  The associated
distributions are
\[
\mathcal Y_n:=\mathcal F_K^{\otimes m},
\qquad
\mathcal N_n:=\mathcal U^{\otimes m}.
\]
\end{definition}

\begin{proposition}[The distributions respect the promise]\label{prop:promise-density}
For all sufficiently large $n$,
\[
\Prb_{Z\sim\mathcal Y_n}\bigl[Z\in\Pi^{\Yes}_n\bigr]\ge 1-e^{-6}\ge 0.99,
\qquad
\Prb_{Z\sim\mathcal N_n}\bigl[Z\in\Pi^{\No}_n\bigr]\ge 1-\frac{4m}{\delta^2N}.
\]
\end{proposition}

\begin{proof}
Under $\mathcal Y_n$ the copies $z^{(j)}$ are independent samples of $\mathcal F_K$, so by
Theorem~\ref{thm:bs-distributions}(2) the probability that no copy has
$\mathsf{forr}_K\ge\delta$ is at most $(1-6\delta)^{m}\le e^{-6\delta m}\le e^{-6}$.  Under
$\mathcal N_n$, a union bound over the $m$ copies and
Theorem~\ref{thm:bs-distributions}(1) give the second claim.
\end{proof}

For the upper bound we run the interference test of
Proposition~\ref{prop:odd-interference-circuit} several times in parallel on each copy, compare the empirical acceptance frequency of each copy with a threshold lying
between the two promised expectations, and accept if some copy exceeds it.  All gates outside the
shared Hadamard layers are basis-preserving, so the Fourier depth does not increase.

\begin{theorem}[Depth $k+1$ decides the problem in the worst case]\label{thm:upper-worst-case}
For every fixed $k\ge2$ there is a uniform polynomial-size $\FH_{k+1}$ circuit family
$B=\{B_n\}$, making $O_k(1)$ phase queries, such that for all sufficiently large $n$:
\[
Z\in\Pi^{\Yes}_n
\ \Longrightarrow\
\Prb[B_n\text{ accepts }Z]\ge\tfrac23,
\qquad
Z\in\Pi^{\No}_n
\ \Longrightarrow\
\Prb[B_n\text{ accepts }Z]\le\tfrac13.
\]
Making $O_k(n)$ phase queries instead, the error bound $\tfrac13$ improves to $2^{-n}$.
\end{theorem}

\begin{proof}
For each copy index $j\in[m]$, run $R$ independent parallel repetitions of the interference test
of Proposition~\ref{prop:odd-interference-circuit} on the instance $z^{(j)}$, where $R=O_k(1)$ is
fixed below.  All $mR$ repetitions share the same $k+1$ global Hadamard layers by
Remark~\ref{rem:register-routing}, and all gates outside those layers are basis-preserving, so
the combined circuit is a uniform $\FH_{k+1}$ circuit making $mRk$ phase queries.  All
control qubits are measured; the accepting set is
\[
S=\Bigl\{\text{outcomes}:\ \exists j\in[m]\ \text{with}\
\#\{\text{accepting repetitions for copy }j\}\ge \tfrac{R}{2}\bigl(1+\tfrac34\delta\bigr)\Bigr\},
\]
recognizable in polynomial time.

Fix a copy $j$ and condition on $Z$.  Each repetition for copy $j$ accepts independently with
probability $p_j=\frac12(1+\mathsf{forr}_K(z^{(j)}))$.  If
$\mathsf{forr}_K(z^{(j)})\ge\delta$ then $p_j\ge\frac12(1+\delta)$, which exceeds the
threshold fraction by $\frac{\delta}{8}$; if
$\abs{\mathsf{forr}_K(z^{(j)})}\le\delta/2$ then $p_j\le\frac12(1+\delta/2)$, which falls
below it by $\frac{\delta}{8}$.  By Hoeffding's inequality
(Lemma~\ref{lem:hoeffding}), the vote for copy $j$ fails with probability at most
$\exp(-2R(\delta/8)^2)$.  Choosing $R=\lceil 32\,\delta^{-2}\ln(6m)\rceil=O_k(1)$ makes
this at most $\frac{1}{6m}$.

On $Z\in\Pi^{\Yes}_n$ some copy $j^*$ has $\mathsf{forr}_K(z^{(j^*)})\ge\delta$, and the
circuit accepts unless the vote for $j^*$ fails, so it accepts with probability at least
$1-\frac{1}{6m}\ge\frac23$.  On $Z\in\Pi^{\No}_n$ every copy satisfies
$\abs{\mathsf{forr}_K}\le\delta/2$, and by a union bound the circuit accepts with probability
at most $m\cdot\frac{1}{6m}=\frac16\le\frac13$.  For the second statement take
$R=\lceil C\delta^{-2}\,n\rceil$ with a suitable absolute constant $C$, which drives both
error bounds below $2^{-n}$.
\end{proof}

We now prove the lower bound, in four steps.  A circuit that decides the promise problem on every
instance must distinguish the two product distributions with constant advantage, since each
distribution is supported almost entirely inside its promise set.  A hybrid argument then gives a
single copy accounting for a constant fraction of that advantage, with all the other copies fixed;
fixing the other copies is permissible because the growth bound of Theorem~\ref{thm:gstw-growth} is
stable under restrictions.  The restricted circuit is a bounded function of one instance, and its biased
Fourier weight is bounded using the round simulation and the transfer theorem.  The advantage
criterion of Theorem~\ref{thm:bs-bias} then reduces the argument to a comparison of two exponents,
and the comparison is strict because $K$ is larger than twice the number of rounds.

\begin{theorem}[Depth $k$ cannot decide the problem]\label{thm:lower-worst-case}
Fix $k\ge2$ and set $c_k:=\frac{1}{4(2k-1)(2k-2)}$.
There is $n_0=n_0(k)$ such that for every $n\ge n_0$ the following holds.  If $C$ is any
$\FH_k$ circuit, of arbitrary size and not necessarily uniform, making $q\le N^{c_k}$ phase queries, then $C$ fails to decide
$\mathrm{ORF}_{K,n}$ with error $\le\frac13$: there exists an instance
$Z\in\Pi^{\Yes}_n\cup\Pi^{\No}_n$ on which $C$ fails with probability greater than $\frac13$.
\end{theorem}

\begin{proof}
Suppose toward a contradiction that $C$ accepts every $Z\in\Pi^{\Yes}_n$ with probability at
least $\frac23$ and every $Z\in\Pi^{\No}_n$ with probability at most $\frac13$.  Let
$g:\pmone^{mKN}\to[0,1]$ be the acceptance probability of $C$ as a function of the
concatenated truth tables.

\emph{Step 1: distributional advantage.}
By Proposition~\ref{prop:promise-density},
\[
\E_{\mathcal Y_n}[g]
\ \ge\
\tfrac23\,\Prb_{\mathcal Y_n}[\Pi^{\Yes}_n]
\ \ge\ \tfrac23\cdot 0.99
\ \ge\ 0.65,
\qquad
\E_{\mathcal N_n}[g]
\ \le\
\tfrac13+\Prb_{\mathcal N_n}\bigl[\neg\Pi^{\No}_n\bigr]
\ \le\ \tfrac13+\frac{4m}{\delta^2N},
\]
so for $n$ large,
\begin{equation}\label{eq:distributional-advantage}
\E_{\mathcal Y_n}[g]-\E_{\mathcal N_n}[g]\ \ge\ \tfrac14.
\end{equation}

\emph{Step 2: hybrid argument.}
For $0\le j\le m$ let $H_j:=\mathcal F_K^{\otimes j}\otimes\,\mathcal U^{\otimes(m-j)}$, so
$H_0=\mathcal N_n$ and $H_m=\mathcal Y_n$.  By \eqref{eq:distributional-advantage} there is an
index $j$ with
\[
\bigl|\E_{H_j}[g]-\E_{H_{j-1}}[g]\bigr|\ \ge\ \frac{1}{4m}.
\]
The two hybrids differ only in copy $j$, and the remaining copies are drawn from the same
distribution under both.  Writing $\Delta(w)$ for the difference of the two
inner expectations at a fixed value $w$ of the other copies, the displayed bound says
$\abs{\E_w[\Delta(w)]}\ge1/(4m)$, and $\abs{\E_w[\Delta(w)]}\le\max_w\abs{\Delta(w)}$.  Hence there
exist \emph{fixed} values
$w=(z^{(1)},\dots,z^{(j-1)},z^{(j+1)},\dots,z^{(m)})$ of the other copies such that the
restricted function
\[
h:\pmone^{KN}\to[0,1],
\qquad
h(z):=g\bigl(z^{(1)},\dots,z^{(j-1)},\,z,\,z^{(j+1)},\dots,z^{(m)}\bigr)
\]
satisfies
\begin{equation}\label{eq:one-block-advantage}
\bigl|\E_{z\sim\mathcal F_K}[h(z)]-\E_{z\sim\mathcal U}[h(z)]\bigr|\ \ge\ \frac{1}{4m}.
\end{equation}

\emph{Step 3: Fourier weight of the restricted function.}
The function $h$ is a restriction of $g$, and so is every further restriction $h_\rho$.  By
Corollary~\ref{cor:FHk-growth}, whose bound holds for all restrictions, for every
$\ell\ge1$ and every restriction $\rho$,
\[
\mathsf{wt}_{\ell}(h_\rho)
=
L_{1,\ell}(h_\rho)
\le
(2^{2k-2}-1)^{2\ell}\,q^{\ell}\,M^{\frac{\ell}{2}\left(1-\frac{1}{2k-2}\right)},
\qquad M=mKN.
\]
By Theorem~\ref{thm:bs-transfer}, for every bias vector $\mu\in[-\frac12,\frac12]^{KN}$,
\[
\mathsf{wt}^{\mu}_{\ell}(h)
\le
4^{\ell}\cdot(2^{2k-2}-1)^{2\ell}\,q^{\ell}\,M^{\frac{\ell}{2}\left(1-\frac{1}{2k-2}\right)}
=:
A_{\ell}.
\]

\emph{Step 4: the advantage criterion.}
The remaining task is to show that the right-hand side of Theorem~\ref{thm:bs-bias}, applied to
$h$, is $o_k(1)$; this is where the exponent gap $2(k-1)<K=2k-1$ is used.
By Theorem~\ref{thm:bs-bias} applied to $h$,
\[
\bigl|\E_{\mathcal F_K}[h]-\E_{\mathcal U}[h]\bigr|
\le
\sum_{\ell=K}^{K(K-1)}
N^{-\frac{\ell}{2}\left(1-\frac1K\right)}\,(8K)^{14\ell}\,A_{\ell}.
\]
Since $M=mKN$ with $m,K=O_k(1)$, and since $\ell\le K(K-1)=O_k(1)$ throughout the sum, the
quantities $(2^{2k-2}-1)^{2\ell}$, $(8K)^{14\ell}$, $4^{\ell}$, and $(mK)^{\ell}$ are all
$O_k(1)$ and may be absorbed into a constant $C_{k}$:
\[
\bigl|\E_{\mathcal F_K}[h]-\E_{\mathcal U}[h]\bigr|
\le
C_k\sum_{\ell=K}^{K(K-1)}
q^{\ell}\;
N^{\frac{\ell}{2}\left(1-\frac{1}{2k-2}\right)-\frac{\ell}{2}\left(1-\frac{1}{2k-1}\right)}
=
C_k\sum_{\ell=K}^{K(K-1)}
q^{\ell}\,N^{-\frac{\ell}{2(2k-1)(2k-2)}},
\]
using $\frac12(1-\frac1{2k-2})-\frac12(1-\frac1{2k-1})
=-\frac{1}{2(2k-1)(2k-2)}$.
With $q\le N^{c_k}$ and $c_k=\frac{1}{4(2k-1)(2k-2)}$, each term is at most
$N^{-\ell c_k}\le N^{-Kc_k}$, so
\[
\bigl|\E_{\mathcal F_K}[h]-\E_{\mathcal U}[h]\bigr|
\le
C_k'\,N^{-\frac{2k-1}{4(2k-1)(2k-2)}}
=
C_k'\,N^{-\frac{1}{4(2k-2)}}.
\]
For $n\ge n_0(k)$ this is smaller than $\frac{1}{4m}$, contradicting
\eqref{eq:one-block-advantage}.
\end{proof}

The proof in fact bounds the distinguishing advantage of \emph{every} depth-$k$ circuit, and not only
that of a hypothetical decider, so the separation holds at vanishing advantage: depth $k$ has distinguishing
advantage $o_k(1)$ in the stated range of query counts.

\begin{corollary}[The lower bound holds at vanishing advantage]\label{cor:error-robust}
Fix $k\ge2$ with $c_k$ as in Theorem~\ref{thm:lower-worst-case}.  There is $n_0(k)$ such that for
every $n\ge n_0(k)$, every $\FH_k$ circuit $C$ making $q\le N^{c_k}$ phase queries, with acceptance
probability $p_C$, satisfies
\[
\bigl|\E_{\mathcal Y_n}[p_C]-\E_{\mathcal N_n}[p_C]\bigr|
\ \le\
C_k\,N^{-\frac{1}{4(2k-2)}}
\ =\ o_k(1).
\]
The structured and uniform ensembles are therefore $o_k(1)$-indistinguishable at Fourier depth $k$
with $N^{c_k}$ queries.  No such circuit decides $\mathrm{ORF}_{K,n}$ with worst-case error below
$\tfrac{1-e^{-6}}{2-e^{-6}}-o_k(1)$, a constant just below $\tfrac12$, whereas the depth-$(k+1)$
circuit of Theorem~\ref{thm:upper-worst-case} fails with probability $2^{-n}$.
\end{corollary}

\begin{proof}
Run the proof of Theorem~\ref{thm:lower-worst-case} without assuming a decider.  Steps~3--4 bound,
for \emph{every} fixing of the other $m-1$ copies, the single-block advantage
$\bigl|\E_{\mathcal F_K}[h]-\E_{\mathcal U}[h]\bigr|\le C_k'\,N^{-1/(4(2k-2))}$, uniformly over
restrictions.  By the hybrid decomposition of Step~2, the full advantage
$\bigl|\E_{\mathcal Y_n}[p_C]-\E_{\mathcal N_n}[p_C]\bigr|$ is a sum of $m$ such single-block
advantages, hence at most $m\,C_k'\,N^{-1/(4(2k-2))}$ with $m=2^{5K}=O_k(1)$, which is the bound
above.  Finally, a circuit with worst-case error $\varepsilon$ accepts every $\Yes$ instance with
probability $\ge1-\varepsilon$ and every $\No$ instance with probability $\le\varepsilon$.  By
Proposition~\ref{prop:promise-density}, $\E_{\mathcal Y_n}[p_C]\ge(1-\varepsilon)(1-e^{-6})$ and
$\E_{\mathcal N_n}[p_C]\le\varepsilon+o_k(1)$, so the advantage is at least
$(1-e^{-6})-(2-e^{-6})\varepsilon-o_k(1)$; the bound above then forces
$\varepsilon\ge\tfrac{1-e^{-6}}{2-e^{-6}}-o_k(1)$.
\end{proof}

The direct product in Definition~\ref{def:hard-problem} converts the $\Theta(\delta)$ mass guarantee
of Theorem~\ref{thm:bs-distributions}(2) into the high-probability guarantee of
Proposition~\ref{prop:promise-density}, and a worst-case solver contradicts that guarantee in Step~1.  The
hybrid argument of Step~2 then reduces the analysis to a single Forrelation instance, where
Theorems~\ref{thm:bs-bias} and~\ref{thm:bs-transfer} apply, and the clause of
Theorem~\ref{thm:gstw-growth} allowing input bits to be fixed in advance is exactly what Step~3
needs.

\begin{remark}[Parametric form of the lower bound]\label{rem:parametric-lower-bound}
The proof of Theorem~\ref{thm:lower-worst-case} used only two properties of the circuit: that
it is simulated exactly by an algorithm with $r=k-1$ rounds of $t=q$ parallel queries
(Lemma~\ref{lem:round-embedding}), and that $2r<K$.  The same proof, run verbatim with general parameters,
shows: for every $K\ge2$ and $r\ge1$ with $2r<K$, no quantum query algorithm with $r$ adaptive
rounds and
\[
t\ \le\ N^{c(r,K)},
\qquad
c(r,K):=\frac14\Bigl(\frac{1}{2r}-\frac1K\Bigr)>0,
\]
parallel queries per round decides the OR problem built from $K$-fold Forrelation instances
(Definition~\ref{def:hard-problem} with $\delta=2^{-5K}$ and $m=2^{5K}$) with error at most
$\frac13$ on every instance of the promise, for all $n\ge n_0(K)$.
Theorem~\ref{thm:lower-worst-case} is the case $r=k-1$, $K=2k-1$, $t=q$, for which $c(r,K)=c_k$.

Three exponents should be kept apart here.  The first is the value stated in
Theorem~\ref{thm:lower-worst-case}, namely $c_k=\frac{1}{4(2k-1)(2k-2)}$, chosen for a convenient
margin.  The second is the supremum that the present proof gives: each term of the sum in Step~4
carries the power $N^{-\ell\left(\frac12\left(\frac1{2r}-\frac1K\right)-c\right)}$, so the sum decays
for every $c<\tfrac12\bigl(\tfrac1{2r}-\tfrac1K\bigr)$, which is twice $c(r,K)$; at $k=2$ the
sharpened growth bound of Remark~\ref{rem:fh2-sharpened} raises this to every $c<\tfrac13$, of which
the stated $c_2=\tfrac16$ is again half.  The third is the optimal exponent, which we do not
determine.  What limits the second is the Bansal--Sinha criterion together with the growth bound, not
the construction; improving beyond it requires a better growth bound, or a criterion that does not
pass through Fourier weight.
\end{remark}

\begin{corollary}[Growing depth]\label{cor:growing-k}
There are absolute constants $c,c'>0$ and $n_0$ with the following two properties.
\begin{enumerate}[leftmargin=2em]
\item For every $n\ge n_0$ and every $k$ with $2\le k\le c\,n^{1/3}$, the promise problem
$\mathrm{ORF}_{2k-1,n}$, whose $mK=2^{O(k)}$ functions are indexed by $O(k)$ bits, is decided at
Fourier depth $k+1$ by a circuit of size $2^{O(k)}\poly(n)$ making $2^{O(k)}$ phase queries, while
every circuit of Fourier depth $k$ making at most $2^{c'n/k^2}$ phase queries fails on some
promised instance.
\item Let $k(\cdot)$ be nondecreasing and polynomial-time computable with $2\le k(n)\le c\log n$
for all sufficiently large $n$.  Then the index width is $O(\log n)$ and the deciding circuit has
polynomial size at those lengths, so the diagonalization of
Theorem~\ref{thm:all-level-separation}, run over those lengths and with the answers at the
finitely many remaining lengths hardwired into the deciding family, gives an oracle relative to
which Fourier depth $k(n)+1$ decides a language that no uniform family of Fourier depth $k(n)$
decides (Definition~\ref{def:variable-depth}); for $k(n)=o(\log n)$ the deciding circuit makes
$n^{o(1)}$ queries.
\end{enumerate}
\end{corollary}

\begin{proof}
We track the dependence on $k$ in the argument for constant $k$.  The upper bound depends only mildly
on $k$: with
$\delta=2^{-5K}$, $m=2^{5K}=2^{\Theta(k)}$, and $R=O(\delta^{-2}\log m)=2^{\Theta(k)}$, the
depth-$(k+1)$ circuit of Theorem~\ref{thm:upper-worst-case} makes $mRk=2^{\Theta(k)}$ queries,
and its size is $2^{\Theta(k)}\poly(n)$; the instance uses $mK=2^{\Theta(k)}$ functions, so its index width
is $\Theta(k)$ bits.  The width and the size are polynomial for $k\le c\log n$ with $c$ small, and the
query count is $n^{o(1)}$ for $k=o(\log n)$.

For the lower bound, Step~4 of Theorem~\ref{thm:lower-worst-case} bounds the advantage by a sum of at
most $K(K-1)=\Theta(k^2)$ terms, the term at level $\ell$ being a product of
$N^{-\ell/(2(2k-1)(2k-2))}$ with the factors $4^\ell$, $(8K)^{14\ell}$, $q^\ell$, $(mK)^\ell$ (from
$M=mKN$), and $(2^{2k-2}-1)^{2\ell}$.  Every factor is now singly exponential in $k$ per unit of
$\ell$: with $m=2^{\Theta(k)}$ and $K=\Theta(k)$,
\[
 4^{\ell}\,(8K)^{14\ell}\,(mK)^{\ell}\,(2^{2k-2}-1)^{2\ell}=2^{O(k)\cdot\ell},
 \qquad
 q^{\ell}\,N^{-\frac{\ell}{2(2k-1)(2k-2)}}\le N^{-\frac{\ell}{4(2k-1)(2k-2)}},
\]
where the second inequality uses $q\le N^{c_k}$ with $c_k=\tfrac1{4(2k-1)(2k-2)}$.  The term at
level $\ell$ is therefore at most $2^{\ell(O(k)-n/(4(2k-1)(2k-2)))}$, and once $k^3\le c''n$ for a
suitable absolute constant $c''$, the exponent is at most $-\ell n/(8(2k-1)(2k-2))$.  Summing the
geometric series from $\ell=K$ gives advantage at most $2^{-\Omega(n/k)}$, which is below the
threshold $\Theta(1/m)=2^{-\Theta(k)}$ of Steps~1--2 because $n/k\ge n^{2/3}\gg k$ in the stated
range.  Thus every depth-$k$ circuit with $2^{c'n/k^2}$ queries fails on the promise.  For
$k(n)\le c\log n$ the instances fit the convention of Definition~\ref{def:FHk-rel} and the decider is
a uniform polynomial-size family, so the diagonalization of Theorem~\ref{thm:all-level-separation},
run against the enumeration of the families obeying the schedule
(Definition~\ref{def:variable-depth}), separates the two depths.
\end{proof}

\subsection{The separation theorem}\label{subsec:separation-theorem}

It remains to combine the worst-case statements, one for each input length, into a single oracle.
This is a standard diagonalization: we enumerate the uniform depth-$k$ machines, use
Theorem~\ref{thm:lower-worst-case} to make the $i$th machine fail at a dedicated input length, and let
the depth-$(k+1)$ circuit of Theorem~\ref{thm:upper-worst-case}, which is correct on every instance
of the promise, decide the resulting language.

\begin{theorem}[Oracle separation at every level]\label{thm:all-level-separation}
For every constant $k\ge 2$ there exists an oracle family $O$ such that
\[
\FH_k^{O}\ \subsetneq\ \FH_{k+1}^{O}
\]
in the phase-query Fourier hierarchy of Definition~\ref{def:FHk-rel}.
\end{theorem}

\begin{proof}
\emph{The oracle and the language.}
At each input length $n$, an oracle instance is a tuple $Z_n$ of $mK$ Boolean functions on $n$
bits, accessed by phase queries as in Definition~\ref{def:hard-problem}; the oracle family is
$O=\{Z_n\}_{n\ge1}$, and we define the unary language
\[
L_O=\bigl\{1^n:\ Z_n\in\Pi^{\Yes}_n\bigr\}.
\]
Choose each $Z_n$ in $\Pi^{\Yes}_n\cup\Pi^{\No}_n$.  Note that
$\Pi^{\Yes}_n\ne\emptyset$ at \emph{every} length: for the all-ones tuple, every copy satisfies
$\mathsf{forr}_K=\frac1N\mathbf 1^{\top}\mathsf H^{K-1}\mathbf 1=\frac1N\mathbf 1^{\top}\mathbf 1=1$
by \eqref{eq:forr-dictionary}, since $\mathsf H^{2}=I$ and $K-1$ is even; and
$\Pi^{\No}_n\ne\emptyset$ for all sufficiently large $n$ by
Proposition~\ref{prop:promise-density}.

\emph{An effective enumeration of the depth-$k$ families.}
Enumerate as $(A^{(i)})_{i\ge1}$ all uniform polynomial-size $\FH_k$ oracle circuit families, that
is, all clocked generators $(G_i,c_i)$ where $G_i$ is a Turing machine truncated after $n^{c_i}$
steps.  An output that is malformed or has more than $k$ Hadamard layers is replaced by a fixed
rejecting circuit, so every uniform $\FH_k$ family occurs in the enumeration.  This is an effective
enumeration with explicit query bounds $q_i(n)\le n^{c_i}$, so that each $q_i$ is polynomial in $n$.  Choose an increasing sequence of lengths $n_1<n_2<\cdots$ such that for
each $i$:
$n_i\ge n_0(k)$ from Theorem~\ref{thm:lower-worst-case}, and $q_i(n_i)\le 2^{c_kn_i}$, which
holds for all sufficiently large $n_i$ since $q_i$ is polynomial.

\emph{Defeating the $i$th family at its own length.}
Fix $i$.  By Theorem~\ref{thm:lower-worst-case} applied to $A^{(i)}_{n_i}$, there is an
instance $Z_{n_i}\in\Pi^{\Yes}_{n_i}\cup\Pi^{\No}_{n_i}$ on which $A^{(i)}_{n_i}$ fails with
probability greater than $\frac13$: either $Z_{n_i}\in\Pi^{\Yes}_{n_i}$ and
$\Prb[A^{(i)}_{n_i}\text{ accepts}]<\frac23$, or $Z_{n_i}\in\Pi^{\No}_{n_i}$ and
$\Prb[A^{(i)}_{n_i}\text{ accepts}]>\frac13$.  Fix such a choice.  At every length
$n\notin\{n_1,n_2,\dots\}$, fix $Z_n$ to be an arbitrary instance of $\Pi^{\No}_n$, and at the
finitely many remaining lengths where $\Pi^{\No}_n$ is empty, fix $Z_n$ to be the all ones
instance of $\Pi^{\Yes}_n$.

\emph{The two containments.}
By construction, no $\FH_k$ family decides $L_O$ with bounded error: family $A^{(i)}$ fails at
length $n_i$.  Hence $L_O\notin\FH_k^{O}$.  On the other hand, every $Z_n$ lies in the promise
$\Pi^{\Yes}_n\cup\Pi^{\No}_n$, and the circuit family $B$ of
Theorem~\ref{thm:upper-worst-case} decides membership of $1^n$ in $L_O$ with error at most
$\frac13$ at every length $n\ge n_B(k)$, where $n_B(k)$ is the threshold beyond which that theorem
applies.  For the finitely many lengths $n<n_B(k)$ we hardwire the correct membership bit for
$1^n\in L_O$ into the uniform family $B$; since this alters $B$ on only constantly many lengths it remains a
uniform polynomial-size $\FH_{k+1}$ family, now correct at every length.  Therefore
$L_O\in\FH_{k+1}^{O}$.
\end{proof}

The construction above uses a different oracle for each $k$.  In fact a single oracle makes every
level strict simultaneously.

\begin{corollary}[One oracle strict at every level]\label{cor:single-oracle-all-levels}
There is a single oracle $O$ for which
\[
\FH_k^{O}\ \subsetneq\ \FH_{k+1}^{O}\qquad\text{for every constant }k\ge2
\]
simultaneously, in the phase-query Fourier hierarchy.
\end{corollary}

\begin{proof}
Split the input lengths into infinitely many infinite classes $\{S_k\}_{k\ge2}$, each decidable in
polynomial time; for concreteness take $S_k=\{n:\nu_2(n)=k-2\}$, where $\nu_2$ is the $2$-adic
valuation.  On the lengths of $S_k$ run the level-$k$ construction of
Definition~\ref{def:hard-problem}, choosing each instance in $\Pi^{\Yes}_n\cup\Pi^{\No}_n$ as in the
proof above; this fixes one oracle $O$.  For each $k$ set
$L_k=\{1^n:n\in S_k,\ Z_n\in\Pi^{\Yes}_n\}$.

Fix $k$.  By the length-preserving convention of Definition~\ref{def:FHk-rel}, a circuit on input
$1^n$ with $n\in S_k$ queries only $O_n$, which is a level-$k$ instance; the diagonalization in the
proof of Theorem~\ref{thm:all-level-separation}, restricted to the lengths of $S_k$, therefore
defeats every $\FH_k$ family on $L_k$ through Theorem~\ref{thm:lower-worst-case}, giving
$L_k\notin\FH_k^{O}$.  A depth-$(k+1)$ circuit decides $L_k$ by testing $n\in S_k$ classically and
then running the level-$k$ test of Theorem~\ref{thm:upper-worst-case}, so $L_k\in\FH_{k+1}^{O}$.
As $k$ was arbitrary, $O$ separates every adjacent level.
\end{proof}

\begin{proposition}[Uniformly efficient oracles compile away]
\label{prop:uniform-oracle-compiles}
If a fixed oracle family $O$ has a deterministic polynomial-time evaluator
$E(1^n,a,x)=f_{n,a}(x)$, then for every $k\ge0$,
\[
\FH_k^O=\FH_k,
\qquad
\FH_{k,\mathrm{std}}^O=\FH_{k,\mathrm{std}}.
\]
\end{proposition}

\begin{proof}
Compile the deterministic evaluator into a polynomial-size reversible circuit over
\textsc{not}, \textsc{cnot}, and Toffoli gates.  A phase query is replaced by computing
$f_{n,a}(x)$ into a clean ancilla, applying $Z$ to that ancilla, and uncomputing.  A standard query is
replaced by computing the value, copying it into the target, and uncomputing.  Every gate in either
replacement is basis-preserving, so the replacement lies inside the same block $U_i$ and adds no
Hadamard layer.  This proves the nontrivial inclusions from left to right; the reverse inclusions use
no oracle queries.
\end{proof}

\begin{remark}[On efficiently computable oracles]\label{rem:efficient-oracle}
We distinguish three notions here.  First, Proposition~\ref{prop:uniform-oracle-compiles} shows that
a fixed oracle computable in uniform polynomial time cannot give a relativized separation unless the
corresponding unrelativized levels already differ.

Second, an oracle may be \emph{succinct at each length}, meaning that its functions have
polynomial-size circuits containing a length-dependent secret key, while algorithms are given only
black-box query access.  This is the sense in which the cryptographic instantiation of Girish and
Servedio~\cite{girish2025forrelationextremallyhard} is efficient; their lower bound is stated in the
truth-table query model and does not give the algorithm the circuit descriptions.  Their local
form
\[
f(x)=\langle A_{\mathrm{up}}x,A_{\mathrm{dn}}x\rangle+\langle x,a\rangle+h(A_{\mathrm{dn}}x)
\]
permits each derived-oracle query to be simulated by one query to $h$, so replacing $h$ by a
pseudorandom function is justified by a standard hybrid.  In our setting the hybrid would require a
pseudorandom function secure against superposition queries~\cite{zhandry2012}, which the
Goldreich--Goldwasser--Micali construction~\cite{GGM86} provides from any post-quantum one-way
function~\cite{HILL99}.  A growing hidden key
is nonuniform advice, however, so the resulting fixed oracle is succinct in $\mathrm P/\poly$,
not uniformly computable from $(n,a,x)$.

Third, in a \emph{white-box} problem the descriptions of the evaluator circuits are part of the
input.  Ordinary PRF security gives no guarantee once these descriptions reveal their keys and affine
maps.  Such a construction would give a conditional separation $\FH_2\ne\FH_3$ in the unrelativized
setting, not an efficient relativized separation.

Neither of the hard distributions used in the present separation has the local structure required for
the second notion: the Bansal--Sinha distribution~\cite{bansal2021k} rounds correlated Gaussians
globally, and the sign ensemble of Appendix~\ref{sec:appendix-sign} uses
$F_1=U\cdot V$, the product of the tie-broken signs of the two Hadamard transforms
(Definition~\ref{def:yes-no-ensembles}), whose value at one point depends on an entire transform.  Efficiently computable hard instances are left open.
\end{remark}

We state the case $k=2$ explicitly.  A candidate ensemble for this case, whose $\FH_2$-hardness we
reduce to a single conjecture on the Fourier coefficients of a product of threshold functions, is presented in
Appendix~\ref{sec:appendix-sign}.

\begin{corollary}[$\FH_2$ versus $\FH_3$]\label{cor:fh2-fh3-separation}
There is an oracle $O$ with $\FH_2^{O}\subsetneq\FH_3^{O}$.  Concretely, the OR of $2^{15}$
threefold Forrelation instances is decided at Fourier depth $3$ by parallel interference tests, while
every circuit of Fourier depth $2$ making at most $2^{n/6}$ phase queries fails on some
instance of the promise.  By Lemma~\ref{lem:round-embedding}, computations of Fourier depth $2$
are simulated exactly by \emph{non-adaptive} quantum query algorithms, so this exhibits the same
rounds-versus-parallel-queries phenomenon studied by~\cite{girish2024power}, established here via
their Fourier-growth theorem applied to threefold Forrelation and the criterion
of~\cite{bansal2021k}: the third Hadamard layer provides the additional round represented by the simulation.
\end{corollary}

\begin{proof}
This is Theorem~\ref{thm:all-level-separation} with $k=2$, $K=3$ and $m=2^{15}$.  The general
threshold of Theorem~\ref{thm:lower-worst-case} is $c_2=\frac{1}{4\cdot3\cdot2}=\frac{1}{24}$;
at this level it improves to $c_2=\frac16$ by Remark~\ref{rem:fh2-sharpened}, and the
diagonalization is unchanged.
\end{proof}

\subsection{Even Forrelation at Fourier depth \texorpdfstring{$k$}{k}}\label{subsec:even-forrelation}

Remark~\ref{rem:why-odd} explains why the separating problem uses odd-order Forrelation.  We now show
that the choice is also tight in the other direction: the next lower order, which is even, is
decided at Fourier depth $k$, one layer below, and therefore cannot separate $\FH_k$ from
$\FH_{k+1}$.  Define the $(2k-2)$-fold Forrelation function
\begin{equation}\label{eq:even-Phi}
\Phi_{2k-2}(F_0,\dots,F_{2k-3})
=
N^{-(2k-1)/2}
\sum_{x_0,\dots,x_{2k-3}\in\bits^n}
\Bigl(\prod_{j=0}^{2k-3}F_j(x_j)\Bigr)
\Bigl(\prod_{j=0}^{2k-4}(-1)^{x_j\cdot x_{j+1}}\Bigr),
\end{equation}
again with $\Phi_{2k-2}\in[-1,1]$ by \eqref{eq:forr-dictionary}.

\begin{proposition}[Even Forrelation at Fourier depth $k$]\label{prop:even-forrelation}
For every fixed $k\ge2$ there is a uniform polynomial-size $\FH_k$ circuit family making $2k-2$
phase queries whose acceptance probability on every oracle tuple $(F_0,\dots,F_{2k-3})$ equals
\[
\frac{1+\Phi_{2k-2}(F_0,\dots,F_{2k-3})^2}{2}.
\]
\end{proposition}

\begin{proof}
\emph{The circuit.}
It uses two $n$-qubit registers $A$ and $B$ and $k$ global Hadamard layers acting on all
$2n$ wires.  Between consecutive layers it performs only phase queries: in the $j$th interior
block ($1\le j\le k-1$) it applies $P_{F_{j-1}}$ to $A$ and $P_{F_{2k-2-j}}$ to $B$, so that $A$
queries $F_0,\dots,F_{k-2}$ in order and $B$ queries $F_{2k-3},F_{2k-4},\dots,F_{k-1}$ in reverse.
In the last interior block it also applies the coupling
$\ket a_A\ket b_B\mapsto(-1)^{a\cdot b}\ket a_A\ket b_B$, realized by $n$ Toffoli gates with controls
$A_i,B_i$ and target an ancilla prepared in $\ket-$ by the first Hadamard layer and routed past the
later layers (Remark~\ref{rem:register-routing}); equivalently, by the diagonal gates
$\prod_{i=1}^n\mathrm{CZ}(A_i,B_i)$.  After the $k$th layer all $2n$
wires are measured, and the circuit accepts iff $s_A\cdot s_B=0$, where $s_A,s_B\in\bits^n$ are the
two register outcomes.  Phase queries are diagonal and Toffoli gates permute the computational basis, so all gates
between the layers are basis-preserving and the circuit is a uniform $\FH_k$ circuit making $2(k-1)=2k-2$ phase
queries, and its accepting set is decidable in polynomial time.

\emph{The two register states.}
Let $\ket A=\sum_u\alpha(u)\ket u$ and $\ket B=\sum_v\beta(v)\ket v$ be the states of the two
registers after the first $k-1$ layers and all $2k-2$ queries, but before the coupling and the final
layer:
\[
\ket A=P_{F_{k-2}}H^{\otimes n}\cdots P_{F_0}H^{\otimes n}\ket{0^n},
\qquad
\ket B=P_{F_{k-1}}H^{\otimes n}\cdots P_{F_{2k-3}}H^{\otimes n}\ket{0^n},
\]
with real amplitudes $\alpha,\beta$.  The coupling multiplies the joint amplitude of
$\ket u_A\ket v_B$ by $(-1)^{u\cdot v}$, and the final Hadamard layer then gives, on outcome
$(s_A,s_B)$,
\[
\mathcal A(s_A,s_B)=N^{-1}\sum_{u,v}\alpha(u)\beta(v)(-1)^{u\cdot v}(-1)^{u\cdot s_A+v\cdot s_B}.
\]
\emph{The accepting parity supplies the missing interference.}
The circuit accepts exactly on the outcomes with $s_A\cdot s_B=0$, and for a bit
$b\in\{0,1\}$ we have $\mathbf 1[b=0]=\tfrac12(1+(-1)^b)$, so the acceptance probability is
$\tfrac12\bigl(1+X\bigr)$ with $X=\sum_{s_A,s_B}(-1)^{s_A\cdot s_B}\abs{\mathcal A(s_A,s_B)}^2$, using
$\sum_{s_A,s_B}\abs{\mathcal A}^2=1$.
Expanding $\abs{\mathcal A}^2$ over pairs of branches $(u,v),(u',v')$ and summing the outcomes with
\[
\sum_{s_A,s_B\in\bits^n}(-1)^{s_A\cdot s_B+d_A\cdot s_A+d_B\cdot s_B}=N\,(-1)^{d_A\cdot d_B},
\qquad d_A=u\oplus u',\quad d_B=v\oplus v'
\]
(sum over $s_B$ first, forcing $s_A=d_B$, then over $s_A$) yields
\[
X=N^{-1}\sum_{u,v,u',v'}\alpha(u)\alpha(u')\beta(v)\beta(v')\,
(-1)^{u\cdot v+u'\cdot v'+(u\oplus u')\cdot(v\oplus v')}.
\]
Modulo $2$, $u\cdot v+u'\cdot v'+(u\oplus u')\cdot(v\oplus v')=u\cdot v'+u'\cdot v$, so the summand
factors across the disjoint pairs $\{u,v'\}$ and $\{u',v\}$:
\[
X=N^{-1}\Bigl(\sum_{u,v}\alpha(u)\beta(v)(-1)^{u\cdot v}\Bigr)^{2}=N^{-1}\Sigma^2.
\]
The coupling $(-1)^{u\cdot v}$ joins the $A$-chain $x_0\!-\!\cdots\!-\!x_{k-2}\!=\!u$ to the reversed
$B$-chain $v\!=\!x_{k-1}\!-\!\cdots\!-\!x_{2k-3}$ into one chain through the $2k-2$ oracles
$F_0,\dots,F_{2k-3}$ in order, so $\Sigma=N^{1/2}\Phi_{2k-2}(F_0,\dots,F_{2k-3})$ by
\eqref{eq:even-Phi}.  Hence $X=\Phi_{2k-2}^2$.
\end{proof}

\begin{remark}[The inner-product accepting set and tightness of the Forrelation order]\label{rem:even-tightness}
This circuit replaces the control-qubit SWAP test of Proposition~\ref{prop:odd-swap-circuit} by an
accepting set defined by the inner product $s_A\cdot s_B$.  That parity provides, between the two
interfering copies, exactly the coupling $(-1)^{d_A\cdot d_B}$ that the SWAP test creates using its extra layer.  This
saves one layer, but only at even order, where the two halves query the same number of oracles
through the shared layers.  Since the construction is symmetric between the two halves, it always
uses an even number $2(k-1)$ of oracles, which is why it does not extend to $2k-1$.  For $k=2$ the
circuit is closely related to the depth-two IQP construction of Buzet and
Chailloux~\cite{buzet2026iqp} for twofold Forrelation, which uses the same quadratic identity;
Proposition~\ref{prop:even-forrelation} extends it to every even order.

The $(2k-1)$-fold instance of Theorem~\ref{thm:all-level-separation} is therefore the
\emph{smallest-order} problem $\mathrm{ORF}_{K,n}$ separating $\FH_k$ from $\FH_{k+1}$.  Indeed, amplifying
Proposition~\ref{prop:even-forrelation} to bounded error, exactly as in
Theorem~\ref{thm:upper-worst-case}, places the $(2k-2)$-fold Forrelation problem in $\FH_k$, whereas
the $(2k-1)$-fold problem is not in $\FH_k$ (Theorem~\ref{thm:lower-worst-case}).  On the promise problems $\mathrm{ORF}_{K,n}$ of Definition~\ref{def:hard-problem}, $\FH_k$ and the
$(k-1)$-round parallel-query model into which it embeds (Lemma~\ref{lem:round-embedding}) agree at
every fixed order $K$.  For even $K\le2k-2$ both decide $\mathrm{ORF}_{K,n}$, by
Proposition~\ref{prop:even-forrelation} at order $K$.  For odd $K\le2k-3$ both decide it, by the
SWAP test of Proposition~\ref{prop:odd-swap-circuit} at Fourier depth $(K+3)/2\le k$.  For
$K\ge2k-1$ neither decides it within the width bound of Remark~\ref{rem:parametric-lower-bound},
applied with $r=k-1$ and $2r<K$, and in particular neither does so with polynomially many parallel
queries.  Hence no single fixed-order problem $\mathrm{ORF}_{K,n}$ witnesses a separation
of $\FH_k$ from the round model; whether a planted, mixed or length-varying construction from these
problems does so remains open.
\end{remark}

\subsection{Classical simulation at the second level}\label{subsec:fh2-classical}

This subsection is a digression: nothing in Sections~\ref{sec:standard-hierarchy}--\ref{sec:landscape} depends on it.

\begin{remark}[Classical simulation of low Fourier depth]\label{rem:classical-simulation}
Lemma~\ref{lem:round-embedding} also brings $\FH_k$ with polynomially many queries within the scope
of the work on classical simulation of bounded-round quantum algorithms.
Aaronson and Ambainis~\cite{aaronson2014forrelationproblemoptimallyseparates} stated an
$O(M^{1-\frac1{2r}})$-query classical simulation of every $r$-query quantum algorithm, and Bravyi, Gosset,
Grier, and Schaeffer~\cite{bravyi2021classical} repaired a gap in the proof and established the
general result; this is the
one-query-per-round ($t=1$) base case of Conjecture~1.7 of~\cite{girish2024power}, which predicts
that $r$ rounds of $t$ parallel
queries admit classical simulation with $\widetilde O_{t,r}(M^{1-\frac1{2r}})$ queries.  If that
conjecture holds, then via Lemma~\ref{lem:round-embedding} every $\FH_k$ computation with $q$ phase
queries is classically simulable with $\widetilde O_{q,k}(M^{1-\frac{1}{2(k-1)}})$ queries, where the
dependence on $q$ is the one the conjecture leaves unspecified; a polynomial-query simulation of
polynomial-query $\FH_k$ circuits follows only if that dependence is polynomial in $q$.  The sequential bound alone, applied to a $q$-query $\FH_k$ circuit, gives only the generic
$M^{1-\frac1{2q}}$; the improvement to the round exponent $M^{1-1/(2(k-1))}$ is exactly the parallel case
$t>1$.
\end{remark}

For the second level, the non-adaptive case of that conjecture holds unconditionally.

\begin{proposition}[Classical simulation at the second level]\label{prop:fh2-classical}
There is a classical randomized algorithm with the following property.  Let $C$ be an $\FH_2$
circuit with $w$ Fourier wires making $q$ phase queries to an oracle, and write $D=2^{w}$.  Given classical query access
to that oracle and a description of $C$, the algorithm estimates $\Pr[C\text{ accepts}]$ to additive
error $\varepsilon\in(0,1)$ with probability $1-\varrho$, for any $\varrho\in(0,1)$, using
$O\!\bigl(q\sqrt D\,\varepsilon^{-1}\log(2/\varrho)\bigr)$ classical queries.  For circuits with
$w\le n+O(\log n)$ Fourier wires, so that $D=\widetilde O(N)$, this is
$\widetilde O\bigl(q\sqrt N\,\varepsilon^{-1}\log(2/\varrho)\bigr)$.
\end{proposition}

\begin{proof}
\emph{Reduction to a quadratic form in an oracle-dependent phase vector.}
If the exact Fourier depth of $C$ is at most one, or if $w=0$, its acceptance probability does not
depend on the oracle (Lemma~\ref{lem:round-embedding}), and the algorithm outputs it without
queries.  Assume exact depth two and $w\ge1$, and write $C=U_2\,\Lambda\,U_1\,\Lambda\,U_0$ with
$\Lambda=H^{\otimes w}\otimes I$ and accepting set $S$.  Queries in $U_0$ act on a computational
basis state and contribute a global phase, and the oracle-dependent phases of $U_2$ cancel in
$U_2^{\dagger}\Pi_SU_2$ (Step~1 in the proof of Lemma~\ref{lem:round-embedding}); deleting them,
$U_0\ket0=e^{i\theta}\ket a_F\ket\eta_A$ for an oracle-independent basis state with Fourier part
$a\in\bits^w$ and ancilla part $\eta$, and $U_2^{\dagger}\Pi_SU_2=\Pi_{S'}$ with $S'=\pi_2^{-1}(S)$
oracle-independent.  The first layer produces $D^{-1/2}\sum_{x\in\bits^w}(-1)^{a\cdot x}\ket x\ket\eta$.
On the branch $\ket x\ket\eta$ entering $U_1$, every gate is basis-preserving, so the register
follows a fixed trajectory and the address and function index of each query are determined by $x$
and the fixed $\eta$ alone.  Hence
\begin{equation}\label{eq:fh2-branch-phase}
U_1\ket x\ket\eta=d_{\mathrm{int}}(x)\,\ket{\sigma(x)},
\qquad
d_{\mathrm{int}}(x)=\lambda(x)\prod_{j=1}^{q_{\mathrm{int}}}F_{\alpha_j(x)}\bigl(a_j(x)\bigr),
\qquad \abs{d_{\mathrm{int}}(x)}=1,
\end{equation}
where $\sigma$ is an oracle-independent injection of $\bits^w$ into the basis states of the full
register, $\lambda(x)$ is an oracle-independent phase, and $a_j(x),\alpha_j(x)$ are the
oracle-independent query data, and $q_{\mathrm{int}}\le q$ is the number of queries in $U_1$.
Put $d(x):=(-1)^{a\cdot x}d_{\mathrm{int}}(x)$; a single value $d(x)$ costs $q_{\mathrm{int}}$
classical queries.

Writing $V$ for the isometry $\ket x\mapsto\ket{\sigma(x)}$, the state before the final measurement
is $D^{-1/2}U_2\Lambda Vd$, so
\begin{equation}\label{eq:fh2-quadratic-form}
\Pr[C\text{ accepts}]=\frac1D\,d^{*}Qd,
\qquad
Q:=V^{*}\Lambda^{*}\Pi_{S'}\Lambda V.
\end{equation}
\emph{Two bounds on $Q$.}
The matrix $Q$ is oracle-free, and $0\preceq Q\preceq I$ because $\Lambda^{*}\Pi_{S'}\Lambda$ is a
projector and $V$ is an isometry.  Two consequences of $0\preceq Q\preceq I$ are all we need:
\begin{equation}\label{eq:fh2-Q-bounds}
\Norm{Q}_{F}^{2}=\operatorname{Tr}(Q^{2})\le\operatorname{Tr}Q\le D,
\qquad
\Norm{Q^{\circ}d}_{2}\le\Norm{Qd}_{2}+\Norm{(\operatorname{diag}Q)d}_{2}\le2\sqrt D,
\end{equation}
where $Q^{\circ}:=Q-\operatorname{diag}Q$; the second uses $\Norm{Q}_{\mathrm{op}}\le1$,
$\Norm d_2=\sqrt D$, and $\sum_xQ_{xx}^{2}\le\sum_xQ_{xx}\le D$.

\emph{The estimator.}
The diagonal of \eqref{eq:fh2-quadratic-form} contributes $\operatorname{Tr}(Q)/D$, which is
oracle-free and needs no queries.  Put $s:=\lceil C_0\sqrt D/\varepsilon\rceil$, where $C_0$ is an
absolute constant fixed below.  If $s\ge D$, evaluate $d$ at all $D$ points and return
\eqref{eq:fh2-quadratic-form} exactly, using $q_{\mathrm{int}}D\le q_{\mathrm{int}}s$ classical
queries and no randomness.  Assume therefore $s<D$.  For the off-diagonal part, draw $s$ points
$\mathcal X\subseteq\bits^w$ uniformly without replacement, evaluate $d$ at each of them using
$q_{\mathrm{int}}s$ classical queries, and form
\[
\widehat p=\frac{\operatorname{Tr}Q}{D}
+\frac{D-1}{s(s-1)}\sum_{\substack{x,x'\in\mathcal X\\ x\ne x'}}\overline{d(x)}\,Q_{x,x'}\,d(x').
\]
\emph{Its variance.}
Since $\mathcal X$ contains each ordered pair of distinct points with probability $s(s-1)/(D(D-1))$, the
estimator is unbiased.  It is a $U$-statistic of degree two~\cite{hoeffding1948class} with kernel
$a(x,x')=\overline{d(x)}Q_{x,x'}d(x')$, which is real after symmetrization because $Q$ is Hermitian,
so its variance is $O\bigl(\zeta_1/s+\zeta_2/s^{2}\bigr)$ after the scaling by $D-1$, where
$\zeta_1$ is the variance of the Hájek projection and $\zeta_2$ bounds the second moment of the
kernel.  By \eqref{eq:fh2-Q-bounds},
\[
(D-1)^{2}\zeta_1\le\frac{\Norm{Q^{\circ}d}_2^{2}}{D}\le4,
\qquad
(D-1)^{2}\zeta_2\le\frac{(D-1)\Norm{Q}_{F}^{2}}{D}\le D ,
\]
the first because the projection at $x$ is $\overline{d(x)}(Q^{\circ}d)_x/(D-1)$ and $\abs{d(x)}=1$,
the second because the kernel has mean square at most $\Norm Q_F^2/(D(D-1))$.  Hence
\[
\operatorname{Var}(\widehat p)=O\!\left(\frac1s+\frac{D}{s^{2}}\right).
\]
Since $s<D$ and $s\ge C_0\sqrt D/\varepsilon$, we have $\varepsilon>C_0D^{-1/2}$, hence
$1/s\le\varepsilon/(C_0\sqrt D)<\varepsilon^{2}/C_0^{2}$ and $D/s^{2}\le\varepsilon^{2}/C_0^{2}$.
Choosing $C_0$ large enough therefore gives $\operatorname{Var}(\widehat p)\le\varepsilon^{2}/8$, and
Chebyshev's inequality gives $\abs{\widehat p-\Pr[C\text{ accepts}]}\le\varepsilon$ with probability
at least $\tfrac78$.  In both branches the query count is $O(q\sqrt D\,\varepsilon^{-1})$, and the
median of $O(\log(2/\varrho))$ independent runs gives the claim.
\end{proof}

\begin{remark}[What the estimator uses]\label{rem:fh2-classical-scope}
The argument uses only that the acceptance probability is a quadratic form
\eqref{eq:fh2-quadratic-form} in an oracle-free contraction $Q$ with a unit-modulus oracle-dependent
vector $d$, and that one entry of $d$ costs $q$ queries.  It therefore applies to every
$\FH_2$ circuit, with no restriction on the ancillas but with the bound depending on the number of
Fourier wires through $D=2^{w}$, and in particular to commuting (IQP) circuits on $w$ qubits.  The corresponding statement at Fourier depth $k\ge3$ fails for this estimator, since the
acceptance probability is then a form of degree $2(k-1)$ in $d$ rather than a quadratic one.
\end{remark}

For $w=n$ this bound is tight for constant $q$ and $\varepsilon$, up to polylogarithmic factors, and
it applies to a problem of interest: the sign-decision version of $2$-Forrelation considered by Buzet
and Chailloux~\cite{buzet2026iqp} is decided at Fourier depth two.  Indeed, the identity
$Q(x)+Q(y)+Q(x{+}y)=x\cdot y+\abs{x}\,\abs{y}\pmod 2$ for the quadratic form $Q(x)=\sum_{i<j}x_ix_j$ lets a
single diagonal layer between two Hadamard layers reproduce the Forrelation circuit, which they
realize as a commuting (IQP) circuit $H^{\otimes n}DH^{\otimes n}$, a special case of $\FH_2$
satisfying the hypothesis of Proposition~\ref{prop:fh2-classical}.  Since $2$-Forrelation requires $\widetilde\Omega(\sqrt N)$
classical randomized queries~\cite{aaronson2014forrelationproblemoptimallyseparates}, the classical
query complexity of such second-level circuits is $\widetilde\Theta(\sqrt N)$.

Since $\FH_2$ embeds into a single query round (Lemma~\ref{lem:round-embedding}),
Proposition~\ref{prop:fh2-classical} proves unconditionally the non-adaptive case $r=1$ of the
classical simulation conjecture of~\cite{girish2024power} (their Conjecture~1.7) for the non-adaptive algorithms that arise from $\FH_2$ circuits whose number $w$ of Fourier wires
satisfies $2^{w}=\widetilde O(M)$, with the bound $\sqrt M$ that the conjecture predicts, which is
below the sequential bound $M^{1-\frac1{2q}}$ for every $q\ge2$ and matches it at $q=1$.  Circuits with many more Fourier wires than $\log M$, such as the
amplified circuits of Theorem~\ref{thm:upper-worst-case}, and arbitrary non-adaptive quantum algorithms
are not covered.  The method is specific to this case.  Writing the acceptance probability of an $\FH_k$
circuit as $\abs a^2$ pairs the amplitude with its conjugate, so the oracle positions are the
$k-1$ interior blocks together with their conjugates.  At $k=2$ these are the two endpoint positions $x$ and $x'$
of \eqref{eq:fh2-quadratic-form}, coupled only through the oracle-free matrix $Q$; this is why $Q$ is
known and the variance of the estimator remains $O(D/s^2)$.  For $k\ge3$ the expression contains
adjacent oracle positions, and we do not know how to control the variance of the corresponding
estimator.  The higher levels correspond
to the case $t>1$ of that conjecture, which is open.

\section{The standard-query hierarchy}\label{sec:standard-hierarchy}
\label{app:standard-growth}

In the standard-query hierarchy a query writes its answer into a register.  In that model a block
can carry out a
coherent, classically adaptive computation between two Hadamard layers, so the round simulation of
Lemma~\ref{lem:round-embedding} is no longer available.  We prove a Fourier growth bound for such
circuits directly, by factoring each block through the leaves of its decision tree, and deduce the
adjacent-level separation in this model as well
(Theorem~\ref{thm:standard-query-adjacent}).  Comparing the two lower bounds then shows that a
standard query is strictly more powerful than a phase query at a fixed Fourier depth
(Theorem~\ref{thm:phase-vs-standard}), and that no constant number of Hadamard layers, nor any
classical preprocessing, recovers the difference
(Corollary~\ref{cor:std-not-in-phase} and Proposition~\ref{prop:classical-prefix-separation}).

Throughout this section $\mathrm{ORF}_{K,n}$ denotes the promise sets $\Pi^{\Yes}_n,\Pi^{\No}_n$ of
Definition~\ref{def:hard-problem}, with the underlying Boolean functions presented through the
standard-query oracle of Definition~\ref{def:standard-fh} rather than through phase queries.

The relevant Forrelation order changes between the two query models.  In
Section~\ref{sec:higher-level-conjecture} the separating order is $K=2k-1$, because Fourier
depth $k$ decides $(2k-2)$-fold Forrelation and no higher order.  Standard access reaches one order
further: retaining the answer bit lets a depth-$k$ circuit run the Aaronson--Ambainis interference
test on $(2k-1)$-fold Forrelation (Proposition~\ref{prop:standard-forrelation-test} with $s=k$).  This
order is therefore unsuitable for separating standard depth $k$ from standard depth $k+1$,
since standard depth $k$ already solves it.  The next odd order is $2k+1$, and it is decided at
standard depth $k+1$ by the same test.  What has to be shown is that standard depth $k$ does not
decide it, and that is exactly what the growth bound of this section provides.  Concretely, we prove
that a standard-query $\FH_k$ circuit making $q$ queries has level-$\ell$ Fourier weight at most
$2^{O_{k,\ell}(1)}(q+1)^{\ell}M^{\frac\ell2(1-\frac1{2k-1})}$
(Theorem~\ref{thm:standard-query-growth}); the exponent $\tfrac12(1-\tfrac1{2k-1})$ falls below the
$\tfrac12(1-\tfrac1{K})$ that the Bansal--Sinha criterion needs at $K=2k+1$, which is what rules the
problem out.  The rest of the argument, the direct product, the hybrid, and the diagonalization, is the
one already used in Section~\ref{sec:higher-level-conjecture}.

A standard-query block is therefore a coherent classical decision tree rather than a non-adaptive
batch of phase queries.  The proof of the growth bound has three parts.  We first factor each adaptive block through a \emph{leaf space}, in
which the query path taken on a branch is recorded as an orthogonal classical label; this makes
the adaptivity explicit and bounded, and is the step that replaces the round simulation. We then
group the degree-$\ell$ Fourier contributions by the set of positions at which a Fourier variable
is queried, and pass by a finite unitriangular change of basis to \emph{certificate profiles},
which can be written as an operator product.  Finally, we cut that product at a
position chosen to minimize the number of certificates lying entirely on one side of the cut. A
certificate of that kind contributes a power of the ambient dimension $M$, whereas one meeting the cut
lies on a query path of length at most $q$; averaging over the cut positions balances the two and
gives the exponent of Lemma~\ref{lem:standard-path-growth}.

\subsection{A decision-tree factorization}
\label{subsec:standard-leaf-factorization}

Throughout, the input variables are $x\in\pmone^M$. On each incoming computational basis state a
standard-query basis-preserving block acts classically: it queries input variables adaptively,
reaches a leaf, accumulates a phase, and outputs another basis state.

We first explain why such a block cannot be treated directly.  The argument of
Section~\ref{subsec:standard-path-growth} requires the operator product to have factors of bounded
norm whose oracle dependence lies in explicit monomials.  Expanding a block over the leaves of its
decision tree is the standard Fourier expansion of a decision
tree~\cite{odonnell2021analysisbooleanfunctions}, and it does produce such monomials.
Unfortunately, summing the leaves back together destroys the norm control, because different
branches query different variables and the resulting operator mixes them.  We therefore do not sum the branches.  The
branch taken is a function of the incoming basis label and the input, so we compute it into an
auxiliary register, in the same way that a classically adaptive computation is
made coherent by retaining its branch rather than uncomputing
it~\cite{nielsen2010quantum}.  The block then factors into an isometry followed by a contraction,
adjacent factors of the product are forced to agree on the branch, and the product structure that the
growth argument needs is restored.  The factorization is elementary; it restores the norm control
that the diagonal oracles of~\cite{girish2024power} provide automatically.  The following definition abstracts the
behavior described above, and allows leaf coefficients of modulus at most one so that it covers the
diagonal acceptance operator of Theorem~\ref{thm:standard-query-growth} as well.

\begin{definition}[Decision-tree contraction]\label{def:dt-contraction}
Let $\mathcal H_Z$ be a computational-basis space with basis $\{\ket z:z\in Z\}$. A
\emph{decision-tree contraction of depth at most $q$} is specified, for each $z\in Z$, by a
deterministic decision tree of depth at most $q$.

Delete any branch that assigns two different answers to the same input variable, and suppress a
repeated query whenever its answer has already been determined earlier on the path. Every
remaining leaf $p$ then carries a set $P_p\subseteq[M]$ of distinct variables queried on the path,
an answer pattern $\sigma_p\in\pmone^{P_p}$, an output label $w_p\in Z$, and a coefficient
$c_p\in\mathbb C$ with $\abs{c_p}\le1$. On input $x$, let $p(z,x)$ be the unique leaf in the tree
associated with $z$ whose answer pattern agrees with $x$. The resulting input-dependent operator
is $B(x)\ket z=c_{p(z,x)}\ket{w_{p(z,x)}}$. We require that, for every fixed $x$, the map
$z\mapsto w_{p(z,x)}$ is injective on the set of $z$ for which $c_{p(z,x)}\neq0$. Hence
$\Norm{B(x)}\le1$ for every $x$.
\end{definition}

Two cases will occur. A standard-query basis-preserving unitary is a decision-tree contraction
with $\abs{c_p}=1$ at every reachable leaf. An input-dependent diagonal operator whose $z$th entry
is computed by a depth-$q$ decision tree and lies in the unit disk is also one: use that tree for
the label $z$, take $w_p=z$, and let $c_p$ be the diagonal value at the leaf. The second case is
the operator $\Theta(x)$ of Theorem~\ref{thm:standard-query-growth}.

It is convenient to record Fourier subsets in an auxiliary register $\mathcal T:=\ell_2(2^{[M]})$
with orthonormal basis $(e_S)_{S\subseteq[M]}$, on which $X^Ae_S=e_{S\triangle A}$ for
$A\subseteq[M]$. Writing $B(x)=\sum_{A\subseteq[M]}x^AB_A$ with $x^A:=\prod_{j\in A}x_j$, set
\[
\mathcal G_B:=\sum_{A\subseteq[M]}B_A\otimes X^A .
\]
Multiplication by the monomial $x^A$ has become the shift $S\mapsto S\triangle A$, so a product of
input-dependent operators becomes an ordinary product of the corresponding $\mathcal G$; and after a
Boolean Fourier transform on $\mathcal T$, $\mathcal G_B$ is the direct sum of the matrices $B(x)$
over $x\in\pmone^M$.

For a leaf $p$, define
\begin{equation}\label{eq:standard-leaf-factor}
\Phi_p
:=
\sum_{A\subseteq P_p}
2^{-\abs{P_p}}\sigma_p^A X^A,
\qquad
\sigma_p^A:=\prod_{j\in A}\sigma_p(j).
\end{equation}
On the Fourier block indexed by $x$, the operator $\Phi_p$ is multiplication by
\[
\prod_{j\in P_p}\frac{1+\sigma_p(j)x_j}{2}
=
\mathbf 1[x|_{P_p}=\sigma_p].
\]
Thus, for the leaves of any one decision tree, the operators $\Phi_p$ are orthogonal projections
summing to the identity.

\begin{lemma}[Decision-tree factorization]
\label{lem:standard-leaf-factorization}
Let $B$ be a decision-tree contraction. Write $\mathsf{Leaves}(z)$ for the leaves of the tree
associated with $z$, and let
\[
\mathcal H_P
:=
\operatorname{span}
\{\ket{z,p}:p\text{ is a leaf of the tree associated with }z\}
\]
be its leaf space, and set
\begin{equation}\label{eq:standard-LR}
\begin{aligned}
\mathcal L_B
&:={\sum_z\sum_{p\in\mathsf{Leaves}(z)}}
\ket{z,p}\!\bra z\otimes\Phi_p,\\
\mathcal R_B
&:={\sum_z\sum_{p\in\mathsf{Leaves}(z)}}
c_p\ket{w_p}\!\bra{z,p}\otimes\Phi_p.
\end{aligned}
\end{equation}
Then
\begin{equation}\label{eq:standard-leaf-factorization}
\mathcal R_B\mathcal L_B=\mathcal G_B,
\qquad
\Norm{\mathcal L_B}=1,
\qquad
\Norm{\mathcal R_B}\le1.
\end{equation}
Moreover, $\mathcal G_{B^\dagger} = \mathcal L_B^\dagger\mathcal R_B^\dagger$.
\end{lemma}

\begin{proof}
Fix a Fourier block $x$. On this block the leaf projections become scalar indicators. For every
basis label $z$, exactly one leaf $p(z,x)$ is active, and therefore $\mathcal L_B:\ket
z\mapsto\ket{z,p(z,x)}$. Hence $\mathcal L_B$ is an isometry on every Fourier block.

The second map acts by $\mathcal R_B:\ket{z,p(z,x)}\mapsto c_{p(z,x)}\ket{w_{p(z,x)}}$. For fixed
$x$, the nonzero output basis labels are distinct as $z$ varies, and each coefficient has modulus
at most one. Thus the nonzero columns of $\mathcal R_B$ are orthogonal and have norm at most one,
so $\Norm{\mathcal R_B}\le1$.

The composition agrees with $B(x)$ on every Fourier block, proving $\mathcal R_B\mathcal
L_B=\mathcal G_B$. Taking adjoints gives the final identity.
\end{proof}

The adaptive query path is determined by the common leaf label $(z,p)$. We insert an auxiliary
operator $Q$ acting on this label as $\mathcal R_B Q\mathcal L_B$.  For an adjoint factor
$B^\dagger$, the corresponding insertion is $\mathcal L_B^\dagger Q\mathcal R_B^\dagger$. In
either case the two factors adjacent to $Q$ are forced to refer to the same leaf.

\subsection{Fourier growth along the operator product}
\label{subsec:standard-path-growth}

The argument follows the proof of~\cite[Theorem~4.1]{girish2024power}, from which we take four
devices: grouping the Fourier configurations into profiles, passing to an auxiliary vector $h$
related to the target weights by a triangular matrix, enlarging the operator product by subset
registers, and cutting the product to bound row and column sums.  The additional ingredient required
in the coherent standard-query setting is the leaf-space factorization of
Lemma~\ref{lem:standard-leaf-factorization} and the passage from incidence patterns to certificates.

We prove a slightly more general statement. Let
\begin{equation}\label{eq:standard-general-chain}
f(x)
=
u^\dagger B_1(x)\Lambda_1B_2(x)\Lambda_2\cdots \Lambda_{d-1}B_d(x)v,
\end{equation}
where $\Norm{u},\Norm{v}\le1$, every separator satisfies $\Norm{\Lambda_i}\le1$, and every $B_i$
is either a depth-$q$ decision-tree contraction or the adjoint of one. All products are assumed to
be dimensionally compatible; the ambient dimensions are otherwise unrestricted.

The argument below carries several indices at once, so we fix them here and use them consistently.

\begin{center}
\begin{tabular}{@{}ll@{}}
$d$ & the number of factors $B_i$ in the product \eqref{eq:standard-general-chain}\\
$\ell$ & the Fourier level being bounded\\
$q$ & the depth of each decision tree, so at most $q$ variables on any one path\\
$M$ & the number of input variables left free by the restriction\\
$\mathcal C$ & the nonempty subsets of $[d]$, used to label sets of positions\\
$\gamma\in\mathcal C$ & the positions whose query paths happen to reach a given variable\\
$c\in\mathcal C$ & the positions at which a variable is \emph{required} to be reachable\\
$m$, $s$ & how many variables carry each $\gamma$, respectively each $c$\\
$J_c$ & the set of variables assigned the requirement $c$, of size $s_c$\\
$r$ & the position at which the product is cut\\
$e_r$ & how many variables have all their required positions on one side of the cut
\end{tabular}
\end{center}

\noindent
Only $\gamma$ and $c$ need care: $\gamma$ records what the query paths make available, and $c$
records what we choose to demand of them.  The passage from the first to the second is Step~2 below,
and it is what makes the count finite.

Let $\mathcal C:=2^{[d]}\setminus\{\emptyset\}$. An element $\gamma\in\mathcal C$ will describe
the positions in the product whose decision-tree paths queried a given Fourier variable.  We call
$\gamma$ a \emph{factor-incidence pattern}.  A profile is a vector
$m=(m_\gamma)_{\gamma\in\mathcal C} \in\mathbb Z_{\ge0}^{\mathcal C}$. For fixed total size
$\ell$, the number of profiles satisfying $\sum_\gamma m_\gamma=\ell$ is
\begin{equation}\label{eq:standard-profile-count}
D_{d,\ell}
:=
\binom{\ell+2^d-2}{2^d-2}.
\end{equation}

The passage from incidence patterns to certificates is governed by the following counting matrix.
A source variable of incidence pattern $\gamma$ may receive a certificate type $c$ only when
$c\subseteq\gamma$, and the matrix records how many assignments realize a given pair of profiles.

\begin{lemma}[Certificate count]\label{lem:standard-cert-count}
For profiles $s,m$ of total size $\ell$, let $N[s,m]$ be the assignment count
\begin{equation}\label{eq:standard-flow-count}
\begin{aligned}
N[s,m]
:=\#\bigl\{\varphi:
&\bigsqcup_{\gamma\in\mathcal C}
\bigl(\{\gamma\}\times[m_\gamma]\bigr)
\longrightarrow\mathcal C:\\
&\varphi(\gamma,t)\subseteq\gamma\text{ for every }(\gamma,t),\quad
\abs{\varphi^{-1}(c)}=s_c\text{ for every }c\in\mathcal C
\bigr\}.
\end{aligned}
\end{equation}
Then $N$ is invertible, and
\begin{equation}\label{eq:standard-flow-inverse}
\Norm{N^{-1}}_1
\le
(2^d-1)^{\ell},
\end{equation}
where $\Norm{\,\cdot\,}_1$ denotes the maximum absolute column sum.
\end{lemma}

\begin{proof}
The matrix is inverted exactly, by M\"obius inversion on the Boolean
lattice~\cite[the M\"obius inversion formula, Prop.~3.7.1, and the Boolean lattice,
Ex.~3.8.3]{stanley2012enumerative}.  Introduce commuting
variables $y_c$, $c\in\mathcal C$.  Since an assignment sends each of the $m_\gamma$ labeled
sources of pattern $\gamma$ independently to some $c\subseteq\gamma$, the definition of $N$ gives,
for every profile $m$ of total size $\ell$,
\begin{equation}\label{eq:standard-flow-genfun}
\sum_{s}N[s,m]\prod_{c\in\mathcal C}y_c^{\,s_c}
=
\prod_{\gamma\in\mathcal C}
\Bigl(\sum_{c\subseteq\gamma}y_c\Bigr)^{m_\gamma},
\end{equation}
the inner sums running over nonempty $c$.  Thus $N$ is the matrix, in the monomial basis of the
homogeneous polynomials of degree $\ell$, of the substitution
$y_\gamma\mapsto\sum_{c\subseteq\gamma}y_c$.  The transformation
$u_\gamma=\sum_{c\subseteq\gamma}y_c$ of the variables is the zeta transform of the lattice
$\mathcal C$, and M\"obius inversion inverts it explicitly:
$y_\gamma=\sum_{c\subseteq\gamma}(-1)^{\abs\gamma-\abs c}u_c$.  Hence $N$ is invertible, and
$N^{-1}$ is the matrix of the inverse substitution.  A column of $N^{-1}$ lists the coefficients of
$\prod_\gamma\bigl(\sum_{c\subseteq\gamma}(-1)^{\abs\gamma-\abs c}y_c\bigr)^{m_\gamma}$, a product
of $\ell$ linear forms, each with at most $2^d-1$ terms of unit modulus.  Since the sum of absolute
values of coefficients is submultiplicative under polynomial multiplication, every absolute column
sum is at most $(2^d-1)^{\ell}$, which is \eqref{eq:standard-flow-inverse}.
\end{proof}

\begin{lemma}[Fourier growth of an adaptive product]
\label{lem:standard-path-growth}
For every restriction $\rho$ leaving $\widetilde M$ variables free and every $\ell\ge1$,
\begin{equation}\label{eq:standard-chain-growth}
L_{1,\ell}(f_\rho)
\le
A_{d,\ell}(q+1)^{\ell}
\max\left\{
1,
\widetilde M^{\frac12\lfloor(d-1)\ell/d\rfloor}
\right\},
\end{equation}
where
\[
A_{d,\ell}
:=
D_{d,\ell}\,(2^d-1)^{\ell}
\le
(2^d-1)^{2\ell},
\]
the last inequality because a profile of total size $\ell$ is a multiset of $\ell$ elements of
$\mathcal C$, so $D_{d,\ell}\le(2^d-1)^{\ell}$.
\end{lemma}

\begin{proof}
Restricting variables only prunes the decision trees and fixes some answers on their paths. It
does not increase their depth, and the resulting operators remain decision-tree contractions. If
no variables remain free, then $L_{1,\ell}(f_\rho)=0$ for $\ell\ge1$. We may therefore assume
$\widetilde M\ge1$, replace $M$ by $\widetilde M$, and prove the unrestricted statement, writing
$M$ again for the number of free variables.

The proof has six steps. We first expand the degree-$\ell$ Fourier weight into leaf configurations
and group them by factor-incidence profile, and then use Lemma~\ref{lem:standard-cert-count} to
replace those profiles by certificate profiles, reducing the bound to a single profile $s$. We
next choose a split position $r$ in the product, minimizing the number of certificates whose entire
lifetime lies on one side of it. We then represent the resulting quantity $h(s)$ as a scalar
product of two vectors built from auxiliary certificate registers, verify that this representation
is exact, and finally bound the norms of the two vectors by counting the certificates on each side
of the split.

\emph{Step 1: Fourier configurations.} Choose phases $a(S)$ of modulus one so that
\begin{equation}\label{eq:standard-dual-phases}
L_{1,\ell}(f)
=
\sum_{\abs S=\ell}a(S)\widehat f(S).
\end{equation}
Take $a(S)=\overline{\widehat f(S)}/\abs{\widehat f(S)}$ when $\widehat f(S)\neq0$ and $a(S)=1$
otherwise, extended arbitrarily to subsets of size other than $\ell$, whose values never
contribute.

For every leaf $p$,
\begin{equation}\label{eq:standard-leaf-expansion}
\mathbf 1[x\text{ follows }p]
=
\sum_{A\subseteq P_p}
2^{-\abs{P_p}}\sigma_p^A x^A.
\end{equation}
Expanding all decision-tree factors in \eqref{eq:standard-general-chain} produces terms indexed by
a leaf tuple $\mathbf p=(p_1,\ldots,p_d)$ and subsets $\mathbf A=(A_1,\ldots,A_d)$ with
$A_i\subseteq P_{p_i}$.  An adjoint factor is expanded through the leaves of the underlying
decision-tree contraction, with the local bra and ket reversed and the leaf coefficient
conjugated. Put $S(\mathbf A):=A_1\triangle\cdots\triangle A_d$. Let $\omega(\mathbf p,\mathbf A)$
denote the resulting scalar coefficient, including the endpoint vectors, the leaf coefficients,
the separator matrix elements, the factors $2^{-\abs{P_{p_i}}}\sigma_{p_i}^{A_i}$, and the dual
phase $a(S(\mathbf A))$ when $\abs{S(\mathbf A)}=\ell$. Set $\omega(\mathbf p,\mathbf A)=0$ when
$\abs{S(\mathbf A)}\neq\ell$.

For each $j\in S(\mathbf A)$, define its \emph{factor-incidence pattern}
\begin{equation}\label{eq:standard-incidence-pattern}
\gamma_{\mathbf p}(j)
:=
\{i\in[d]:j\in P_{p_i}\}
\in\mathcal C.
\end{equation}
This records the positions whose query paths make $j$ available. It does not record the positions
for which $j\in A_i$; that parity information is carried by the XOR shifts and will be enforced at
the split operator.

Write $m_\gamma(\mathbf p,\mathbf A):=\abs{\{j\in S(\mathbf A):\gamma_{\mathbf p}(j)=\gamma\}}$
for the incidence profile of the configuration, and collect the configurations with a given
profile into
\[
g(m)
:=
\sum_{\substack{\mathbf p,\mathbf A\\
 \text{incidence profile }m}}
\omega(\mathbf p,\mathbf A),
\]
so that, by construction,
\begin{equation}\label{eq:standard-L1-to-g}
L_{1,\ell}(f)
=
\left|\sum_m g(m)\right|
\le
\sum_m|g(m)|
=
\Norm{g}_1.
\end{equation}

\emph{Step 2: certificate profiles.} Fix a profile $s=(s_c)_{c\in\mathcal C}$ with $\sum_c
s_c=\ell$. For a configuration $(\mathbf p,\mathbf A)$ with $S=S(\mathbf A)$, a \emph{certificate
tuple of profile $s$} is a family $(J_c)_{c\in\mathcal C}$ satisfying
\begin{equation}\label{eq:standard-certificates}
\abs{J_c}=s_c,
\qquad
J_c\subseteq\bigcap_{i\in c}P_{p_i},
\qquad
\bigcup_{c\in\mathcal C}J_c=S.
\end{equation}
Because $\sum_c\abs{J_c}=\abs S=\ell$, the sets in a valid tuple are automatically pairwise
disjoint. Thus each $j\in S$ receives a unique certificate type $c$, and the containment condition
is exactly $c\subseteq\gamma_{\mathbf p}(j)$. Weighting each configuration by the number of
certificate tuples of a given profile that it admits, set
\begin{equation}\label{eq:standard-hNg}
h(s)
:=
\sum_{\mathbf p,\mathbf A}
\omega(\mathbf p,\mathbf A)
\cdot
\#\{\text{certificate tuples of profile }s
 \text{ for }(\mathbf p,\mathbf A)\}.
\end{equation}
A configuration of incidence profile $m$ admits exactly $N[s,m]$ compatible certificate tuples, so
$h=Ng$. Combining \eqref{eq:standard-L1-to-g}, \eqref{eq:standard-hNg}, and
\eqref{eq:standard-flow-inverse}, we obtain
\begin{align}
L_{1,\ell}(f)
&\le \Norm{g}_1
\le \Norm{N^{-1}}_1\Norm{h}_1\notag\\
&\le
D_{d,\ell}\Norm{N^{-1}}_1\max_s\abs{h(s)}
\le
D_{d,\ell}\,(2^d-1)^{\ell}
\max_s\abs{h(s)}.
\label{eq:standard-reduce-h}
\end{align}
It remains to bound $h(s)$ for one fixed certificate profile.

\emph{Step 3: choose the split point.} For $c\in\mathcal C$, regard $J_c$ as having a lifetime
from its first required position $\min c$ to its last required position $\max c$. For a split
point $r\in[d]$, define
\begin{equation}\label{eq:standard-er}
e_r
:=
\sum_{\min c>r}s_c
+
\sum_{\max c<r}s_c.
\end{equation}
Thus $e_r$ counts the certificate variables whose entire lifetime lies strictly to one side of the
split.

A certificate of type $c$ lies wholly to the right of exactly $\min c-1$ split points and wholly
to the left of exactly $d-\max c$ split points. Therefore
\begin{align}
\sum_{r=1}^d e_r
&=
\sum_{c\in\mathcal C}
s_c\bigl((\min c-1)+(d-\max c)\bigr)\notag\\
&\le
(d-1)\sum_cs_c
=
(d-1)\ell.
\label{eq:standard-er-average}
\end{align}
Hence some $r\in[d]$ satisfies
\begin{equation}\label{eq:standard-best-r}
e_r
\le
\left\lfloor\frac{(d-1)\ell}{d}\right\rfloor.
\end{equation}
Fix such an $r$.

\emph{Step 4: certificate registers and local transition matrices.} For each $c$ with $s_c>0$,
introduce a register $R_c$ whose basis values are $\bot$ and the subsets of $[M]$ of size $s_c$.
Every such register is initialized and postselected in $\bot$. The auxiliary matrices below have
entries in $\{0,1\}$ except for the dual phase inserted at the split.

For $i<r$, define $Q_i$ at the leaf interface of the $i$th factor. We read its matrix entries in
the direction in which the left boundary row vector is propagated toward the split. A nonzero
transition acts on the certificate register of a type $c$ in one of three ways. If $i=\min c$, it
replaces $\bot$ by a set $J_c\subseteq P_p$ of size $s_c$. If $\min c<i\le\max c$ and $i\in c$, it
keeps $J_c$ unchanged and requires $J_c\subseteq P_p$. Otherwise it copies the register unchanged.
In particular, a certificate whose lifetime ends before the split is retained until the split,
where it will be erased.

For $i>r$, define $Q'_i$ in the same left-to-right matrix orientation, now propagating from the
split toward the right boundary. A nonzero transition keeps every active $J_c$ unchanged, requires
$J_c\subseteq P_p$ whenever $i\in c$, and, when $i=\max c$, erases $J_c$ after performing this
check. Thus a certificate needed only on the right is created at the split, checked at the
positions in $c$, and erased at its final required position.

At the split position $r$, use an operator $W$ that is diagonal in the common leaf label $(z,p)$
and acts between a row Fourier-subset label $T$, a column Fourier-subset label $T'$, and the row
and column certificate-register values. Say that $c$ \emph{crosses} $r$ when $\min c\le r\le\max
c$. Its certificate transition is:
\[
\begin{array}{c|c}
\text{type of }c & \text{row-to-column action at the split}\\ \hline
\max c<r
 & \text{erase the left-created set }J_c\\
\min c>r
 & \text{create a right-going set }J_c\\
\min c<r<\max c,\ r\notin c
 & \text{copy }J_c\\
\min c<r<\max c,\ r\in c
 & \text{copy }J_c\text{ and require }J_c\subseteq P_p\\
r=\min c<\max c
 & \text{create }J_c\subseteq P_p\\
r=\max c>\min c
 & \text{require }J_c\subseteq P_p\text{ and erase it}\\
c=\{r\}
 & \text{sum over }J_c\subseteq P_p,\ \abs{J_c}=s_c.
\end{array}
\]
Let $\widetilde J_c$ denote the set visible at the split in the relevant case, with $\widetilde
J_c=\emptyset$ when $s_c=0$. A matrix entry of $W$ is zero unless
\begin{equation}\label{eq:standard-pivot-cover}
\abs{T\triangle T'}=\ell,
\qquad
T\triangle T'
=
\bigcup_{c\in\mathcal C}\widetilde J_c.
\end{equation}
Notice that, because $\sum_c|\widetilde J_c|=\ell$, these two conditions also force the sets
$\widetilde J_c$ to be pairwise disjoint. When \eqref{eq:standard-pivot-cover} holds, the matrix entry is
$a(T\triangle T')$.

Each entry of $W$ is therefore $0$ or a single phase of modulus one. Indeed, a fixed row and
column fix the leaf $(z,p)$, the labels $T,T'$, and every visible certificate $\widetilde J_c$
with $c\ne\{r\}$; the covering condition \eqref{eq:standard-pivot-cover} then determines the one
internal set $J_{\{r\}}=(T\triangle T')\setminus\bigcup_{c\ne\{r\}}\widetilde J_c$, so the
internal sum defining that entry has at most one surviving term. Consequently the row- and
column-sum estimates below count reachable rows and columns, with no coincident configurations
adding coherently.

\emph{Step 5: the scalar product equals $h(s)$.} Lift the chain
\eqref{eq:standard-general-chain} to $\mathcal T$ by replacing each factor $B_i(x)$ with $\mathcal
G_{B_i}$, factor every $\mathcal G_{B_i}$ by Lemma~\ref{lem:standard-leaf-factorization}, adjoin
the certificate registers of Step~4, and insert an auxiliary operator at the leaf interface of
each factor,
\[
\mathcal Q_i:=Q_i\quad(i<r),
\qquad
\mathcal Q_r:=W,
\qquad
\mathcal Q_i:=Q'_i\quad(i>r).
\]
The factor at position $i$ thus becomes
\[
\widetilde B_i
:=
\begin{cases}
\mathcal R_{B_i}\,\mathcal Q_i\,\mathcal L_{B_i},
&\text{$B_i$ a decision-tree contraction},\\[2pt]
\mathcal L_{B_i}^\dagger\,\mathcal Q_i\,\mathcal R_{B_i}^\dagger,
&\text{$B_i$ the adjoint of one},
\end{cases}
\]
which reduces to $\mathcal G_{B_i}$, respectively $\mathcal G_{B_i^\dagger}$, when $\mathcal Q_i$
is the identity.  Tensor every separator $\Lambda_i$ with the identity on $\mathcal T$ and on the
certificate registers, and initialize and postselect $\mathcal T$ in $e_\emptyset$ and every
certificate register in $\bot$: writing $e_\bot$ for the basis vector of the certificate registers
all of whose entries are $\bot$, set $\widetilde u:=u\otimes e_\emptyset\otimes e_\bot$ and
$\widetilde v:=v\otimes e_\emptyset\otimes e_\bot$.  We claim that
\begin{equation}\label{eq:standard-augmented-product}
h(s)
=
\widetilde u^\dagger\,
\widetilde B_1(\Lambda_1\otimes I)\widetilde B_2(\Lambda_2\otimes I)
\cdots(\Lambda_{d-1}\otimes I)\widetilde B_d\,
\widetilde v.
\end{equation}

First, the two factors adjacent to an auxiliary operator share the same leaf-space basis. They
therefore use the same leaf $p_i$ and the same adaptive query path.

Second, the left and right copies of the leaf projection at position $i$ contribute subsets
$A_i^L,A_i^R\subseteq P_{p_i}$. Starting and ending the Fourier-subset register at $\emptyset$,
their XOR shifts imply
\[
T\triangle T'
=
\bigtriangleup_{i=1}^d(A_i^L\triangle A_i^R),
\]
the total Fourier monomial of the scalar expansion.

Third, the certificate-register lifetimes enforce $\abs{J_c}=s_c, \qquad J_c\subseteq\bigcap_{i\in
c}P_{p_i}$. The split condition \,\eqref{eq:standard-pivot-cover} enforces that these sets are
pairwise disjoint and have union equal to the total Fourier set. Thus the auxiliary registers
count exactly the certificate tuples in \eqref{eq:standard-certificates}.

Finally, for each fixed $A_i=A_i^L\triangle A_i^R$,
\[
\sum_{A_i^L\triangle A_i^R=A_i}
\sigma_{p_i}^{A_i^L}\sigma_{p_i}^{A_i^R}2^{-2\abs{P_{p_i}}}
=
\sigma_{p_i}^{A_i}2^{-\abs{P_{p_i}}}.
\]
Indeed, each $A_i^L$ determines the unique set $A_i^R=A_i^L\triangle A_i$, and
$\sigma_{p_i}^{A_i^L}\sigma_{p_i}^{A_i^R}=\sigma_{p_i}^{A_i}$. This is precisely the coefficient
in \eqref{eq:standard-leaf-expansion}.  The same verification applies to an adjoint factor, with
its local leaf factorization reversed. Hence \eqref{eq:standard-augmented-product} holds.

\emph{Step 6: norm bounds.} We use the standard inequality, valid for every matrix $G$,
$\Norm{G}\le\sqrt{\Norm{G}_1\Norm{G}_\infty}$, where $\Norm{G}_1$ is the maximum absolute column
sum and $\Norm{G}_\infty$ is the maximum absolute row sum.

For $Q_i$, each column has at most one predecessor, whereas a fixed row has at most
\[
\prod_{\min c=i}\binom{\abs{P_p}}{s_c}
\le
\prod_{\min c=i}\binom q{s_c}
\]
possible successors. Therefore
\[
\Norm{Q_i}_1\le1,
\qquad
\Norm{Q_i}_\infty
\le
\prod_{\min c=i}\binom q{s_c},
\]
and
\begin{equation}\label{eq:standard-Q-norm}
\Norm{Q_i}
\le
\prod_{\min c=i}\binom q{s_c}^{1/2}.
\end{equation}
For $Q'_i$, each row has at most one successor, while a fixed column has at most the possible
certificate sets erased at $i=\max c$. Hence
\[
\Norm{Q'_i}_\infty\le1,
\qquad
\Norm{Q'_i}_1
\le
\prod_{\max c=i}\binom q{s_c},
\]
and
\begin{equation}\label{eq:standard-Qprime-norm}
\Norm{Q'_i}
\le
\prod_{\max c=i}\binom q{s_c}^{1/2}.
\end{equation}

Now fix a row of $W$ and count the columns it can reach. Reading the transition table at the
split, the column value of a certificate register is determined by the row except in three cases:
a type with $\min c>r$, whose set is created at the split and may be chosen from all $M$
variables; a type with $r=\min c<\max c$, whose set is created at the split inside $P_p$; and the
type $c=\{r\}$, whose set is summed over inside $P_p$. Once these are fixed,
\eqref{eq:standard-pivot-cover} determines $T'$ uniquely from $T$. Thus
\[
\Norm{W}_\infty
\le
\prod_{\min c>r}\binom M{s_c}
\prod_{r=\min c<\max c}\binom q{s_c}
\cdot\binom q{s_{\{r\}}}.
\]
Fixing a column gives the mirror-image estimate, in which the free row-side choices are the types
with $\max c<r$, erased at the split from all $M$ variables, those with $r=\max c>\min c$, erased
at the split from inside $P_p$, and again $c=\{r\}$:
\[
\Norm{W}_1
\le
\prod_{\max c<r}\binom M{s_c}
\prod_{r=\max c>\min c}\binom q{s_c}
\cdot\binom q{s_{\{r\}}}.
\]
A type that crosses the split strictly, $\min c<r<\max c$, appears on neither side: its set is
copied at the split, so it is constrained there but not free, whether or not $r\in c$.
Consequently,
\begin{align}
\Norm{W}
&\le
\sqrt{\Norm{W}_1\Norm{W}_\infty}\notag\\
&\le
M^{e_r/2}
\prod_{\substack{r\in\{\min c,\max c\}\\ \min c\ne\max c}}\binom q{s_c}^{1/2}
\cdot\binom q{s_{\{r\}}}.
\label{eq:standard-W-norm}
\end{align}

All leaf factors in \eqref{eq:standard-augmented-product} have norm at most one by
Lemma~\ref{lem:standard-leaf-factorization}, and every separator $\Lambda_i$ is a contraction. Multiplying \,\eqref{eq:standard-Q-norm},
\eqref{eq:standard-Qprime-norm}, and \eqref{eq:standard-W-norm} gives a factor of at most
\[
(q+1)^{E_r},
\qquad
E_r
:=
\frac12\sum_{\min c<r}s_c
+
\frac12\sum_{\max c>r}s_c
+
\frac12\sum_{\substack{r\in\{\min c,\max c\}\\ \min c\ne\max c}}s_c
+
s_{\{r\}}.
\]
The coefficient of a fixed $s_c$ in $E_r$ is at most $1$. Indeed, if $\max c<r$ or $\min c>r$,
exactly one of the first two sums contributes and the third does not, giving $\tfrac12$; if $\min
c<r<\max c$, both of the first two contribute and the third does not, giving $1$; if $r=\min
c<\max c$ or $r=\max c>\min c$, exactly one of the first two contributes and the third does, again
giving $1$; and if $c=\{r\}$, only the last term contributes, giving $1$. Since $\sum_cs_c=\ell$,
we have $E_r\le\ell$. Therefore
\begin{equation}\label{eq:standard-h-bound}
\abs{h(s)}
\le
(q+1)^{\ell}M^{e_r/2}.
\end{equation}
Combining \eqref{eq:standard-reduce-h}, \eqref{eq:standard-h-bound}, and
\eqref{eq:standard-best-r} proves \eqref{eq:standard-chain-growth}.
\end{proof}

Notice where the saving lies.  The exponent is $E_r\le\ell$, rather than the $2\ell$ that would
result from counting every type with $r\in c$ in full, because the row-side and the column-side counts at the split are
kept separate until the geometric mean is taken. A certificate created at the split is unrestricted
on the column side only; one erased at the split is unrestricted on the row side only; and one that
passes through the split is unrestricted on neither side. Each of the first two therefore contributes
$\binom q{s_c}^{1/2}$, and the third contributes nothing, in place of the full factor
$\binom q{s_c}$. Only the type $c=\{r\}$, both created and erased at the same
position, is unrestricted on both sides and contributes the full factor. This is the same counting
that gives the sharper non-adaptive phase-query bound of Theorem~\ref{thm:nonadaptive-growth},
where the two ways of factoring a single level interchange the two binomial factors between the
ambient dimension and the number of queries.

\subsection{Application to standard-query Fourier depth}
\label{subsec:standard-growth-application}

\begin{theorem}[Fourier growth for standard-query Fourier circuits]
\label{thm:standard-query-growth}
Fix $k\ge2$. Let $C$ be a standard-query $\FH_k$ circuit making $q$ queries, with acceptance
probability $f:\pmone^M\to[0,1]$. For every restriction $\rho$ leaving $\widetilde M$ variables
free and every $\ell\ge1$,
\begin{equation}\label{eq:standard-FH-growth}
L_{1,\ell}(f_\rho)
\le
A'_{k,\ell}(q+1)^{\ell}
\max\left\{
1,
\widetilde M^{
\frac12
\left\lfloor
\frac{(2k-2)\ell}{2k-1}
\right\rfloor}
\right\},
\end{equation}
where $A'_{k,\ell}\le(2^{2k-1}-1)^{2\ell}=2^{O(k\ell)}$.
\end{theorem}

\begin{proof}
\emph{Reduction to exact depth $k$.}
Let $j\le k$ be the exact Fourier depth of $C$.  If $j\le1$, then $f$ is the acceptance probability
of a randomized decision tree of depth at most $q$ (Lemmas~\ref{lem:fh0-std} and~\ref{lem:fh1-std}),
and so is every restriction $f_\rho$; expanding over the leaves,
$L_{1,\ell}(f_\rho)\le\binom q\ell\le(q+1)^{\ell}$, which is at most the right side of
\eqref{eq:standard-FH-growth}.  If $2\le j<k$, apply the argument below at depth $j$; its bound is
monotone in the depth, since $(2^{2j-1}-1)^{2\ell}\le(2^{2k-1}-1)^{2\ell}$ and
$\lfloor(2j-2)\ell/(2j-1)\rfloor\le\lfloor(2k-2)\ell/(2k-1)\rfloor$.  So assume exact depth $k$.

\emph{The two query counts.}
Write the circuit as in \eqref{eq:FHk-form}, let $q_j$ be the number of queries in the block $U_j$,
and set $q_*:=\max_{1\le j\le k}q_j$.  Each block is represented, on every incoming computational
basis state, by a decision tree of depth at most its own query count, so the blocks
$U_1,\dots,U_k$ give decision trees of depth at most $q_*$ and the initial block $U_0$ one of depth
at most $q_0$.  Since the $q_j$ sum to $q$, we have $q_0+q_*\le q$, and this is the only way the two
counts are combined below.

\emph{The initial block.}
The initial block $U_0$ begins on one computational basis state.  At each leaf of its classical
decision tree it produces one basis state up to a global phase. Conditioned on that leaf, the
state immediately after the first Hadamard layer is therefore a fixed unit vector depending only
on the leaf.

\emph{The final block.}
For the final block, let $\Pi_S$ be the acceptance projector and define
$\Theta(x):=U_k(x)^\dagger\Pi_SU_k(x)$. Because $U_k(x)$ is basis-preserving, $\Theta(x)$ is
diagonal. Its diagonal entry at $\ket z$ is the indicator that the decision tree followed by $U_k$
from $z$ ends in the accepting set. Thus $\Theta$ is a depth-$q$ decision-tree contraction.

\emph{The product form.}
Fix a leaf of the initial tree and let $v_z$ be the corresponding unit vector after the first
Hadamard layer. The conditional acceptance probability is
\[
\begin{aligned}
&v_z^\dagger
U_1^\dagger H\cdots H U_{k-1}^\dagger H\,
\Theta\,
H U_{k-1}H\cdots H U_1v_z,
\end{aligned}
\]
with identities on non-Fourier wires suppressed.  Hence it has the product form
\eqref{eq:standard-general-chain} with $d=2(k-1)+1=2k-1$ decision-tree factors
\[
U_1^\dagger,\ldots,U_{k-1}^\dagger,
\Theta,
U_{k-1},\ldots,U_1,
\]
and with the intervening global Hadamards as the contractions $\Lambda_i$.
Lemma~\ref{lem:standard-path-growth} applies after every further restriction.

\emph{Combining the two counts.}
Apply the classical-preprocessing lifting lemma, Lemma~\ref{lem:classical-prefix-lifting}, to the
decision tree for $U_0$, whose depth is at most $q_0$. For $0\le a\le\ell-1$, apply
Lemma~\ref{lem:standard-path-growth} at level $\ell-a$ with the factor depth $q_*$; for $a=\ell$,
use $L_{1,0}(f_z)=\abs{\widehat f_z(\emptyset)}\le1$. Since
\[
\left\lfloor
\frac{(2k-2)(\ell-a)}{2k-1}
\right\rfloor
\le
\left\lfloor
\frac{(2k-2)\ell}{2k-1}
\right\rfloor,
\]
we obtain, using that $A_{2k-1,\ell-a}\le(2^{2k-1}-1)^{2(\ell-a)}\le(2^{2k-1}-1)^{2\ell}=:A'_{k,\ell}$
for every $0\le a\le\ell$,
\begin{align*}
L_{1,\ell}(f_\rho)
&\le
A'_{k,\ell}
\max\left\{
1,
\widetilde M^{
\frac12
\left\lfloor
\frac{(2k-2)\ell}{2k-1}
\right\rfloor}
\right\}
\sum_{a=0}^\ell\binom{q_0}a(q_*+1)^{\ell-a}\\
&\le
A'_{k,\ell}(q+1)^{\ell}
\max\left\{
1,
\widetilde M^{
\frac12
\left\lfloor
\frac{(2k-2)\ell}{2k-1}
\right\rfloor}
\right\}.
\end{align*}
For the last inequality we used $\binom{q_0}a\le q_0^{\,a}\le\binom\ell a q_0^{\,a}$, so that the
sum is at most $\sum_{a=0}^{\ell}\binom\ell a q_0^{\,a}(q_*+1)^{\ell-a}=(q_0+q_*+1)^{\ell}$, and
then $q_0+q_*\le q$. This is \eqref{eq:standard-FH-growth}.
\end{proof}

We conclude with the standard-query upper bound used in Theorems~\ref{thm:standard-query-lower}
and~\ref{thm:phase-vs-standard}. Aaronson and
Ambainis~\cite{aaronson2014forrelationproblemoptimallyseparates} gave a control-qubit interference
test for $K$-fold Forrelation that makes $\lceil K/2\rceil$ queries and has acceptance probability
$\tfrac12(1+\Phi_K)$, ending with a measurement of the control qubit in the $X$ basis. For odd $K$
this test has a fixed standard-query Fourier depth: the middle query writes its answer into an
ancilla, so that the final comparison in the $X$ basis becomes a predicate in the computational
basis and no additional Hadamard layer is needed. We state this property below.

\begin{proposition}[The Aaronson--Ambainis interference test at standard Fourier depth]
\label{prop:standard-forrelation-test}
For every fixed $s\ge2$, there is a uniform standard-query $\FH_s$ circuit making $s$ queries
whose acceptance probability on $(F_0,\ldots,F_{2s-2})$ is
\begin{equation}\label{eq:standard-forrelation-test}
\frac{1+\Phi_{2s-1}(F_0,\ldots,F_{2s-2})}{2},
\end{equation}
where $\Phi_{2s-1}$ is the $(2s-1)$-fold Forrelation of \eqref{eq:odd-Phi}.
\end{proposition}

\begin{proof}
\emph{The registers.}
Use a control qubit $c$, an $n$-qubit index register $w$, a persistent query target $t$, and a
fresh answer bit $a$. Before the first global Hadamard layer, prepare $t$ in $\ket1$.

\emph{The layer schedule.}
The global Hadamard layers are scheduled as follows. Layer $1$ acts on $c$, $w$ and $t$, creating
$\ket+_c\ket+_w\ket-_t$; when $s\ge3$, layers $2,\dots,s-1$ act on $w$ alone; and layer $s$ acts
on both $c$ and $w$. When a layer should not act on $c$ or $t$, route that logical register off
the designated Fourier wires using basis-preserving SWAP gates and fresh computational-basis
placeholder wires.

\emph{The first $s-1$ queries, on the two branches of the control.}
For $j=1,\ldots,s-1$, the basis-preserving block following layer $j$ makes one standard query
using $t$ for phase kickback. Conditioned on $c=0$, it queries $F_{j-1}$ at $w$; conditioned on
$c=1$, it queries $F_{2s-1-j}$ at $w$.  Between consecutive such blocks, the next global Hadamard
layer Fourier-transforms $w$.

\emph{The middle query, written into a register rather than a phase.}
After the $s$th Hadamard layer, the final block makes the $s$th standard query to the middle
oracle $F_{s-1}$. Write $F_{s-1}(y)=(-1)^{f_{s-1}(y)}$. This last query is deliberately not used
for phase kickback: it writes $f_{s-1}(w)$ into the fresh answer bit $a$. The circuit accepts
exactly when $c\oplus a=0$. All routing and controlled selection of the indexed oracle are
basis-preserving, so the circuit lies in standard-query $\FH_s$ and makes exactly $s$ queries.

\emph{The acceptance probability.}
Let $L,R$ be the profiles from \eqref{eq:left-profile} and \eqref{eq:right-profile}, with $k=s$.
The two control branches prepare the left and right profiles, respectively.  After the
final Hadamard on $c$, the amplitude at $(c,w)=(b,y)$, immediately before the middle standard
query, is
\[
\frac{1}{2\sqrt N}
\bigl(L(y)+(-1)^bR(y)\bigr).
\]
Acceptance requires $b=f_{s-1}(y)$, and hence
\begin{align*}
\Pr[\mathrm{accept}]
&=
\frac{1}{4N}
\sum_y
\bigl(L(y)+F_{s-1}(y)R(y)\bigr)^2\\
&=
\frac{1}{4N}
\sum_y
\bigl(
L(y)^2+R(y)^2
+2L(y)F_{s-1}(y)R(y)
\bigr)\\
&=
\frac{1+\Phi_{2s-1}}2,
\end{align*}
by Parseval and Proposition~\ref{prop:odd-split}.
\end{proof}

\subsection{The adjacent-level separation}\label{subsec:standard-separation}

\begin{theorem}[Standard-query lower bound]\label{thm:standard-query-lower}
Fix $k\ge2$, and set
\[
d:=2k-1,\qquad K:=2k+1,\qquad \eta_k:=\frac{1}{2dK}.
\]
For all sufficiently large $n$, $\mathrm{ORF}_{K,n}$ is decided by a uniform
standard-query $\FH_{k+1}$ circuit with $O_k(1)$ queries, whereas every
standard-query $\FH_k$ circuit making at most $N^{\eta_k}$ queries fails on
some promised input with error greater than $\frac13$.
\end{theorem}

\begin{proof}
For the upper bound, apply Proposition~\ref{prop:standard-forrelation-test}
with $s=k+1$, so that $2s-1=K$ and the acceptance probability on a copy is
$\frac12(1+\mathsf{forr}_K)$.  On a $\Yes$ copy this is at least $\frac12+\frac\delta2$, and on
every $\No$ copy at most $\frac12+\frac\delta4$.  Taking the threshold $\frac12+\frac{3\delta}8$,
Hoeffding's inequality shows that $R=O(\delta^{-2}\log m)=O_k(1)$ repetitions per copy decide that
copy correctly except with probability $\frac1{3m}$; a union bound over the $m=2^{5K}$ copies of
Definition~\ref{def:hard-problem} and their OR then decides the promise with error at most
$\frac13$.  All repetitions and copies share the same $k+1$ Hadamard layers, and the final OR is a
basis-preserving predicate.

For the lower bound, suppose that a circuit $C$ with $q\le N^{\eta_k}$
decides the promise.  Recall Steps~1--2 of Theorem~\ref{thm:lower-worst-case}; they apply
without change here: after fixing all but one Forrelation block, we obtain a
restriction $h$ with
\[
 \left|\E_{\mathcal F_K}[h]-\E_{\mathcal U}[h]\right|\ge\frac1{4m}.
 \tag{\(\ddagger\)}
\]
Every further restriction of $h$ is also a restriction of the acceptance
probability of $C$.  Theorem~\ref{thm:standard-query-growth} therefore gives,
with $M=mKN$ and every $\ell\ge1$,
\[
 \mathsf{wt}_{\ell}(h_\rho)
 \le A'_{k,\ell}(q+1)^{\ell}
 M^{\frac{\ell}{2}(1-1/d)}.
\]
The restriction-stable transfer theorem (Theorem~\ref{thm:bs-transfer}) and
the Bansal--Sinha bias bound (Theorem~\ref{thm:bs-bias}) now imply, after
absorbing constants depending only on $k$ and the bounded range of $\ell$,
\begin{align*}
 \left|\E_{\mathcal F_K}[h]-\E_{\mathcal U}[h]\right|
 &\le C_k\sum_{\ell=K}^{K(K-1)}
 (q+1)^{\ell}
 N^{-\ell/(dK)}.
\end{align*}
Here the exponent is strict because
\[
 \frac\ell2(1-1/d)-\frac\ell2(1-1/K)=-\frac\ell{dK},
\]
using $K=d+2$.  We also used $\frac12\lfloor(d-1)\ell/d\rfloor\le\frac\ell2(1-1/d)$ and
$M=mKN=O_k(N)$.  Since $q\le N^{\eta_k}$, we have $(q+1)^{\ell}\le2^{\ell}N^{\eta_k\ell}$, with
$2^{\ell}$ absorbed into the constant, and $\eta_k=1/(2dK)$ gives
$-\ell/(dK)+\eta_k\ell=-\ell/(2dK)$.  The sum is therefore at most
\[
 C_k'\sum_{\ell=K}^{K(K-1)}N^{-\ell/(2dK)}=o_k(1),
\]
contradicting \((\ddagger)\).  This proves the lower bound.
\end{proof}

As in the phase-query model, the singly exponential growth constant permits the depth to grow.

\begin{corollary}[Growing depth, standard queries]\label{cor:growing-k-std}
There are absolute constants $c,c'>0$ and $n_0$ with the following two properties.
\begin{enumerate}[leftmargin=2em]
\item For every $n\ge n_0$ and every $k$ with $2\le k\le c\,n^{1/3}$, the promise problem
$\mathrm{ORF}_{2k+1,n}$, whose $2^{O(k)}$ functions are indexed by $O(k)$ bits, is decided by a
standard-query circuit of Fourier depth $k+1$ and size $2^{O(k)}\poly(n)$ making $2^{O(k)}$
queries, while every standard-query circuit of Fourier depth $k$ making at most $2^{c'n/k^2}$
queries fails on some promised instance.
\item Let $k(\cdot)$ be nondecreasing and polynomial-time computable with $2\le k(n)\le c\log n$
for all sufficiently large $n$.  Then the index width and the deciding circuit are polynomial at
those lengths, and the diagonalization of Theorem~\ref{thm:standard-query-adjacent}, with the
finitely many remaining lengths hardwired, separates
$\FH^{O}_{k(\cdot),\mathrm{std}}$ from $\FH^{O}_{k(\cdot)+1,\mathrm{std}}$, the standard-query
analogue of Definition~\ref{def:variable-depth}, relative to an oracle; for $k(n)=o(\log n)$ the
deciding circuit makes $n^{o(1)}$ queries.
\end{enumerate}
\end{corollary}

\begin{proof}
We track the dependence on $k$ in the proof of Theorem~\ref{thm:standard-query-lower}, with $d=2k-1$
and $K=2k+1$.  The upper bound is the test of Proposition~\ref{prop:standard-forrelation-test} with
$s=k+1$, repeated $R=O(\delta^{-2}\log m)=2^{\Theta(k)}$ times on each of the $m=2^{5K}$ copies; it
makes $2^{\Theta(k)}$ queries in total, and its size is $2^{\Theta(k)}\poly(n)$.

For the lower bound, the constants absorbed into $C_k$ there are now tracked.  With $q\le N^{\eta_k}$
and $\eta_k=\tfrac1{2dK}$, the term at level $\ell$ of that sum is at most
\[
A'_{k,\ell}\;4^{\ell}\,(8K)^{14\ell}\,(mK)^{\ell/2}\,2^{\ell}\;N^{-\ell/(2dK)}
\ \le\ 2^{C_0k\ell}\,N^{-\ell/(2dK)}
\]
for an absolute constant $C_0$, using $A'_{k,\ell}\le(2^{2k-1}-1)^{2\ell}$
(Theorem~\ref{thm:standard-query-growth}) and $m=2^{5K}$.  Once $k^{3}\le c''n$ for a suitable
absolute constant $c''$, the exponent is at most $-\ell n/(4dK)$, and the geometric sum from
$\ell=K$ is at most $2^{-\Omega(n/k)}$, below the threshold $\tfrac1{4m}=2^{-\Theta(k)}$ of
\((\ddagger)\) because $n/k\ge n^{2/3}\gg k$ in the stated range.  The contradiction and the
diagonalization are then as before.
\end{proof}

\begin{theorem}[Adjacent standard-query oracle separation]
\label{thm:standard-query-adjacent}
For every constant $k\ge2$, there is an oracle family $O$ such that
\[
 \FH_{k,\mathrm{std}}^{O}\subsetneq\FH_{k+1,\mathrm{std}}^{O}.
\]
\end{theorem}

\begin{proof}
Put $K=2k+1$.  At length $n$, let the oracle contain one instance $Z_n$ of
$\mathrm{ORF}_{K,n}$, and define the unary language
\[
 L_O:=\{1^n:Z_n\in\Pi_n^{\Yes}\}.
\]
Since $K=2k+1$ is odd, $K-1$ is even, and for the all-ones tuple
\eqref{eq:forr-dictionary} gives
$\mathsf{forr}_K=\frac1N\mathbf1^\top\mathsf H^{K-1}\mathbf1=\frac1N\mathbf1^\top\mathbf1=1$ by
$\mathsf H^2=I$.  Hence the all-ones instance belongs to $\Pi_n^{\Yes}$ at every length, so there is
always a promised choice.

Enumerate the uniform polynomial-query standard-query $\FH_k$ families as $(A^{(i)})_{i\ge1}$, as
clocked generators with malformed or over-depth outputs replaced by a rejecting circuit (as in the
proof of Theorem~\ref{thm:all-level-separation}), and let $q_i(n)$ be a polynomial query bound for
$A^{(i)}$.

Choose increasing lengths $n_i$ so large that
$q_i(n_i)\le 2^{\eta_k n_i}$ and Theorem~\ref{thm:standard-query-lower}
applies.  At length $n_i$, choose a promised instance on which
$A^{(i)}_{n_i}$ has error greater than $1/3$, which that theorem guarantees.
At all remaining lengths use the all-ones instance.  Then no
standard-query $\FH_k$ family decides $L_O$ with bounded error, since its
enumerated member fails at its dedicated length.  Conversely, the
standard-query $\FH_{k+1}$ circuit from
Theorem~\ref{thm:standard-query-lower} decides $L_O$ at every sufficiently
large length.  Hardwiring the finitely many remaining answers gives a uniform
standard-query $\FH_{k+1}$ family, so
$L_O\in\FH_{k+1,\mathrm{std}}^O\setminus\FH_{k,\mathrm{std}}^O$.
\end{proof}

Comparing the two models at the \emph{same} Fourier depth gives a further separation.

\begin{theorem}[Standard queries are more powerful than phase queries at fixed depth]\label{thm:phase-vs-standard}
For every constant $k\ge2$ there is an oracle $O$ such that
\[
 \FH_k^{O}\ \subsetneq\ \FH_{k,\mathrm{std}}^{O}.
\]
The separating problem is the OR of $(2k-1)$-fold Forrelation instances.
\end{theorem}

\begin{proof}
By Remark~\ref{rem:phase-vs-bit}, phase kickback places $\FH_k^{O}\subseteq\FH_{k,\mathrm{std}}^{O}$
for every oracle $O$, so it remains to make one inclusion strict.  Take the oracle $O$ and unary
language $L_O$ of Theorem~\ref{thm:all-level-separation}, whose instances are drawn from
$\mathrm{ORF}_{2k-1,n}$; that theorem proves $L_O\notin\FH_k^{O}$ in the phase-query model.  On the
other hand, Proposition~\ref{prop:standard-forrelation-test} with $s=k$ gives a uniform
standard-query $\FH_k$ circuit whose acceptance probability on a single instance is
$\tfrac12(1+\mathsf{forr}_{2k-1})$; amplifying over the $m$ copies and taking their OR, exactly as in
the upper bound of Theorem~\ref{thm:standard-query-lower}, decides $L_O$ with bounded error.  Hence
$L_O\in\FH_{k,\mathrm{std}}^{O}$, and $\FH_k^{O}\subsetneq\FH_{k,\mathrm{std}}^{O}$.
\end{proof}

Adaptive classical access does not by itself account for the advantage that
Theorem~\ref{thm:phase-vs-standard} exhibits: the separation survives when the phase model is granted
such access at no cost.  Write
$\P^{O}\text{-}\FH_k^{O}$ for the class of promise problems decided with bounded error by a uniform
family that first runs a polynomial-time adaptive classical computation with oracle access, and then,
on each classical transcript, a phase-$\FH_k$ circuit.

\begin{proposition}[Classical preprocessing does not close the gap]\label{prop:classical-prefix-separation}
Fix $k\ge2$ and let $c_k=\frac{1}{4(2k-1)(2k-2)}$.  For all sufficiently large $n$,
no algorithm that makes at most $N^{c_k}$ adaptive classical queries and then runs a phase-$\FH_k$
circuit making at most $N^{c_k}$ phase queries decides $\mathrm{ORF}_{2k-1,n}$ with error at most
$\tfrac13$ on every promised instance.  Consequently the oracle $O$ of
Theorem~\ref{thm:phase-vs-standard} satisfies
\[
\FH_{k,\mathrm{std}}^{O}\ \not\subseteq\ \P^{O}\text{-}\FH_{k}^{O}.
\]
\end{proposition}

\begin{proof}
Let $d_0\le N^{c_k}$ be the number of classical queries and $q\le N^{c_k}$ the number of phase
queries.  By Lemma~\ref{lem:round-embedding} each continuation is an algorithm with $r=k-1$ rounds of
at most $q_+:=\max\{1,q\}$ parallel phase queries, a continuation without queries being padded with
one query to the dummy index $0$, so Lemma~\ref{lem:classical-prefix-lifting} applies with $t=q_+$
and gives, for every restriction and every $\ell\ge1$,
\[
L_{1,\ell}(f_\rho)\ \le\ 2^{O_{k,\ell}(1)}\,\bigl(d_0+q_+\bigr)^{\ell}\,
M^{\frac{\ell}{2}\left(1-\frac1{2k-2}\right)}.
\]
This is the bound of Corollary~\ref{cor:FHk-growth} with $q$ replaced by $d_0+q_+$.  The proof of
Theorem~\ref{thm:lower-worst-case} uses the growth bound only through that factor, so running Steps
1--4 verbatim with $d_0+q_+$ in place of $q$ shows that the advantage is $o_k(1)$ provided
$d_0+q_+\le N^{c}$ for some $c<2c_k$, since by Remark~\ref{rem:parametric-lower-bound} the argument of
Theorem~\ref{thm:lower-worst-case} applies to every such query bound.  Since $d_0\le N^{c_k}$ and $q_+\le N^{c_k}$, we have
$d_0+q_+\le2N^{c_k}\le N^{3c_k/2}$ for large $n$, and $3c_k/2<2c_k$, so the conclusion of
Theorem~\ref{thm:lower-worst-case} holds: some promised instance is decided with error greater than
$\tfrac13$.

For the consequence, the diagonalization of Theorem~\ref{thm:all-level-separation} may be run against
the enumeration of uniform $\P^{O}\text{-}\FH_k$ families with explicit polynomial clocks in place of
the $\FH_k$ families, since the bound above is uniform over transcripts; it fixes an oracle $O$ with
$L_O\notin\P^{O}\text{-}\FH_k^{O}$.  The upper bound in Theorem~\ref{thm:phase-vs-standard} is
unchanged, so $L_O\in\FH_{k,\mathrm{std}}^{O}$.
\end{proof}

An adaptive classical preprocessing stage therefore does not close the gap between the two access
models.

\begin{remark}[One order of Forrelation at fixed depth]\label{rem:phase-vs-standard-reach}
Within the Forrelation family the two models differ by exactly one order.  At Fourier depth
$k$, phase access solves $(2k-2)$-fold Forrelation (Proposition~\ref{prop:even-forrelation}) but not
$(2k-1)$-fold Forrelation (Theorem~\ref{thm:lower-worst-case}), whereas standard access solves
$(2k-1)$-fold Forrelation (Proposition~\ref{prop:standard-forrelation-test}).  The OR of
$(2k-1)$-fold Forrelation therefore lies in standard-query depth $k$ but not in phase-query depth
$k$.  The upper bound used here is the Aaronson--Ambainis test at depth $k$; unlike the
adjacent separation in the standard-query model, it does not use the standard-query Fourier growth
theorem (Theorem~\ref{thm:standard-query-growth}).
\end{remark}

\subsection{Standard queries versus bounded Fourier depth}\label{subsec:pointer-chasing}

Within the Forrelation family, standard access reaches one order beyond phase access
(Remark~\ref{rem:phase-vs-standard-reach}).  As a statement about the classes, the advantage is
not bounded at all: relative to a suitable oracle, no constant number of Hadamard layers suffices for
a phase circuit to simulate a standard-query circuit, and already $\FH_{0,\mathrm{std}}^{O}$ is
contained in no level of the phase hierarchy.  The reason is visible in the two simulations proved
above.  A phase query acts trivially on a basis state, so the initial block of a phase circuit
contributes only a global phase, and Lemma~\ref{lem:round-embedding} places $\FH_k^{O}$ inside the
model with $k-1$ rounds of parallel queries.  A standard query writes its answer into a register, so
the initial block of a standard-query circuit is instead an adaptive classical computation, the
classical preprocessing that appears in Proposition~\ref{prop:three-round-bit-embedding}.  The
problem that separates the two models is pointer chasing; recall from Lemma~\ref{lem:fh0-std} that
$\FH_{0,\mathrm{std}}^{O}=\P^{O}$ for every oracle $O$.

Fix a polynomial-time computable schedule $1\le d(n)\le\poly(n)$, and fix a polynomial-time encoding
$\langle\cdot\rangle$ of the triples $(1^n,j,i)$ with $1\le j\le d(n)$ and $1\le i\le n$ into strings
of a length $n'=n'(n)\ge n$ from which $n$, $j$ and $i$ are recoverable, with distinct $n$ giving
distinct $n'$.  For a function $F_n:\bits^n\to\bits^n$, use $n$ of the oracle indices at length $n'$
to encode it by
\[
f_{n',i}(z)=F_n(z_{1:n})_i\qquad(z\in\bits^{n'},\ i\in[n]),
\]
where $z_{1:n}$ is the prefix of $z$ of length $n$, leave the remaining indices constant zero, and
set $x_0=0^n$ and $x_j=F_n(x_{j-1})$.  Define
\begin{equation}\label{eq:pointer-chasing-language}
 L_O=\bigl\{\langle 1^n,j,i\rangle\ :\ 1\le j\le d(n),\ 1\le i\le n,\ (x_j)_i=1\bigr\}.
\end{equation}
This keeps the length-preserving convention of Definition~\ref{def:FHk-rel}, under which a circuit
on an input of length $n'$ queries only $O_{n'}$; a query to $f_{n',i}$ at $z$ is answered by one
query to $F_n$ at $z_{1:n}$.

\begin{theorem}[Quantum sequentiality for exact chains; Chung, Fehr, Huang, and
Liao~\cite{chung2021compressed}]\label{thm:cfhl-chain}
Let $F:\bits^n\to\bits^n$ be uniformly random.  A quantum algorithm making $r$ rounds of at most $t\ge1$
parallel queries to $F$ outputs a sequence $x_0,x_1,\dots,x_{r+1}$ with $x_i=F(x_{i-1})$ for every
$1\le i\le r+1$ with probability $O\bigl(t^3(r+1)^3/2^n\bigr)$.  The starting point $x_0$ is free, so
the bound holds a fortiori when $x_0$ is prescribed.
\end{theorem}

In~\cite{chung2021compressed} the parameter that is cubed in the numerator is the \emph{chain
length}, which here equals $r+1$, and the algorithm is allowed strictly fewer query rounds than that
length.  We have written the bound in terms of the number of rounds $r$ accordingly, so that it
remains valid at $r=0$ and $t\ge1$, where producing $x_1=F(x_0)$ without querying succeeds with
probability $2^{-n}$.

\begin{remark}[One bit query costs one word query]\label{rem:bit-to-word-query}
Theorem~\ref{thm:cfhl-chain} is stated for queries to a word-valued oracle, whereas
Lemma~\ref{lem:round-embedding} gives algorithms making phase queries to the single-bit functions
$f_{n',a}(z)=F_n(z_{1:n})_a$ defined above, each of which is answered by a word query to $F_n$ at
$y=z_{1:n}$.  The conversion does not increase the number
of rounds.  Write $e(a)=e_a$ for the assigned indices $a\in[n]$ and $e(a)=0^n$ for the dummy and
unassigned indices, whose functions are constant zero.  Prepare the target register, controlled on the
index $a$, in the state $\Had\ket{e(a)}$, and apply the word query
$\ket y\ket s\mapsto\ket y\ket{s\oplus F_n(y)}$.  Its action is
\[
\ket y\,\Had\ket{e(a)}
\ \mapsto\
2^{-n/2}\sum_{s}(-1)^{e(a)\cdot s}\ket y\ket{s\oplus F_n(y)}
=(-1)^{e(a)\cdot F_n(y)}\,\ket y\,\Had\ket{e(a)}.
\]
With $y=z_{1:n}$, the exponent $e(a)\cdot F_n(y)$ equals $F_n(y)_a=f_{n',a}(z)$ at the assigned
indices $a\in[n]$, and $0$ otherwise, in agreement with the constant-zero functions.  One word query
therefore realizes the required phase bit query and restores the target unchanged.  The preparation
and its inverse are oracle-free, so a round of $t$ parallel phase bit queries becomes a round of $t$
parallel word queries.
\end{remark}

\begin{theorem}[Pointer chasing separates the two models]\label{thm:pointer-chasing}
$L_O\in\P^{O}$ for every $O$.  Let $F_n$ be drawn uniformly and independently at each length, and let
$\{\mathcal A_n\}_n$ be a family of quantum algorithms, each using $r$ rounds of $\poly(n)$ parallel
queries per round, where $r+1\le d(n)$ for all large $n$.  Then, with probability $1$ over the draw,
there are only finitely many $n$ at which $\mathcal A_n$ decides $L_O$ with error at most $\tfrac13$
on every input with parameter $n$.
\end{theorem}

\begin{proof}
On input $\langle 1^n,j,i\rangle$ the chain $x_1,\dots,x_j$ is computed by $j\le d(n)$ rounds of $n$
classical queries each, so $L_O\in\P^{O}$.

Fix $n$, let $t_0=\poly(n)$ bound the parallelism of $\mathcal A_n$, and let $E_n$ be the event, over
the draw of $F_n$, that $\mathcal A_n$ decides $L_O$ with error at most $\tfrac13$ on every input
$\langle1^n,j,i\rangle$.  Independent repetitions share the same rounds, so running
$O(\log(d(n)n))$ copies of $\mathcal A_n$ and taking a majority gives an algorithm $\mathcal A'$ with
$r$ rounds whose error on each fixed input is, on $E_n$, below $1/(10\,d(n)n)$.  The $d(n)n$ inputs
with parameter $n$ can all be run simultaneously in the same $r$ rounds, and by a union bound, on
$E_n$, every output is correct with probability at least $\tfrac9{10}$ over the measurements.

Those outputs determine the entire list $x_1,\dots,x_{d(n)}$, and in particular the chain
$x_0\mapsto x_1\mapsto\cdots\mapsto x_{r+1}$ that
Theorem~\ref{thm:cfhl-chain} forbids.  The total parallelism is
$t=O\bigl(t_0\,d(n)\,n\log(d(n)n)\bigr)=\poly(n)$, and by Remark~\ref{rem:bit-to-word-query} these
are $t$ parallel word queries to $F_n$ in the same $r$ rounds.  The algorithm that runs
$\mathcal A'$ and outputs the chain therefore succeeds with probability at least
$\tfrac9{10}\Prb[E_n]$ over the draw of $F_n$ and its own randomness, and
Theorem~\ref{thm:cfhl-chain} bounds that probability by $O\bigl(t^3(r+1)^3/2^n\bigr)$.  Hence
$\Prb[E_n]=O\bigl(t^3(r+1)^3/2^n\bigr)=2^{-\Omega(n)}$, which is summable in $n$, and the
Borel--Cantelli lemma gives the claim.
\end{proof}

\begin{corollary}[The containment fails at every level]\label{cor:std-not-in-phase}
For every constant $j\ge1$ there is an oracle $O$ with
$\FH_{0,\mathrm{std}}^{O}\not\subseteq\FH_{j}^{O}$.  In particular, for every $k\ge0$ there is an
oracle with $\FH_{k,\mathrm{std}}^{O}\not\subseteq\FH_{k+1}^{O}$.
\end{corollary}

\begin{proof}
Take $d(n)=j$ in \eqref{eq:pointer-chasing-language}.  By Lemma~\ref{lem:round-embedding} an
$\FH_j^{O}$ circuit making $\poly(n)$ phase queries is simulated exactly by $r=j-1$ rounds of
$\poly(n)$ parallel queries, and $r+1=j\le d(n)$, so by Theorem~\ref{thm:pointer-chasing}, for $F=(F_n)_n$ drawn uniformly, each fixed uniform
$\FH_j$ family almost surely decides $L_O$ at only finitely many lengths.  The uniform families are
countable, so almost surely this holds for all of them at once, and any such draw gives an oracle
$O$ with $L_O\notin\FH_j^{O}$; meanwhile $L_O\in\P^{O}=\FH_{0,\mathrm{std}}^{O}$ for every $O$ by
Lemma~\ref{lem:fh0-std}.  Taking $j=k+1$ and using
$\FH_{0,\mathrm{std}}^{O}\subseteq\FH_{k,\mathrm{std}}^{O}$ gives the second statement.
\end{proof}

\begin{corollary}[The phase hierarchy does not contain $\P$]\label{cor:p-not-in-fh}
There is an oracle $O$ with $\P^{O}\not\subseteq\bigcup_{k\ge0}\FH_k^{O}$.
\end{corollary}

\begin{proof}
Take $d(n)=\max\{1,\lceil\log_2 n\rceil\}$, so that $L_O\in\P^{O}$ still.  For each constant $j$ the round budget
$r=j-1$ of Lemma~\ref{lem:round-embedding} satisfies $r+1=j\le d(n)$ at all large $n$, so by Theorem~\ref{thm:pointer-chasing} each fixed uniform family of any constant depth $j$ almost
surely decides $L_O$ at only finitely many lengths.  The pairs consisting of a family and a level are
countable, so a single draw of $F$ almost surely defeats all of them.
\end{proof}

Two distinct advantages are involved here.  The advantage that
Corollary~\ref{cor:std-not-in-phase} exhibits is matched by no constant number of Hadamard layers, and its cause is unrelated to interference: the standard-query
model contains adaptive classical computation and the phase model does not.
Theorem~\ref{thm:phase-vs-standard} is the separation that remains once this advantage is removed.
Its separating problem, the OR of $(2k-1)$-fold Forrelation, is classically hard~\cite{bansal2021k}, so an adaptive classical stage does not decide it; Proposition~\ref{prop:classical-prefix-separation} makes this
precise by keeping the separation in force even when the phase model is given an arbitrary classical
preprocessing stage.

Shi~\cite{Shi_2005} introduced the Fourier hierarchy, conjectured adjacent strictness, and observed
that Simon's problem separates the first two levels relative to an oracle;
Theorem~\ref{thm:standard-query-adjacent} settles the adjacent oracle question at every remaining
fixed level $k\ge2$.  The round simulation of Proposition~\ref{prop:three-round-bit-embedding} is
useful for comparing the two models, but it is not what gives that separation.  The tool that does is
the standard-query growth theorem (Theorem~\ref{thm:standard-query-growth}), in which an adaptive
standard-query block is factored through a leaf of its decision tree, where the entire queried path
is determined.  This handles both the coherent classical adaptivity and the input-dependent diagonal
contribution of the final block, without converting the circuit into parallel-query rounds of
exponential width.

\begin{remark}[The lowest levels]\label{rem:low-levels}
In the phase-query model the separation between the first two levels is immediate: by
Lemma~\ref{lem:round-embedding}, the acceptance probability of an $\FH_1$ computation with
polynomially many queries does not depend on the oracle, whereas $\FH_2$ already detects the squared mean
$(N^{-1}\sum_xF_0(x))^2$ of a single oracle: a SWAP test between $H^{\otimes n}\ket{0^n}$ and
$P_{F_0}H^{\otimes n}\ket{0^n}$ has Fourier depth two and acceptance probability
$\tfrac12(1+(N^{-1}\sum_xF_0(x))^2)$ by Lemma~\ref{lem:swap-test}, which distinguishes the all-ones
function from a balanced one.  Hence
$\FH_1^{O}\ne\FH_2^{O}$ for this direct reason.  The second level is not weak, however: it decides
$2$-Forrelation (Proposition~\ref{prop:even-forrelation}) and therefore lies outside the polynomial
hierarchy relative to an oracle (Corollary~\ref{thm:fh2-not-in-ph}).  In the standard-query model,
$\FH_1$ computations are exactly randomized classical query algorithms (Lemma~\ref{lem:fh1-std}), and
Simon's problem separates $\FH_1$ from $\FH_2$~\cite{Shi_2005,365701}.
Theorems~\ref{thm:all-level-separation} and~\ref{thm:standard-query-adjacent} concern the levels
$k\ge2$, for which, to our knowledge, no separation was previously known in either query model.
\end{remark}

\section{Consequences for relativized classes}\label{sec:landscape}

The separations proved so far are internal to the two hierarchies.  This section places them among the
surrounding classes.  We locate the phase hierarchy relative to the polynomial hierarchy
(Section~\ref{subsec:fh2-vs-ph}) and to $\BQP$ (Section~\ref{subsec:fh-vs-bqp}), and then
show that the adjacent-level question does not relativize
(Section~\ref{subsec:collapse}), so that no oracle construction of the preceding sections can settle
Shi's conjecture itself.

\subsection{The second level and the polynomial hierarchy}\label{subsec:fh2-vs-ph}

We first compare the hierarchy with the polynomial
hierarchy, and show that the two cross between the first and the second level.  Two statements are
needed, and both follow from what we have already proved.  Lemma~\ref{lem:fh1-std} identifies the
first standard-query level with $\BPP$, which is contained in $\PH$ relative to every oracle,
while Corollary~\ref{thm:fh2-not-in-ph} places the second level outside it.  Throughout, $\PH^{O}$ denotes the relativized polynomial hierarchy
under the length-preserving query convention of Definition~\ref{def:FHk-rel}.

For the second level we use the oracle separation of Raz and Tal in the following form.  Recall
from~\eqref{eq:even-Phi} that $\Phi_2(F_0,F_1)=N^{-3/2}\sum_{x,y}F_0(x)(-1)^{x\cdot y}F_1(y)$ is
twofold Forrelation.

\begin{theorem}[{Raz--Tal~\cite[Theorem 1.1]{raz2019oracle}}]\label{thm:raz-tal}
For each $n$ there is an explicitly samplable distribution $\mathcal D_n$ on pairs
$(F_0,F_1)$ of Boolean functions $\bits^n\to\pmone$, with $\mathcal U_n$ the uniform distribution on
such pairs, such that, writing $N=2^n$:
\begin{enumerate}[leftmargin=2em]
\item $\displaystyle \E_{\mathcal D_n}[\Phi_2]=\Omega(1/n)$, while $\E_{\mathcal U_n}[\Phi_2]=0$;
\item no Boolean circuit of quasipolynomial size and constant depth over the $2N$ input bits
distinguishes $\mathcal D_n$ from $\mathcal U_n$ with advantage better than
$\mathrm{polylog}(N)/\sqrt N$.
\end{enumerate}
\end{theorem}

Item~(1) is their Theorem~1.1(1), stated in terms of the Aaronson--Ambainis one-query test, whose
acceptance probability is $\tfrac12(1+\Phi_2)$.  Since the mean under the uniform distribution
vanishes, an advantage of $\Omega(1/\log N)$ for that test is exactly the stated bound on
$\E_{\mathcal D_n}[\Phi_2]$, with an explicit constant coming from their parameter
$\varepsilon=1/(24\ln N)$.  Item~(2) is their Theorem~1.1(2).  The theorem and corollary numbers here
follow the full version, ECCC TR18-107.

The following corollary is due to Buzet and Chailloux~\cite{buzet2026iqp}; we state it here in the
form we use.  They show that $2$-Forrelation is solved by an IQP computation, that is, by constantly
many IQP circuits together with efficient classical pre- and post-processing, and deduce that
$(\BPP^{\mathrm{IQP}})^{O}\not\subseteq\PH^{O}$.  Their computation is an $\FH_2$ circuit, so
the same oracle separates $\FH_2$ from $\PH$.

\begin{corollary}[{The second Fourier level is not contained in $\PH$; Buzet--Chailloux~\cite{buzet2026iqp}}]\label{thm:fh2-not-in-ph}
There is an oracle $O$ and a unary language $L\in\FH_2^{O}$ with $L\notin\PH^{O}$.  Hence
$\FH_2^{O}\not\subseteq\PH^{O}$, and the same holds in the standard-query model.
\end{corollary}

\begin{proof}
An IQP circuit has the form $H^{\otimes m}DH^{\otimes m}$ with $D$ diagonal in the computational
basis, which is a circuit of exact Fourier depth two whose single interior block is
basis-preserving.  The computation of~\cite{buzet2026iqp} consists of polynomially many such circuits,
each containing the indexed phase oracle $O_{f,g}$ of Definition~\ref{def:FHk-rel} as a diagonal
gate and applied to independently planted oracle blocks, together with classical processing that
tosses coins to select among circuits and thresholds the measurement outcomes.

All of this is available at Fourier depth two: the circuits run in parallel on disjoint registers and share the two
Hadamard layers (Remark~\ref{rem:register-routing}); a coin is a designated wire placed in $\ket+$ by
the first layer, a diagonal gate controlled on it is again diagonal, and the wire is routed past the
second layer so that its value is recorded by the measurement; and the threshold is an accepting
set recognizable in polynomial time (Definition~\ref{def:FH2FH3}).  The whole computation is therefore a uniform $\FH_2$ circuit, so
the language witnessing $(\BPP^{\mathrm{IQP}})^{O}\not\subseteq\PH^{O}$
in~\cite{buzet2026iqp} lies in $\FH_2^{O}$.  Finally
$\FH_2^{O}\subseteq\FH_{2,\mathrm{std}}^{O}$ by Remark~\ref{rem:phase-vs-bit}, which gives the
standard-query statement.
\end{proof}

\begin{remark}[{A second route to Corollary~\ref{thm:fh2-not-in-ph}}]\label{rem:fh2-ph-alternative}
Proposition~\ref{prop:even-forrelation} gives the same conclusion without the quadratic identity
behind the IQP construction, at a weaker gap.  By that proposition with $k=2$ there is a uniform
$\FH_2$ circuit making two phase queries that accepts a pair $(F_0,F_1)$ with probability
$p(F_0,F_1)=\tfrac12(1+\Phi_2(F_0,F_1)^2)$.  Orthogonality of the characters gives
$\E_{\mathcal U_n}[\Phi_2^2]=1/N$, while Jensen's inequality and Theorem~\ref{thm:raz-tal}(1) give
$\E_{\mathcal D_n}[\Phi_2^2]\ge(\E_{\mathcal D_n}[\Phi_2])^2=\Omega(1/n^2)$, so
$\gamma_n:=\E_{\mathcal D_n}[p]-\E_{\mathcal U_n}[p]\ge c/n^2$ for an explicit constant $c>0$ and all large $n$.  This
is a gap between ensemble means, not a pointwise gap, so it is amplified over independent instances
rather than by repetition on one instance.  Plant at each length $T=\lceil n^5\rceil$ independent
pairs, run the circuit once on each pair through the same two Hadamard layers
(Remark~\ref{rem:register-routing}), and accept if at least $T(\E_{\mathcal U_n}[p]+c/(2n^2))$
runs accept, where $\E_{\mathcal U_n}[p]=\tfrac12(1+1/N)$ is known exactly.  Under $\mathcal D_n^{\otimes T}$ the $T$ acceptance indicators are independent with
mean $\E_{\mathcal D_n}[p]$, since run $j$ depends only on pair $j$ and on its own measurement, and
likewise under $\mathcal U_n^{\otimes T}$ with mean $\E_{\mathcal U_n}[p]$; by Hoeffding's
inequality (Lemma~\ref{lem:hoeffding}) the circuit therefore distinguishes
$\mathcal D_n^{\otimes T}$ from $\mathcal U_n^{\otimes T}$ with error
$e^{-Tc^2/(2n^4)}=2^{-\Omega(n)}$.  On the classical side, a hybrid argument over the $T$
coordinates turns a constant-depth circuit distinguishing the two product distributions with
advantage $\alpha$ into one distinguishing $\mathcal D_n$ from $\mathcal U_n$ with advantage
$\alpha/T$, so Theorem~\ref{thm:raz-tal}(2) bounds $\alpha$ by $T\,\mathrm{polylog}(N)/\sqrt N=2^{-\Omega(n)}$,
and the planting argument of~\cite[Corollary~1.5]{raz2019oracle} applies to the product
distributions.  The gap is $\Omega(1/n^2)$ rather than the $\Omega(1/n)$ of~\cite{buzet2026iqp}
because the circuit computes $\Phi_2^2$ and not $\Phi_2$.  The argument uses only
Proposition~\ref{prop:even-forrelation}, which is needed for general even $k$ in any case.
\end{remark}

Lemma~\ref{lem:fh1-std} and Corollary~\ref{thm:fh2-not-in-ph} together determine where the crossing
occurs: $\FH_{1,\mathrm{std}}^{O}=\BPP^{O}\subseteq\PH^{O}$ for \emph{every} oracle, whereas
$\FH_2^{O}\not\subseteq\PH^{O}$ for some oracle.  One further Hadamard layer therefore takes
the hierarchy outside the polynomial hierarchy.  By Lemma~\ref{lem:round-embedding} the algorithm
that achieves this is moreover non-adaptive: a depth-two circuit is simulated exactly by one round of
parallel queries, so one non-adaptive batch of parallel quantum queries, followed by an oracle-independent
unitary and a measurement, already lies outside $\PH$ relative to an oracle.

Corollary~\ref{cor:p-not-in-fh} places the phase-query model relative to $\P$ as well.
  That corollary gives an oracle relative to which $\P\not\subseteq\bigcup_k\FH_k$, and
interleaving the two constructions on disjoint blocks of input lengths, as in
Remark~\ref{rem:tower}, gives both statements relative to a single oracle.  The phase-query hierarchy
is therefore incomparable with $\P$ relative to an oracle: it lies outside the polynomial hierarchy
at its second level, yet it does not contain $\P$ at any level.  Both statements have the same cause,
namely that a phase query acts trivially on a basis state, which gives interference but not classical
adaptivity.  This is one reason to regard Shi's standard-query model as the primary one, and
Theorem~\ref{thm:standard-query-adjacent} rather than Theorem~\ref{thm:all-level-separation} as the
answer to the question he asked; in that model $\P^{O}\subseteq\FH_{k,\mathrm{std}}^{O}$ holds at
every level by Lemma~\ref{lem:fh0-std}.

We include Corollary~\ref{thm:fh2-not-in-ph}, which is not our result, because it determines the
position of the whole hierarchy relative to a standard classical class, which none of the internal
separations does.

\subsection{Bounded Fourier depth versus \texorpdfstring{$\BQP$}{BQP}}\label{subsec:fh-vs-bqp}

The level separations give a strictly increasing sequence of classes inside the hierarchy.  We now
show that this sequence is strictly contained in $\BQP$, by exhibiting a single promise problem that
is decided with polynomially many quantum queries but at no Fourier depth up to
$\tfrac14\log_2 n$.  Here $\BQP^{O}$ denotes the class of
promise problems decided with bounded error by a uniform family of quantum circuits making
$\poly(n)$ phase queries to $O$, with no bound on the number of Hadamard layers.  Every
$\FH_k$ circuit is such a circuit, so $\FH_k^{O}\subseteq\BQP^{O}$ for every $k$.

Here the Forrelation order must grow with the input.  Aaronson and
Ambainis~\cite[Theorem~5]{aaronson2014forrelationproblemoptimallyseparates} showed that $K$-fold
Forrelation with $K=\poly(n)$, presented by explicit circuits, is $\mathrm{PromiseBQP}$-complete, so
the order $K$ is the parameter along which Forrelation captures the power of quantum computation.
Bounded Fourier depth bounds the order that a circuit can handle
(Remark~\ref{rem:even-tightness}), and it is by letting $K$ grow that we separate the two.

The hard problem is Forrelation whose order grows with the input.  Fix
$K=K(n):=\max\{2,\lceil\log_2 n\rceil\}$ and, at each length $n$, one instance $Z_n$ of $\mathrm{ORF}_{K(n),n}$
from Definition~\ref{def:hard-problem}, read with the growing parameter $K(n)$.  Since $K(n)\to\infty$
slowly, the gap $\delta=2^{-5K}=n^{-5+o(1)}$ and the copy count $m=2^{5K}=n^{5+o(1)}$ stay polynomial,
and $M=mKN=2^{n+O(\log n)}$.

\begin{definition}[Variable Fourier depth]\label{def:variable-depth}
For a nondecreasing, polynomial-time computable schedule $d:\mathbb N\to\mathbb N$, the class
$\FH^{O}_{d(\cdot)}$ consists of the promise problems solvable, under the conventions of
Definition~\ref{def:FHk-rel}, by a uniform circuit family whose length-$n$ member uses at most $d(n)$
global Hadamard layers; the standard-query analogue $\FH^{O}_{d(\cdot),\mathrm{std}}$ is defined in the
same way with standard queries.  Since $d$ is computable and the gate set is finite (Definition~\ref{def:FH2FH3}), the families
obeying a given schedule can be enumerated up to normalization: enumerate all clocked generators
and replace, at each length, an output that is malformed or uses more than $d(n)$ layers by a fixed
rejecting circuit.  The normalized family obeys the schedule and coincides with the original whenever
the original does, so every family obeying the schedule is one of the normalized families, and the
diagonalization of Theorem~\ref{thm:all-level-separation} can be applied to them.
\end{definition}

\begin{theorem}[Logarithmic Fourier depth is weaker than $\BQP$]\label{thm:fh-vs-bqp}
There is an oracle $O$ and a language $L\in\BQP^{O}$ such that $L\notin\FH^{O}_{d(\cdot)}$ for every
nondecreasing, polynomial-time computable schedule $d$ with $d(n)\le\tfrac14\log_2 n$
(Definition~\ref{def:variable-depth}).  In particular $L\notin\FH_k^{O}$ for every constant $k$, and
\[
\Bigl(\textstyle\bigcup_{k}\FH_k\Bigr)^{O}\ \subsetneq\ \BQP^{O}.
\]
\end{theorem}

\begin{proof}
Let $L=\{1^n:Z_n\in\Pi^{\Yes}_{K(n),n}\}$.

\emph{Upper bound.}  A $\BQP$ circuit has no Fourier-depth budget, so at length $n$ it runs
the $K$-fold Forrelation expression directly.  On one instance $z$, the circuit
$\Had P_{F_{K-1}}\Had\cdots\Had P_{F_0}\Had\ket{0^n}$ places amplitude $\mathsf{forr}_K(z)$ on $\ket{0^n}$
by \eqref{eq:forr-dictionary}, so a computational-basis measurement returns $0^n$ with probability
$\mathsf{forr}_K(z)^2$.  For each of the $m$ copies, estimate this probability to additive
$\delta^2/8$ from $O(\delta^{-4}\log m)$ repetitions, and accept if some estimate exceeds
$\delta^2/2$.  A $\Yes$ instance has a copy with $\mathsf{forr}_K\ge\delta$, hence probability
$\ge\delta^2$; a $\No$ instance has every copy with $\abs{\mathsf{forr}_K}\le\delta/2$, hence probability
$\le\delta^2/4$.  A union bound over the $m$ copies makes the decision correct with bounded error,
using $O(m\,\delta^{-4}\log m)=\poly(n)$ repetitions of an instance with
$K=O(\log n)$ queries.
Thus $L\in\BQP^{O}$.

\emph{Lower bound.}  Consider a uniform family $A$ making $q(n)=\poly(n)$ queries whose circuit at
length $n$ uses $k=k(n)$ Hadamard layers, where $2\le k(n)\le\tfrac14\log_2 n$; then
$4(k-1)\le\lceil\log_2 n\rceil=K$, so in particular $2(k-1)<K$.  (If $k(n)\le1$, the acceptance
probability at that length does not depend on the oracle by Lemma~\ref{lem:round-embedding}, so the
family fails on one of the two promised instances.)  By Lemma~\ref{lem:round-embedding},
Corollary~\ref{cor:FHk-growth} and Theorem~\ref{thm:bs-bias}, exactly as in the proof of
Theorem~\ref{thm:lower-worst-case} (whose Step~4 uses only the gap $2(k-1)<K$; cf.\
Remark~\ref{rem:parametric-lower-bound}), were $A$ to decide $\mathrm{ORF}_{K(n),n}$ with bounded
error at length $n$, some single-copy restriction $h$ would satisfy
\[
\frac1{4m}\ \le\ \bigl|\E_{\mathcal F_K}[h]-\E_{\mathcal U}[h]\bigr|
\ \le\ \sum_{\ell=K}^{K(K-1)} N^{-\frac\ell2\left(1-\frac1K\right)}(8K)^{14\ell}A_\ell,
\qquad
A_\ell=4^{\ell}(2^{2k-2}-1)^{2\ell}q^{\ell}M^{\frac\ell2\left(1-\frac1{2k-2}\right)}.
\]
Writing $M=mKN$, the net power of $N$ in the term $\ell$ is $N^{-\frac\ell2\beta_k}$ with
$\beta_k=\frac1{2k-2}-\frac1K\ge\frac1{4(k-1)}$, the last inequality because $4(k-1)\le K$.  With
$k\le\tfrac14\log_2n$ and $\ell\le K(K-1)=O(\log^2 n)$, the factor
$(2^{2k-2}-1)^{2\ell}\le2^{4k\ell}\le2^{\ell\log_2 n}$, together with
$4^{\ell},(8K)^{14\ell},q^{\ell}$ and $(mK)^{\ell}$, is at most $2^{O(\log^3 n)}=:Q(n)$, with a
constant depending on the polynomial query bound of the family, as $q,mK=\poly(n)$ and $\ell=O(\log^2n)$.  Each term is thus at most
$Q(n)N^{-\frac\ell2\beta_k}\le Q(n)\,2^{-nK\beta_k/2}\le Q(n)\,2^{-n/2}$, using
$K\beta_k\ge K/(4(k-1))\ge1$, and the $O(\log^2 n)$ terms sum to $2^{-n/2+O(\log^3n)}\le2^{-n/3}$,
below $\frac1{4m}=\Theta(n^{-5})$.  Hence $A$ fails on some promise instance at every large length at
which its depth respects the cap.

\emph{Diagonalization.}  Enumerate all clocked uniform families $A^{(i)}$, with malformed outputs
replaced by a rejecting circuit, and assign to each family infinitely many dedicated lengths, all so
large that the bound above is below $\frac1{4m}$: for instance the lengths $n_{\langle i,j\rangle}$,
$j\ge1$, of an increasing sequence indexed by pairs.

At a dedicated length $n$ of $A^{(i)}$, if the
circuit of $A^{(i)}$ uses more than $\tfrac14\log_2 n$ Hadamard layers, plant any promised instance.
Otherwise the bound above, or oracle-independence when the depth is at most one, provides a promised
instance on which $A^{(i)}_{n}$ fails; plant it.  At every other length plant any instance of
$\Pi^{\Yes}_{K(n),n}\cup\Pi^{\No}_{K(n),n}$, nonempty by
Proposition~\ref{prop:promise-density} (for $\Pi^{\No}$) and Theorem~\ref{thm:bs-distributions}(2)
(for $\Pi^{\Yes}$, since $\Pr_{\mathcal F_{K(n)}}[\mathrm{forr}_{K(n)}\ge\delta]\ge6\delta>0$, which
holds for every $K(n)\ge2$, odd or even).  Every family obeying some admissible schedule is enumerated and fails at each of its dedicated
lengths, while every planted instance is promised, so the decider above is correct at every length.
Hence $L\in\BQP^{O}$ and $L\notin\FH^{O}_{d(\cdot)}$ for every admissible schedule $d$.  For the
constant-depth consequence, if a uniform $\FH_k$ family decided $L$, then replacing its members at
the finitely many lengths $n<2^{4k}$ by depth-zero circuits with the correct answers hardwired gives a
family obeying the schedule $d(n)=\lfloor\tfrac14\log_2n\rfloor$ that decides $L$, a contradiction.
\end{proof}

\begin{corollary}[A $\BQP$ problem of prescribed Fourier depth]\label{cor:prescribed-fdepth}
There is an absolute constant $c>0$ with the following property.  Let
$d:\mathbb N\to\mathbb N$ be nondecreasing and polynomial-time computable, with
$2\le d(n)\le c\log n$ for all sufficiently large $n$.  Then there are an oracle $O$ and a unary
language $L_O$, whose length-$n$ instance is $\mathrm{ORF}_{2d(n)-1,n}$ at every length where
$d(n)\ge2$, such that
\[
L_O\in\FH^{O}_{d(\cdot)+1}\setminus\FH^{O}_{d(\cdot)},
\]
and $L_O$ is decided in $\BQP^{O}$ using $\poly(n)$ queries, or $n^{o(1)}$ queries when
$d(n)=o(\log n)$.  If moreover $d(n)\to\infty$, then
$\bigl(\bigcup_j\FH_j\bigr)^{O}\subseteq\FH^{O}_{d(\cdot)}\subsetneq\BQP^{O}$, refining
Theorem~\ref{thm:fh-vs-bqp}.
\end{corollary}

\begin{proof}
By Definition~\ref{def:variable-depth} the uniform families obeying the schedule $d$ are effectively
enumerable.  Running the diagonalization of Theorem~\ref{thm:all-level-separation} over these, and
using at length $n$ the depth-$d(n)$ lower bound of Corollary~\ref{cor:growing-k} (valid since
$d(n)\le c\log n\le c'n^{1/3}$ for large $n$), fixes an oracle $O$ with
$L_O\notin\FH^{O}_{d(\cdot)}$; the depth-$(d(n)+1)$ circuit of Corollary~\ref{cor:growing-k} has
polynomial size in this range and decides $L_O$, so $L_O\in\FH^{O}_{d(\cdot)+1}$.  For the $\BQP$
bound, the gap $\delta=2^{-5(2d(n)-1)}=2^{-\Theta(d(n))}$ gives copy count $m=\delta^{-1}$ and
repetition count $\delta^{-4}$, so the direct estimator of Theorem~\ref{thm:fh-vs-bqp} makes
$O\bigl(K\,m\,\delta^{-4}\log m\bigr)=2^{\Theta(d(n))}$ queries, polynomial for $d(n)\le c\log n$
with $c$ small and $n^{o(1)}$ for $d(n)=o(\log n)$.  At the finitely many lengths with $d(n)<2$, plant no instance and hardwire the answer into both
families.  Finally, if $d(n)\to\infty$ then for every constant $j$ we have $d(n)\ge j$ for all large
$n$, so $\FH_j\subseteq\FH_{d(\cdot)}$ as classes of promise problems, again after hardwiring the
finitely many smaller lengths.  The last inclusion is strict for this
same oracle by the two facts already established here, namely $L_O\in\BQP^{O}$ and
$L_O\notin\FH^{O}_{d(\cdot)}$, and not by Theorem~\ref{thm:fh-vs-bqp}, which constructs a different
oracle.
\end{proof}

\begin{remark}[The chain of inclusions]\label{rem:tower}
Together with the level separations, Theorem~\ref{thm:fh-vs-bqp} gives the strict relativized chain
of inclusions
\[
\FH_2^{O}\subsetneq\FH_3^{O}\subsetneq\FH_4^{O}\subsetneq\cdots\subsetneq
\Bigl(\textstyle\bigcup_{k}\FH_k\Bigr)^{O}\subsetneq\BQP^{O};
\]
interleaving the construction of Corollary~\ref{cor:single-oracle-all-levels} with that of
Theorem~\ref{thm:fh-vs-bqp} on disjoint length blocks realizes it relative to a single oracle, and
the lowest levels are as in Remark~\ref{rem:low-levels}.  The two hard problems used here are of
opposite types.  The growing order $K=\lceil\log_2 n\rceil$ makes the Fourier depth of the hard
problem as large as possible, up to $\Theta(\log n)$, and rules out every depth up to
$\tfrac14\log_2 n$, whereas the minimal order $K=2d-1$ of Corollary~\ref{cor:prescribed-fdepth}
fixes the Fourier depth of the problem at $d+1$ while ruling out depth $d$ up to $c\log n$.  With the constant of
Appendix~\ref{app:gstw-growth-constant} the growth bound is singly exponential and supports depth up
to $cn^{1/3}$ in the lower bound.  The ceiling in the latter is therefore no longer that constant but
the deciding circuit itself, whose $2^{\Theta(d)}$ repetitions must fit into a polynomial-size
family.  The same
argument with the standard-query growth bound (Theorem~\ref{thm:standard-query-growth}) gives
$\bigl(\bigcup_k\FH_{k,\mathrm{std}}\bigr)^{O}\subsetneq\BQP^{O}$: by
Lemma~\ref{lem:standard-cert-count} the constant there satisfies $A'_{k,\ell}=2^{O_k(\ell)}$, so with
$K=\lceil\log_2 n\rceil$, $\ell\le K(K-1)=O((\log n)^2)$, and $q=\poly(n)$, all elementary factors
are at most $2^{O_k((\log n)^{3})}$, dominated by the $2^{-\Omega_k(n\log n)}$ decay at the growing
Forrelation order.
\end{remark}

\subsection{A collapsing oracle: the adjacent-level question does not relativize}\label{subsec:collapse}

All the separations proved above are relativized, in the sense that they exhibit an oracle
relative to which adjacent levels differ.  We close with an oracle in the opposite
direction, relative to which every level from the second onwards equals $\mathrm{PSPACE}$, in both
query models.  Taken together, the two directions show that Shi's unrelativized conjecture cannot be
settled by a relativizing argument.

For this subsection only, we use the standard relativization convention: the oracle is a single
language $O\subseteq\bits^*$, exposed as the family of its slices $O\cap\bits^m$, and a circuit on
an input of length $n$ may query slices at every length $m$ up to its polynomial size, through the
phase or standard gates of Definition~\ref{def:FHk-rel}.  We used the length-preserving convention in
Definition~\ref{def:FHk-rel} because the diagonalizations there fix one instance at each length;
Proposition~\ref{rem:standard-convention-separations} shows that the separations also hold under the
standard convention, so that the two directions concern the same model.  Under this convention
$\BQP^{O}$ denotes the class of promise problems decided with bounded error by uniform
polynomial-size quantum circuit families with standard oracle gates.

\begin{proposition}[Collapse relative to a $\mathrm{PSPACE}$-complete oracle]\label{prop:pspace-collapse}
Let $O$ be a $\mathrm{PSPACE}$-complete language; replacing $O$ by $\{1y:y\in O\}$ if necessary,
assume no member of $O$ begins with $0$.  Under the standard convention,
\[
\FH_{k,\mathrm{std}}^{O}\ =\ \P^{O}
\quad\text{for every }k\ge0,
\qquad
\FH_{k}^{O}\ =\ \FH_{2}^{O}\ =\ \BQP^{O}
\quad\text{for every }k\ge2
\]
in the standard-query and phase models respectively, and as language classes all of these equal
$\mathrm{PSPACE}$.  In the phase model the statement begins at the second level because the first
is degenerate: $\FH_1^{O'}=\FH_1$ for every oracle $O'$ (Lemma~\ref{lem:round-embedding}), so
whether $\FH_1^{O}$ also equals $\FH_2^{O}$ relative to this oracle is the unrelativized question
of whether $\FH_1=\mathrm{PSPACE}$, which we do not address.
\end{proposition}

\begin{proof}
The proof is the standard collapse of $\BQP$ to $\P$ relative to a
$\mathrm{PSPACE}$-complete oracle~\cite{bernstein1997quantum}, carried out level by level.  The one
point that requires care is that a phase query must be able to read the oracle from the second level
onwards.

\emph{Upper inclusions.}  Every class above is contained in $\BQP^{O}$: an $\FH_k$
circuit is a polynomial-size quantum circuit with oracle gates, and a phase gate is one standard
gate by kickback.  For a uniform family with oracle $O$, the acceptance probability on a given
input is computed to precision $2^{-n-3}$ in polynomial space by the path-sum
method of~\cite{bernstein1997quantum}, evaluating each oracle answer $O(y)$, for the queried string
$y$ of polynomial length, in polynomial space using $O\in\mathrm{PSPACE}$ and reusing the space
across paths.  Hence for every promise problem in $\BQP^{O}$ there is a total language
$L'\in\mathrm{PSPACE}$, namely the inputs on which the computed acceptance probability exceeds
$\tfrac12$, that contains every $\Yes$ instance and no $\No$ instance.

\emph{Lower inclusions.}  Let $L'\in\mathrm{PSPACE}$ and let $g$ be a polynomial-time reduction
with $x\in L'\iff g(x)\in O$.  A deterministic polynomial-time machine decides $L'$ with one
query, so $\mathrm{PSPACE}\subseteq\P^{O}$, and $\P^{O}=\FH_{0,\mathrm{std}}^{O}\subseteq
\FH_{k,\mathrm{std}}^{O}$ by Lemma~\ref{lem:fh0-std}, whose proof does not use the access
convention.  Combined with the upper inclusions, every promise problem in
$\FH_{k,\mathrm{std}}^{O}$ is decided by a $\P^{O}$ machine through its language $L'$, which gives
the standard-query equalities.

For the phase model, it suffices to decide any $L'\in\mathrm{PSPACE}$ at Fourier depth two.  The
first global Hadamard layer prepares a control qubit $c$ in $\ket+$.  The interior block
reversibly computes, by Bennett's method, the string $g(x)$ into the query register on the $c=1$
branch and the string $0^{\abs{g(x)}}$ on the $c=0$ branch; both maps are basis-preserving.  It then
makes one phase query and uncomputes the routing.  The $c=0$ branch queries a string beginning with
$0$, which is not in $O$ by assumption, so it acquires phase $+1$; the $c=1$ branch acquires
$(-1)^{O(g(x))}$.  The second Hadamard layer, applied to $c$, therefore maps the control to the
basis state $\ket{O(g(x))}$, and the measurement decides $L'$ with certainty.  Hence
$\mathrm{PSPACE}\subseteq\FH_2^{O}$, and with the upper inclusions
$\FH_2^{O}=\FH_k^{O}=\BQP^{O}$ for every $k\ge2$, all equal to $\mathrm{PSPACE}$ as
language classes.
\end{proof}

\begin{proposition}[The separations under the standard convention]\label{rem:standard-convention-separations}
Theorems~\ref{thm:all-level-separation}, \ref{thm:standard-query-adjacent}
and~\ref{thm:phase-vs-standard} remain true under the standard convention, and a single language
oracle realizes all three separations at every level $k\ge2$ at once.
\end{proposition}

\begin{proof}
\emph{The encoding.}
The indexed functions of a block at length $n$ are encoded in
the slice of $O$ at length $n+s(n)+1$, with $s(n)=\lceil\log_2n\rceil$: the string $1ay$ with
$a\in\bits^{s(n)}$ and $y\in\bits^n$ lies in $O$ exactly when $f_{n,a}(y)=1$, and no string beginning
with $0$ lies in $O$.  Only the diagonalization changes, from an adversarial choice at each length to
a probabilistic one, since a circuit may now query slices at lengths other than its own.

\emph{The random planting.}
Enumerate the triples $(\mu,k,i)$, where $\mu$ is one of the two query models, $k\ge2$ is a level,
and $i$ is a clocked machine of that model with clock bound $n^{c_i}$ on its circuit size, which
also bounds the length of every string it queries.  Choose
designated lengths $n_1<n_2<\cdots$ with $n_{j+1}=2^{n_j}$, assigning infinitely many designated
lengths to every triple by a polynomial-time computable rule, each triple receiving only lengths
large enough for its block to fit in the index width $s(n)$.  For a fixed machine, $n_j^{c_i}<n_{j+1}$
for all large $j$, so at all but finitely many of its designated lengths the machine queries no slice
at length $n_{j+1}$ or beyond.  At a length $n_j$ designated for $(\mu,k,i)$ we toss a fair coin and
plant a level-$k$ block of the model $\mu$, with $K=2k-1$ or $K=2k+1$ accordingly, drawn either from
$\mathcal F_K^{\otimes m}$ conditioned on $\Pi^{\Yes}$ or from $\mathcal U^{\otimes m}$ conditioned
on $\Pi^{\No}$, independently of everything else; at every other length we plant the all-ones
instance, which lies in $\Pi^{\Yes}$ at every level.

\emph{At a designated length the machine has small advantage, whatever was planted before.}
Fix a triple $(\mu,k,i)$ and one of its designated lengths $n_j$, and condition on all blocks at
lengths below $n_j$.  The acceptance probability of machine $i$ at input $1^{n_j}$ is then a
restriction of its acceptance function in which only the block at $n_j$ remains free, so
Corollary~\ref{cor:error-robust}, whose single-block bound holds uniformly over restrictions, bounds
its advantage between the two unconditioned draws by $o_k(1)$.  Here the restriction clause is used
in the free-variable form of~\cite[Theorem~4.1]{girish2024power}: the growth bound for a restricted
function depends only on its number of free coordinates, which is $mKN$ however many oracle bits at
other lengths have been fixed.  Passing to the conditioned draws changes this by at most twice the
conditioning mass, which is at most $2e^{-6}/(1-e^{-6})+o(1)\le\tfrac1{100}$ by
Proposition~\ref{prop:promise-density}.  \emph{Hence it fails at infinitely many of them.}
Write $p_1$ and $p_2$ for the conditional probabilities that
the machine is correct in the bounded-error sense on a $\Yes$ draw (acceptance probability at least
$\tfrac23$) and on a $\No$ draw (at most $\tfrac13$).  Since the acceptance probability is at least
$\tfrac23$ with probability $p_1$ under the first draw and at most $1-\tfrac23p_2$ in expectation
under the second, the advantage is at least $\tfrac23(p_1+p_2)-1$, so
$p_1+p_2\le\tfrac32(1+\tfrac1{100}+o_k(1))$ and the conditional probability that the error exceeds
$\tfrac13$ at length $n_j$ is $1-\tfrac12(p_1+p_2)\ge\tfrac15$ for large $j$, whatever the earlier
blocks are.  The product of these conditional bounds over the designated lengths of the triple tends
to $0$, so almost surely machine $i$ fails at infinitely many of them; the triples are countable, so
almost surely this holds for every triple at once.  For the standard model the single-block advantage
bound is the one in the proof of Theorem~\ref{thm:standard-query-lower}, which is likewise uniform
over restrictions.

\emph{The separating languages.}
For each model $\mu$ and level $k$ let $L_{\mu,k}$ consist of the strings $1^n$ such that $n$ is
designated for a triple $(\mu,k,i)$ and the block at $n$ is a $\Yes$ instance.  The set of such $n$
is polynomial-time decidable, every block lies in its promise set by construction, and the
depth-$(k+1)$ decider of Theorem~\ref{thm:upper-worst-case} (or Theorem~\ref{thm:standard-query-lower})
queries only its own slice, so $L_{\mu,k}$ is decided at Fourier depth $k+1$ in the model $\mu$ at
every length, whereas every level-$k$ machine of that model fails on it at infinitely many lengths.
Any realization with these properties is an oracle $O_1$ that separates every adjacent pair of
levels in both models.  The same oracle realizes Theorem~\ref{thm:phase-vs-standard} as well: the
language $L_{\mathrm{phase},k}$ lies outside $\FH_k^{O_1}$, and its blocks are instances of
$\mathrm{ORF}_{2k-1,n}$, which the standard-query depth-$k$ circuit in the proof of
Theorem~\ref{thm:phase-vs-standard} decides while querying only its own slice; hence
$L_{\mathrm{phase},k}\in\FH_{k,\mathrm{std}}^{O_1}\setminus\FH_k^{O_1}$.
\end{proof}

\begin{corollary}[The adjacent-level question does not relativize]\label{cor:nonrelativizing}
Under the standard convention there are oracles $O_1$ and $O_2$ such that, for every $k\ge2$,
\[
\FH_{k,\mathrm{std}}^{O_1}\subsetneq\FH_{k+1,\mathrm{std}}^{O_1},
\qquad
\FH_{k,\mathrm{std}}^{O_2}=\FH_{k+1,\mathrm{std}}^{O_2},
\]
and likewise in the phase model.  Consequently no relativizing argument resolves Shi's conjecture,
in either direction and in either query model.
\end{corollary}

\begin{proof}
Take $O_1$ from Proposition~\ref{rem:standard-convention-separations}, which handles every $k\ge2$ and both
models with one oracle, and $O_2$ any $\mathrm{PSPACE}$-complete language as in
Proposition~\ref{prop:pspace-collapse}.
\end{proof}

Relative to $O_2$ the chain of inclusions in Remark~\ref{rem:tower} also collapses:
$\FH_2^{O_2}=\bigl(\bigcup_k\FH_k\bigr)^{O_2}=\BQP^{O_2}$, in contrast with
Theorem~\ref{thm:fh-vs-bqp}, and the phase hierarchy contains $\P$ relative to $O_2$, in contrast
with Corollary~\ref{cor:p-not-in-fh}.  Which of the two sides of
Corollary~\ref{cor:nonrelativizing} holds in the unrelativized setting is precisely Shi's
conjecture, and deciding this requires techniques that do not relativize.

\section{Discussion and open problems}\label{sec:open-problems}

We conclude with four questions that the results leave open.

First, the quantitative form of the separation.  The query threshold of
Theorem~\ref{thm:lower-worst-case} is $N^{c_k}$ with $c_k=\Theta(1/k^2)$, and
Remark~\ref{rem:parametric-lower-bound} records how far the present argument reaches.  The optimal
exponent is open, in both query models.

Second, how much is gained by retaining the answer bit.  Theorem~\ref{thm:phase-vs-standard}
separates the two access models at a fixed Fourier depth, and
Proposition~\ref{prop:classical-prefix-separation} shows that the separation survives an adaptive
classical preprocessing stage.  Whether one further Hadamard layer suffices,
\[
\FH_{k,\mathrm{std}}^{O}\ \subseteq\ \P^{O}\text{-}\FH_{k+1}^{O}?
\]
remains open.  The question is related to Aaronson's question about standard versus erasing
oracles~\cite[Problem~11]{aaronson2021open}, although the two are not the same: an erasing oracle
maps an input to its image and discards the input register, and is not a phase oracle.

Third, the sign ensemble of Appendix~\ref{sec:appendix-sign}.  Its analysis reduces to the decay of
a single Fourier coefficient, which Theorem~\ref{thm:revealed-decay} establishes in one of the two
parameter ranges and Conjecture~\ref{conj:sharp-decay} asserts in general.  Proving the conjecture
would make the ensemble an explicitly defined hard distribution for the oracle separation
$\FH_2\subsetneq\FH_3$.

Fourth, an intrinsic characterization of Fourier depth.  Arunachalam, Bri\"et, and
Palazuelos~\cite{arunachalam2017quantumqueryalgorithmscompletely} characterize $q$-query quantum
algorithms as exactly the completely bounded forms of degree $2q$, a characterization that counts
queries but takes no account of Fourier depth.  What additional structure on the form corresponds to a bound
of $k$ on the number of Hadamard layers?  A related question is whether $\FH_k$ and the
$(k-1)$-round parallel-query model of Lemma~\ref{lem:round-embedding} differ as classes.

Two further points are settled by oracles.  The containment of $\mathrm{NP}$
in the standard-query levels is oracle-dependent: relative to a $\mathrm{PSPACE}$-complete
oracle $\mathrm{NP}^{O}=\FH_{k,\mathrm{std}}^{O}$ for every $k$
(Proposition~\ref{prop:pspace-collapse}), while relative to a random oracle
$\mathrm{NP}^{O}\not\subseteq\BQP^{O}$~\cite{bennett1997strengths}.  Shi's conjecture itself is
likewise beyond relativizing techniques: both directions are realized by oracles, the separations of
this paper on one side and the collapse just mentioned on the other, so no relativizing argument
resolves it (Corollary~\ref{cor:nonrelativizing}).

\section*{Acknowledgments}
We thank Uma Girish and Makrand Sinha for helpful discussion regarding~\cite{girish2024power} (personal communication, August 2026). We also thank the anonymous reviewers of FOCS 2026 for constructive feedback on an earlier version of this work, which claimed only a partial resolution of Shi's conjecture and whose contents now form Appendix~\ref{sec:appendix-sign}. Large language models were used in preparing this paper: Claude Opus and ChatGPT 5 throughout, and Claude Opus and Fable in the later stages of writing and to check proofs and computations. Their role was substantial in identifying errors in the author's proof sketches and in expanding those sketches into full rigorous proofs. We verified the originality of the results, checked every reference against its source, and take full responsibility for all content. This work is supported by the author's faculty startup grant from Virginia Tech.

\printbibliography

@article{girish2025fourierspectrumnoisyquantum,
  title={{F}ourier Spectrum of Noisy Quantum Algorithms},
  author={Girish, Uma},
  journal={arXiv preprint arXiv:2510.06385},
  year={2025}
}

@article{GGM86,
  author = {Goldreich, Oded and Goldwasser, Shafi and Micali, Silvio},
  title = {How to construct random functions},
  journal = {Journal of the {ACM}},
  volume = {33},
  number = {4},
  pages = {792--807},
  year = {1986},
  doi = {10.1145/6490.6503}
}

@article{HILL99,
  author = {H{\aa}stad, Johan and Impagliazzo, Russell and Levin, Leonid A. and Luby, Michael},
  title = {A Pseudorandom Generator from any One-way Function},
  journal = {{SIAM} Journal on Computing},
  volume = {28},
  number = {4},
  pages = {1364--1396},
  year = {1999},
  doi = {10.1137/S0097539793244708}
}

@inproceedings{zhandry2012,
  author = {Zhandry, Mark},
  title = {How to Construct Quantum Random Functions},
  booktitle = {53rd Annual {IEEE} Symposium on Foundations of Computer Science ({FOCS})},
  pages = {679--687},
  year = {2012},
  doi = {10.1109/FOCS.2012.37}
}

@inproceedings{ABKM16,
  author = {Aaronson, Scott and Bouland, Adam and Kuperberg, Greg and Mehraban, Saeed},
  title = {The computational complexity of ball permutations},
  booktitle = {Proceedings of the 49th Annual {ACM} {SIGACT} Symposium on Theory of Computing ({STOC})},
  pages = {317--327},
  year = {2017},
  doi = {10.1145/3055399.3055453}
}

@misc{JM24,
  author = {Dale Jacobs and Saeed Mehraban},
  title = {The Space Just Above One Clean Qubit},
  eprint = {2410.08051},
  archiveprefix = {arXiv},
  year = {2024},
  note = {Manuscript}
}

@article{aaronson2021open,
  title = {Open Problems Related to Quantum Query Complexity},
  author = {Aaronson, Scott},
  journal = {arXiv preprint arXiv:2109.06917},
  year = {2021}
}

@article{Shi_2005,
   title={Quantum and classical tradeoffs},
   volume={344},
   ISSN={0304-3975},
   url={http://dx.doi.org/10.1016/j.tcs.2005.03.053},
   DOI={10.1016/j.tcs.2005.03.053},
   eprint={quant-ph/0312213},
   archiveprefix={arXiv},
   note={The class $\FH_k$ and Conjecture~4.1 appear in the concluding discussion; section and
   statement numbers cited here follow the preprint, \texttt{arXiv:quant-ph/0312213}},
   number={2--3},
   journal={Theoretical Computer Science},
   author={Shi, Yaoyun},
   year={2005},
   month=nov, pages={335--345} }

@inproceedings{aaronson2014forrelationproblemoptimallyseparates,
  title={{F}orrelation: A problem that optimally separates quantum from classical computing},
  author={Aaronson, Scott and Ambainis, Andris},
  booktitle={Proceedings of the Forty-Seventh Annual ACM Symposium on Theory of Computing (STOC 2015)},
  pages={307--316},
  year={2015},
  note={Theorem numbers cited here follow the full version, \texttt{arXiv:1411.5729}}
}

@inproceedings{bansal2021k,
  eprint = {2008.07003},
  archiveprefix = {arXiv},
  note = {Theorem numbers cited here follow the full version, \texttt{arXiv:2008.07003}},
  title={k-{F}orrelation optimally separates quantum and classical query complexity},
  author={Bansal, Nikhil and Sinha, Makrand},
  booktitle={Proceedings of the 53rd Annual ACM SIGACT Symposium on Theory of Computing (STOC 2021)},
  pages={1303--1316},
  year={2021},
  doi={10.1145/3406325.3451040}
}

@inproceedings{girish2025forrelationextremallyhard,
  title={Forrelation is extremally hard},
  author={Girish, Uma and Servedio, Rocco},
  booktitle={17th Innovations in Theoretical Computer Science Conference (ITCS 2026)},
  series={Leibniz International Proceedings in Informatics (LIPIcs)},
  volume={362},
  pages={72:1--72:22},
  year={2026},
  publisher={Schloss Dagstuhl -- Leibniz-Zentrum f{\"u}r Informatik},
  doi={10.4230/LIPIcs.ITCS.2026.72}
}

@inproceedings{girish2024power,
  title={The power of adaptivity in quantum query algorithms},
  author={Girish, Uma and Sinha, Makrand and Tal, Avishay and Wu, Kewen},
  booktitle={Proceedings of the 56th Annual ACM Symposium on Theory of Computing (STOC 2024)},
  pages={1488--1497},
  year={2024},
  doi={10.1145/3618260.3649621}
}

@book{stanley2012enumerative,
  title={Enumerative Combinatorics},
  author={Stanley, Richard P.},
  volume={1},
  edition={2},
  publisher={Cambridge University Press},
  year={2012},
  doi={10.1017/CBO9781139058520}
}

@inproceedings{girish2021parity,
  title={Fourier growth of parity decision trees},
  author={Girish, Uma and Tal, Avishay and Wu, Kewen},
  booktitle={36th Computational Complexity Conference (CCC 2021)},
  year={2021},
  eprint={2103.11604},
  archivePrefix={arXiv},
  primaryClass={cs.CC}
}

@article{bremner2016average,
  doi = {10.1103/PhysRevLett.117.080501},
  title={Average-case complexity versus approximate simulation of commuting quantum computations},
  author={Bremner, Michael J and Montanaro, Ashley and Shepherd, Dan J},
  journal={Physical Review Letters},
  volume={117},
  number={8},
  pages={080501},
  year={2016},
}

@article{Morimae_2018,
   title={{M}erlin-{A}rthur with efficient quantum Merlin and quantum supremacy for the second level of the {F}ourier hierarchy},
   volume={2},
   ISSN={2521-327X},
   url={http://dx.doi.org/10.22331/q-2018-11-15-106},
   DOI={10.22331/q-2018-11-15-106},
   journal={Quantum},
   author={Morimae, Tomoyuki and Takeuchi, Yuki and Nishimura, Harumichi},
   year={2018},
   month=nov, pages={106} }

@misc{fefferman2015powerquantumfouriersampling,
      title={The Power of Quantum Fourier Sampling},
      author={Bill Fefferman and Chris Umans},
      year={2015},
      eprint={1507.05592},
      archivePrefix={arXiv},
      primaryClass={cs.CC},
      url={https://arxiv.org/abs/1507.05592},
}

@INPROCEEDINGS{365700,
  author={Shor, P.W.},
  booktitle={Proceedings 35th Annual Symposium on Foundations of Computer Science},
  title={Algorithms for quantum computation: discrete logarithms and factoring},
  year={1994},
  pages={124--134},
  doi={10.1109/SFCS.1994.365700}}

@INPROCEEDINGS{365701,
  author={Simon, D.R.},
  booktitle={Proceedings 35th Annual Symposium on Foundations of Computer Science},
  title={On the power of quantum computation},
  year={1994},
  pages={116--123},
  doi={10.1109/SFCS.1994.365701}}

@book{odonnell2021analysisbooleanfunctions,
  title={Analysis of {B}oolean Functions},
  author={O'Donnell, Ryan},
  year={2014},
  publisher={Cambridge University Press}
}

@inproceedings{mossel2005noisestabilityfunctionslow,
  doi = {10.4007/annals.2010.171.295},
  title={Noise stability of functions with low influences: invariance and optimality},
  author={Mossel, Elchanan and O'Donnell, Ryan and Oleszkiewicz, Krzysztof},
  booktitle={46th Annual IEEE Symposium on Foundations of Computer Science (FOCS'05)},
  pages={21--30},
  year={2005},
  organization={IEEE},
  note={Journal version: \emph{Annals of Mathematics} \textbf{171} (2010), 295--341}
}

@misc{ananth2024pseudorandomnessinverselesshaarrandom,
      title={Pseudorandomness in the (Inverseless) Haar Random Oracle Model},
      author={Prabhanjan Ananth and John Bostanci and Aditya Gulati and Yao-Ting Lin},
      year={2024},
      eprint={2410.19320},
      archivePrefix={arXiv},
      primaryClass={quant-ph},
      url={https://arxiv.org/abs/2410.19320},
}

@article{bu2025quantumhigherorderfourier,
  doi = {10.1073/pnas.2515667122},
  title={Quantum higher-order Fourier analysis and the Clifford hierarchy},
  author={Bu, Kaifeng and Gu, Weichen and Jaffe, Arthur},
  journal={Proceedings of the National Academy of Sciences},
  volume={122},
  number={45},
  pages={e2515667122},
  year={2025},
}

@article{shepherd2009temporally,
  doi = {10.1098/rspa.2008.0443},
  title={Temporally unstructured quantum computation},
  author={Shepherd, Dan and Bremner, Michael J},
  journal={Proceedings of the Royal Society A: Mathematical, Physical and Engineering Sciences},
  volume={465},
  number={2105},
  pages={1413--1439},
  year={2009},
}

@article{bremner2011classical,
  doi = {10.1098/rspa.2010.0301},
  title={Classical simulation of commuting quantum computations implies collapse of the polynomial hierarchy},
  author={Bremner, Michael J and Jozsa, Richard and Shepherd, Dan J},
  journal={Proceedings of the Royal Society A: Mathematical, Physical and Engineering Sciences},
  volume={467},
  number={2126},
  pages={459--472},
  year={2011},
}

@article{shepherd2010quantum,
  title={Quantum Complexity: restrictions on algorithms and architectures},
  author={Shepherd, Daniel James},
  journal={arXiv preprint arXiv:1005.1425},
  year={2010}
}

@article{arunachalam2017quantumqueryalgorithmscompletely,
  title={Quantum Query Algorithms are Completely Bounded Forms},
  author={Arunachalam, Srinivasan and Bri{\"e}t, Jop and Palazuelos, Carlos},
  journal={arXiv preprint arXiv:1711.07285},
  year={2017},
  eprint={1711.07285},
  archivePrefix={arXiv},
  primaryClass={quant-ph},
  url={https://arxiv.org/abs/1711.07285}
}

@article{nisan1994hardness,
  doi = {10.1016/S0022-0000(05)80043-1},
  title={Hardness vs randomness},
  author={Nisan, Noam and Wigderson, Avi},
  journal={Journal of Computer and System Sciences},
  volume={49},
  number={2},
  pages={149--167},
  year={1994},
}

@article{Demarie_2018,
   title={Classical verification of quantum circuits containing few basis changes},
   volume={97},
   ISSN={2469-9934},
   url={http://dx.doi.org/10.1103/PhysRevA.97.042319},
   DOI={10.1103/physreva.97.042319},
   number={4},
   journal={Physical Review A},
   author={Demarie, Tommaso F. and Ouyang, Yingkai and Fitzsimons, Joseph F.},
   year={2018},
   month=apr }

@article{shor1999polynomial,
  doi = {10.1137/S0036144598347011},
  title={Polynomial-time algorithms for prime factorization and discrete logarithms on a quantum computer},
  author={Shor, Peter W},
  journal={SIAM Review},
  volume={41},
  number={2},
  pages={303--332},
  year={1999},
}

@book{nielsen2010quantum,
  title={Quantum Computation and Quantum Information},
  author={Nielsen, Michael A and Chuang, Isaac L},
  year={2010},
  publisher={Cambridge University Press}
}

@article{kitaev1995quantum,
  title={Quantum measurements and the {A}belian stabilizer problem},
  author={Kitaev, A Yu},
  journal={arXiv preprint quant-ph/9511026},
  year={1995}
}

@misc{shi2002toffolicontrollednotneedlittle,
      title={Both {T}offoli and {C}ontrolled-{NOT} need little help to do universal quantum computation},
      author={Yaoyun Shi},
      year={2002},
      eprint={quant-ph/0205115},
      archivePrefix={arXiv},
      primaryClass={quant-ph},
      url={https://arxiv.org/abs/quant-ph/0205115},
}

@article{vadhan2012pseudorandomness,
  doi = {10.1561/0400000010},
  title={Pseudorandomness},
  author={Vadhan, Salil P},
  journal={Foundations and Trends in Theoretical Computer Science},
  volume={7},
  number={1--3},
  pages={1--336},
  year={2012},
}

@inproceedings{ji2018pseudorandom,
  doi = {10.1007/978-3-319-96878-0_5},
  title={Pseudorandom quantum states},
  author={Ji, Zhengfeng and Liu, Yi-Kai and Song, Fang},
  booktitle={Annual International Cryptology Conference},
  pages={126--152},
  year={2018},
  organization={Springer}
}

@inproceedings{krestchmer_pru,
  doi = {10.4230/LIPIcs.TQC.2021.2},
  title={Quantum Pseudorandomness and Classical Complexity},
  author={Kretschmer, William},
  booktitle={16th Conference on the Theory of Quantum Computation, Communication and Cryptography (TQC 2021)},
  pages={2:1--2:20},
  year={2021},
  organization={Schloss Dagstuhl--Leibniz-Zentrum f{\"u}r Informatik}
}

@inproceedings{ma2025constructrandomunitaries,
  doi = {10.1145/3717823.3718254},
  title={How to construct random unitaries},
  author={Ma, Fermi and Huang, Hsin-Yuan},
  booktitle={Proceedings of the 57th Annual ACM Symposium on Theory of Computing},
  pages={806--809},
  year={2025}
}

@inproceedings{metger2024simpleconstructionslineardepthtdesigns,
  doi = {10.1109/FOCS61266.2024.00038},
  title={Simple constructions of linear-depth t-designs and pseudorandom unitaries},
  author={Metger, Tony and Poremba, Alexander and Sinha, Makrand and Yuen, Henry},
  booktitle={2024 IEEE 65th Annual Symposium on Foundations of Computer Science (FOCS)},
  pages={485--492},
  year={2024},
  organization={IEEE}
}

@article{bernstein1997quantum,
  title={Quantum complexity theory},
  author={Bernstein, Ethan and Vazirani, Umesh},
  journal={SIAM Journal on Computing},
  volume={26},
  number={5},
  pages={1411--1473},
  year={1997},
  doi={10.1137/S0097539796300921}
}

@inproceedings{cleve2000fast,
  doi = {10.1109/SFCS.2000.892140},
  title={Fast parallel circuits for the quantum {F}ourier transform},
  author={Cleve, Richard and Watrous, John},
  booktitle={Proceedings 41st Annual Symposium on Foundations of Computer Science (FOCS 2000)},
  pages={526--536},
  year={2000},
  organization={IEEE}
}

@article{raz2019oracle,
  title={Oracle separation of {BQP} and {PH}},
  author={Raz, Ran and Tal, Avishay},
  journal={Journal of the ACM},
  volume={69},
  number={4},
  pages={30:1--30:21},
  year={2022},
  doi={10.1145/3530258},
  note={Preliminary version in STOC 2019}
}

@inproceedings{tal2020towards,
  title={Towards optimal separations between quantum and randomized query complexities},
  author={Tal, Avishay},
  booktitle={Proceedings of the 61st IEEE Annual Symposium on Foundations of Computer Science (FOCS 2020)},
  pages={228--239},
  year={2020},
  doi={10.1109/FOCS46700.2020.00030}
}

@article{sherstov2023optimal,
  title={An optimal separation of randomized and quantum query complexity},
  author={Sherstov, Alexander A. and Storozhenko, Andrey A. and Wu, Pei},
  journal={SIAM Journal on Computing},
  volume={52},
  number={2},
  pages={525--567},
  year={2023},
  doi={10.1137/22M1468943},
  note={Preliminary version in STOC 2021}
}

@article{bravyi2021classical,
  title={Classical algorithms for {F}orrelation},
  author={Bravyi, Sergey and Gosset, David and Grier, Daniel and Schaeffer, Luke},
  journal={arXiv preprint arXiv:2102.06963},
  year={2021},
  url={https://arxiv.org/abs/2102.06963}
}

@inproceedings{carolan2025parallel,
  title={Quantum advantage and lower bounds in parallel query complexity},
  author={Carolan, Joseph and Gilani, Amin Shiraz and Vempati, Mahathi},
  booktitle={16th Innovations in Theoretical Computer Science Conference (ITCS 2025)},
  series={Leibniz International Proceedings in Informatics (LIPIcs)},
  volume={325},
  pages={31:1--31:14},
  year={2025},
  publisher={Schloss Dagstuhl -- Leibniz-Zentrum f{\"u}r Informatik},
  doi={10.4230/LIPIcs.ITCS.2025.31}
}

@inproceedings{coudron2020computations,
  doi = {10.1145/3357713.3384269},
  title={Computations with greater quantum depth are strictly more powerful (relative to an oracle)},
  author={Coudron, Matthew and Menda, Sanketh},
  booktitle={Proceedings of the 52nd Annual ACM SIGACT Symposium on Theory of Computing (STOC 2020)},
  pages={889--901},
  year={2020}
}

@inproceedings{chia2020need,
  title={On the need for large quantum depth},
  author={Chia, Nai-Hui and Chung, Kai-Min and Lai, Ching-Yi},
  booktitle={Proceedings of the 52nd Annual ACM SIGACT Symposium on Theory of Computing (STOC 2020)},
  pages={902--915},
  year={2020}
}

@article{arora2022oracle,
  title={Oracle separations of hybrid quantum-classical circuits},
  author={Arora, Atul Singh and Gheorghiu, Alexandru and Singh, Uttam},
  journal={arXiv preprint arXiv:2201.01904},
  year={2022},
  url={https://arxiv.org/abs/2201.01904}
}

@inproceedings{hasegawa2022optimal,
  title={An optimal oracle separation of classical and quantum hybrid schemes},
  author={Hasegawa, Atsuya and Le Gall, Fran{\c{c}}ois},
  booktitle={33rd International Symposium on Algorithms and Computation (ISAAC 2022)},
  series={Leibniz International Proceedings in Informatics (LIPIcs)},
  volume={248},
  pages={6:1--6:14},
  year={2022},
  publisher={Schloss Dagstuhl -- Leibniz-Zentrum f{\"u}r Informatik},
  doi={10.4230/LIPIcs.ISAAC.2022.6}
}

@article{hoeffding1963probability,
  doi = {10.1080/01621459.1963.10500830},
  title={Probability inequalities for sums of bounded random variables},
  author={Hoeffding, Wassily},
  journal={Journal of the American Statistical Association},
  volume={58},
  number={301},
  pages={13--30},
  year={1963}
}

@article{buzet2026iqp,
  title={{IQP} circuits for 2-{F}orrelation},
  author={Buzet, Quentin and Chailloux, Andr\'e},
  journal={arXiv preprint arXiv:2604.15248},
  year={2026}
}

@inproceedings{aaronson2010bqp,
  author = {Aaronson, Scott},
  title = {{BQP} and the Polynomial Hierarchy},
  booktitle = {Proceedings of the 42nd Annual {ACM} Symposium on Theory of Computing ({STOC})},
  pages = {141--150},
  year = {2010},
  doi = {10.1145/1806689.1806711}
}

@article{bennett1973logical,
  author = {Bennett, Charles H.},
  title = {Logical Reversibility of Computation},
  journal = {IBM Journal of Research and Development},
  volume = {17},
  number = {6},
  pages = {525--532},
  year = {1973},
  doi = {10.1147/rd.176.0525}
}

@article{hoeffding1948class,
  author = {Hoeffding, Wassily},
  title = {A Class of Statistics with Asymptotically Normal Distribution},
  journal = {The Annals of Mathematical Statistics},
  volume = {19},
  number = {3},
  pages = {293--325},
  year = {1948},
  doi = {10.1214/aoms/1177730196}
}

@book{bhattacharya1976normal,
  author = {Bhattacharya, R. N. and Ranga Rao, R.},
  title = {Normal Approximation and Asymptotic Expansions},
  publisher = {John Wiley \& Sons},
  address = {New York},
  year = {1976}
}

@inproceedings{chung2021compressed,
  author = {Chung, Kai-Min and Fehr, Serge and Huang, Yu-Hsuan and Liao, Tai-Ning},
  title = {On the Compressed-Oracle Technique, and Post-Quantum Security of Proofs of Sequential Work},
  booktitle = {Advances in Cryptology --- {EUROCRYPT} 2021},
  year = {2021},
  note = {arXiv:2010.11658}
}

@inproceedings{klivans2002learning,
  doi = {10.1109/SFCS.2002.1181894},
  title={Learning intersections and thresholds of halfspaces},
  author={Klivans, Adam R. and O'Donnell, Ryan and Servedio, Rocco A.},
  booktitle={43rd Annual IEEE Symposium on Foundations of Computer Science (FOCS)},
  pages={177--186},
  year={2002}
}

@article{bennett1997strengths,
  doi = {10.1137/S0097539796300933},
  author  = {Bennett, Charles H. and Bernstein, Ethan and Brassard, Gilles and Vazirani, Umesh},
  title   = {Strengths and weaknesses of quantum computing},
  journal = {SIAM Journal on Computing},
  volume  = {26},
  number  = {5},
  pages   = {1510--1523},
  year    = {1997}
}

\newpage
\appendix
\section{Toward an explicit hard instance}\label{sec:appendix-sign}

The separation $\FH_2\subsetneq\FH_3$ of Corollary~\ref{cor:fh2-fh3-separation} is unconditional, and
nothing in this appendix is needed for it.  What the hard distribution used there does not supply is
explicitness: it is a rounding of correlated Gaussians, its moments are not available in closed form,
and it is not known to be samplable together with a certificate of the promise.  The reason to remove
this defect is that the adjacent-level question does not relativize
(Corollary~\ref{cor:nonrelativizing}), so no oracle separation settles Shi's conjecture, and an
explicitly defined hard distribution is a first step toward instances that could be evaluated
without an oracle.  This appendix proposes such an instance at the
second level, reduces its analysis to a single Fourier coefficient, proves the conjectured exponent of that
coefficient, with a weaker prefactor, in one of the two parameter ranges, and states the conjecture
that would complete the argument.  Apart from Conjecture~\ref{conj:sharp-decay} and the one statement conditional on it,
all the formal results below are proved here, using only Fourier analysis on the Boolean
cube~\cite{odonnell2021analysisbooleanfunctions}.

\subsection{The sign ensemble}\label{subsec:appendix-upper}

Let $A,B:\bits^n\to\pmone$ be independent uniformly random sign functions.  Their normalized
Hadamard transforms are
\begin{equation}\label{eq:walsh-transform-prelim}
\widetilde A(y)=2^{-n/2}\sum_{x\in\bits^n}(-1)^{x\cdot y}A(x),
\qquad
\widetilde B(y)=2^{-n/2}\sum_{z\in\bits^n}(-1)^{y\cdot z}B(z).
\end{equation}
Each coefficient has variance $1$.  When a coefficient vanishes, use an independent uniform
tie-breaking sign $\tau\in\pmone$ and set
\begin{equation}\label{eq:sgnstar-def}
\sgnstar(t;\tau)=
\begin{cases}1,&t>0,\\-1,&t<0,\\\tau,&t=0.
\end{cases}
\end{equation}
Let $T_A,T_B:\bits^n\to\pmone$ be independent uniform sign functions, independent of $A,B$, and define
\begin{equation}\label{eq:U-V-def-prelim}
U(y)=\sgnstar(\widetilde A(y);T_A(y)),
\qquad
V(y)=\sgnstar(\widetilde B(y);T_B(y)).
\end{equation}

\begin{definition}[The $\Yes$ and $\No$ ensembles]\label{def:yes-no-ensembles}
The structured ensemble $\Yes_n$ is the distribution of
\begin{equation}\label{eq:yes-triple}
F_0=A,\qquad F_1=U\cdot V,\qquad F_2=B,
\end{equation}
where $A,B$ are independent uniform sign functions.  The uniform ensemble $\No_n$ is the
distribution of three independent uniform sign functions $F_0,F_1,F_2$.
\end{definition}

The ensemble combines two known constructions: Aaronson's forrelated distribution, which takes the
signs of a Gaussian vector and of its Hadamard transform~\cite{aaronson2010bqp,%
aaronson2014forrelationproblemoptimallyseparates}, and the form of the hard threefold distribution of
Bansal and Sinha~\cite{bansal2021k}, whose middle oracle is a product.  Here the oracle tuple
$(F_0,F_1,F_2)$ is set equal to $(A,\ U\cdot V,\ B)$, with $U,V$ the tie-broken signs
of \eqref{eq:U-V-def-prelim}, so that the middle oracle is the product $U\cdot V$ and its moments
factor exactly into Fourier coefficients of one explicit function
(Proposition~\ref{prop:structured-factorization-rev}).

The following calculation determines the design of the middle oracle.  Substituting $F_0=A$, $F_1=UV$ and
$F_2=B$ into \eqref{eq:odd-Phi} and summing over $x$ and $z$ first turns each oracle sum into a
Fourier coefficient, leaving
\begin{equation}\label{eq:structured-Phi3}
\Phi_3(A,\,UV,\,B)=\frac1N\sum_{y\in\bits^n}\widetilde A(y)\,U(y)\,V(y)\,\widetilde B(y)
=\frac1N\sum_{y\in\bits^n}\abs{\widetilde A(y)}\,\abs{\widetilde B(y)},
\end{equation}
since $\widetilde A(y)U(y)=\abs{\widetilde A(y)}$ and likewise for $B$, both identities holding at a
vanishing coefficient as well because there each side is $0$.  The middle oracle is exactly the phase
correction that makes every term of the Forrelation sum nonnegative.  On a structured instance,
therefore, $\Phi_3\ge0$ always, it does not depend on the tie-breaking signs, and its mean is
$\E\abs{\widetilde A(y)}\E\abs{\widetilde B(y)}\to2/\pi$.  Under $\No_n$ the three oracles are independent and $\Phi_3$
concentrates at $0$.

\begin{proposition}[The two ensembles and the promise]\label{prop:promise-density-sign}
For every $n\ge1$,
\[
\Prb_{\Yes_n}\bigl[\Phi_3\ge\tfrac14\bigr]\ \ge\ \tfrac19,
\qquad
\Prb_{\No_n}\bigl[\abs{\Phi_3}\le\tfrac1{12}\bigr]\ \ge\ 1-\frac{144}{N}.
\]
\end{proposition}

\begin{proof}
Under $\Yes_n$, \eqref{eq:structured-Phi3} gives $\Phi_3\ge0$, and $\Phi_3\le1$ always by
\eqref{eq:forr-dictionary}.  Since $A$ and $B$ are independent with the same law,
\eqref{eq:structured-Phi3} also gives
$\E[\Phi_3]=\frac1N\sum_y\E\abs{\widetilde A(y)}\,\E\abs{\widetilde B(y)}=\bigl(\E\abs{\widetilde A(y)}\bigr)^2$,
where $\E\abs{\widetilde A(y)}$ does not depend on $y$ because $\widetilde A(y)$ has the law of
$N^{-1/2}S$ for a sum $S$ of $N$ independent uniform signs.  For such a sum $\E[S^2]=N$ and
$\E[S^4]=3N^2-2N$, and two applications of the Cauchy--Schwarz inequality,
$\E[S^2]\le(\E\abs S)^{1/2}(\E\abs S^3)^{1/2}$ and $\E\abs S^3\le(\E S^2)^{1/2}(\E S^4)^{1/2}$,
give $\E\abs S\ge(\E S^2)^{3/2}/(\E S^4)^{1/2}\ge\sqrt{N/3}$.  Hence $\E[\Phi_3]\ge\tfrac13$.
For a random variable $X$ with values in $[0,1]$ and $a\in[0,1)$ we have
$\E X\le a+(1-a)\Prb[X\ge a]$, so
$\Prb[\Phi_3\ge\tfrac14]\ge(\tfrac13-\tfrac14)/(1-\tfrac14)=\tfrac19$.

Under $\No_n$ the three oracles are independent and uniform.  Squaring \eqref{eq:odd-Phi} at $k=2$
and taking expectations, the expectation of each product $F_j(w)F_j(w')$ vanishes unless $w=w'$, so
only the $N^3$ terms with $(x_0,x_1,x_2)$ equal to their primed copies survive, each contributing
$N^{-4}$, and $\E[\Phi_3^2]=1/N$.  Chebyshev's inequality gives
$\Prb[\abs{\Phi_3}>\tfrac1{12}]\le144\,\E[\Phi_3^2]=144/N$.
\end{proof}

\begin{definition}[The sign Forrelation promise problem]\label{def:sign-promise-sets}
For each $n$, define the promise sets
\[
\begin{aligned}
\Pi^{\Yes,\mathrm{sign}}_n
&:=\Bigl\{(F_0,F_1,F_2):\ \Phi_3(F_0,F_1,F_2)\ge\tfrac14\Bigr\},\\[2pt]
\Pi^{\No,\mathrm{sign}}_n
&:=\Bigl\{(F_0,F_1,F_2):\ \abs{\Phi_3(F_0,F_1,F_2)}\le\tfrac1{12}\Bigr\}.
\end{aligned}
\]
The sign Forrelation promise problem is to decide, given phase-query access to a triple promised
to lie in $\Pi^{\Yes,\mathrm{sign}}_n\cup\Pi^{\No,\mathrm{sign}}_n$, which case holds.  By Proposition~\ref{prop:promise-density-sign}, $\Yes_n$ gives mass at least $\tfrac19$ to
$\Pi^{\Yes,\mathrm{sign}}_n$ and $\No_n$ gives mass at least $1-144/N$ to $\Pi^{\No,\mathrm{sign}}_n$.
\end{definition}

\begin{theorem}[Depth three decides the promise problem]\label{thm:FH3-pointwise}
There is a uniform polynomial-size $\FH_3$ circuit family, making $O(n)$ phase queries, that
accepts every instance of $\Pi^{\Yes,\mathrm{sign}}_n$ with probability at least $1-2^{-n}$ and
accepts every instance of $\Pi^{\No,\mathrm{sign}}_n$ with probability at most $2^{-n}$.
\end{theorem}

\begin{proof}
Run $R$ parallel repetitions of the interference test of
Proposition~\ref{prop:odd-interference-circuit} at $k=2$; by Remark~\ref{rem:register-routing} all
repetitions share the same three global Hadamard layers, and all other gates are basis-preserving,
so the combined circuit is a uniform $\FH_3$ circuit with $2R$ phase queries.  Measure all control
qubits and accept if at least $\frac R2\bigl(1+\frac16\bigr)$ of them equal $0$; this accepting set
is recognizable in polynomial time.

Each repetition accepts independently with probability $\frac12(1+\Phi_3)$.  The threshold constant
$\frac16$ is the midpoint of $\frac14$ and $\frac1{12}$, so on $\Pi^{\Yes,\mathrm{sign}}_n$, where
$\Phi_3\ge\frac14$, the acceptance probability of each repetition exceeds the threshold fraction by
at least $\frac12(\frac14-\frac16)=\frac1{24}$, and on $\Pi^{\No,\mathrm{sign}}_n$, where
$\Phi_3\le\frac1{12}$, it falls below the threshold by the same amount.  By Hoeffding's inequality
(Lemma~\ref{lem:hoeffding}), the threshold test fails with probability at most $\exp(-2R/24^2)$,
which is at most $2^{-n}$ for $R=\lceil 2^{8}\,n\rceil=O(n)$.
\end{proof}

\begin{remark}[From the ensembles to an oracle]\label{rem:sign-or}
The constant yes-mass of Proposition~\ref{prop:promise-density-sign} is what the direct product of
Definition~\ref{def:hard-problem} amplifies.  For the OR of $m$ independent instances,
$\Yes_n^{\otimes m}$ has some instance with $\Phi_3\ge\tfrac14$ except with probability
$(\tfrac89)^m$, and $\No_n^{\otimes m}$ has every instance with $\abs{\Phi_3}\le\tfrac1{12}$ except
with probability $144m/N$.  Running the circuit of Theorem~\ref{thm:FH3-pointwise} on every instance
in parallel and accepting if some instance accepts decides this OR at Fourier depth three, as in
Theorem~\ref{thm:upper-worst-case}, and the hybrid argument of Theorem~\ref{thm:lower-worst-case}
reduces the advantage of an $\FH_2$ circuit on the OR to its advantage on a single instance with the
other instances fixed.  Proposition~\ref{prop:conditional-lower-bound} below bounds the latter, and
its proof applies to such restrictions because the growth bound it uses does.  Under
Conjecture~\ref{conj:sharp-decay}, the diagonalization of Theorem~\ref{thm:all-level-separation}
therefore gives an oracle $O$ with $\FH_2^{O}\subsetneq\FH_3^{O}$ whose instances are drawn from
these two ensembles.
\end{remark}

\subsection{One Fourier coefficient of a product of threshold functions}\label{subsec:appendix-moments}

Let $A:\bits^n\to\pmone$ be uniform and let $H^{\otimes n}A$ be its Hadamard transform, with entries
$\widetilde A(y)$ as in \eqref{eq:walsh-transform-prelim}.  For a finite
$Y\subseteq\bits^n$ define
\begin{equation}\label{eq:gY-def}
g_Y(A)=\prod_{y\in Y}\sgn_0\bigl(\widetilde A(y)\bigr),
\end{equation}
where $\sgn_0$ is the sign function with $\sgn_0(0)=0$.  Averaging
\eqref{eq:U-V-def-prelim} over the tie-breaking signs, which are independent and mean zero, replaces
each $U(y)$ by $\sgn_0(\widetilde A(y))$, so for finite $X\subseteq\bits^n$
\begin{equation}\label{eq:gamma-def-appendix}
\Gamma(X,Y)
:=\E\Bigl[\prod_{x\in X}A(x)\prod_{y\in Y}U(y)\Bigr]
=\E\bigl[A_X\,g_Y(A)\bigr]
=\widehat{g_Y}(X),
\qquad A_X:=\prod_{x\in X}A(x).
\end{equation}
The rest of the appendix concerns this coefficient.  Here $g_Y$ is a product of $\abs Y$ threshold
functions, each equal to $0$ on a tie, so that $g_Y$ takes values in $\{-1,0,1\}$:
\[
\sgn_0\bigl(\widetilde A(y)\bigr)=\sgn_0\Bigl(N^{-1/2}\textstyle\sum_x(-1)^{x\cdot y}A(x)\Bigr)
\]
whose defining linear forms are the rows of $H^{\otimes n}$ indexed by $Y$.  Those rows are
orthonormal, and every coefficient of every one of them has the same modulus $N^{-1/2}$.  Thus $g_Y$
is a product of $\abs Y$ threshold functions of pairwise orthogonal linear forms, each of which
satisfies $\max_i\abs{w_i}/\Norm w_2=N^{-1/2}$, the smallest value that ratio can take; such forms
are regular in the sense of Mossel, O'Donnell and
Oleszkiewicz~\cite{mossel2005noisestabilityfunctionslow}.  The remainder of this subsection concerns the
low-degree Fourier coefficients of such a product.

For finite $X,Y,Z \subseteq \bits^n$, define the reduced monomial
\begin{equation}\label{eq:reduced-monomial-rev}
M_{X,Y,Z}
=
\prod_{x \in X} F_0(x)
\prod_{y \in Y} F_1(y)
\prod_{z \in Z} F_2(z);
\end{equation}
repeated occurrences cancel because every oracle value lies in $\pmone$.  The two halves of the
ensemble are independent, so its moments are products of two such coefficients.

\begin{proposition}[Exact factorization of structured moments]\label{prop:structured-factorization-rev}
For every finite $X,Y,Z \subseteq \bits^n$,
\begin{equation}\label{eq:structured-factorization}
\E_{\Yes_n}[M_{X,Y,Z}]
=
\widehat{g_Y}(X)\,\widehat{g_Y}(Z).
\end{equation}
\end{proposition}

\begin{proof}
Under the structured ensemble $F_0=A$, $F_1=U\cdot V$ and $F_2=B$ with $A$ independent of $B$, so
$M_{X,Y,Z}=\bigl(A_X\prod_{y}U(y)\bigr)\bigl(B_Z\prod_{y}V(y)\bigr)$ splits into a function of $A$
and a function of $B$.  The expectation factors, each factor is $\widehat{g_Y}$ evaluated at the
corresponding set by \eqref{eq:gamma-def-appendix}, and $B$ has the same law as $A$.
\end{proof}

Under $\No_n$ the three oracles are independent and uniform, so any nonconstant reduced monomial has
at least one oracle value appearing an odd number of times and therefore has expectation $0$.

The whole analysis of the ensemble is therefore an analysis of the low-degree Fourier coefficients
of $g_Y$.  We write $\Gamma(X,Y)$ for $\widehat{g_Y}(X)$ when both arguments vary, and turn to what is
known about it.

Three symmetries of the ensemble force many coefficients to vanish.  All three come from the same source: the law of $A$ is invariant under a group of
transformations of the cube, and each of them acts on $g_Y$ in a way we can follow exactly.  Write
$\Sigma_X=\bigoplus_{x\in X}x$ and $\Sigma_Y=\bigoplus_{y\in Y}y$.

\begin{lemma}[Symmetries of $\Gamma$]\label{lem:gamma-symmetries}
For all finite $X,Y\subseteq\bits^n$ and all $x_0,s\in\Ftwo^n$:
\begin{enumerate}[leftmargin=2em]
\item\label{it:parity} $\Gamma(X,Y)=0$ unless $\abs X+\abs Y$ is even;
\item\label{it:covariance}
$\Gamma(X+x_0,Y)=(-1)^{x_0\cdot\Sigma_Y}\,\Gamma(X,Y)$ and
$\Gamma(X,Y+s)=(-1)^{s\cdot\Sigma_X}\,\Gamma(X,Y)$;
\item\label{it:selection} if $X$ is invariant under translation by some $x_0$ with
$x_0\cdot\Sigma_Y=1$, then $\Gamma(X,Y)=0$, and dually in $Y$.
\end{enumerate}
\end{lemma}

\begin{proof}
Each part exhibits a measure-preserving transformation under which $\Gamma$ acquires an explicit
sign, and the sign is then forced to be $+1$.

For~\ref{it:parity}, flip $A$ and the tie-breaking signs together, $(A,T_A)\mapsto(-A,-T_A)$.  Both
are uniform and independent, so the joint law is unchanged.  Every value $A(x)$ changes sign, every
coefficient $\widetilde A(y)$ changes sign, and $\sgnstar(-t;-\tau)=-\sgnstar(t;\tau)$ holds for all
$t$ including $t=0$, so every $U(y)$ changes sign as well.  The product defining $\Gamma(X,Y)$ is
therefore multiplied by $(-1)^{\abs X+\abs Y}$, which forces it to vanish when $\abs X+\abs Y$ is
odd.

For~\ref{it:covariance}, translate the cube: $\Theta:A\mapsto A(\cdot+x_0)$ permutes coordinates and
so preserves the uniform law, while a change of summation variable gives
$\widetilde{\Theta A}(y)=(-1)^{x_0\cdot y}\widetilde A(y)$.  Pairing $\Theta$ with the tie-bit
change $T_A(y)\mapsto(-1)^{x_0\cdot y}T_A(y)$ preserves the joint law and sends
$U(y)\mapsto(-1)^{x_0\cdot y}U(y)$, so the factors $U(y)$ pick up the phase $(-1)^{x_0\cdot\Sigma_Y}$
while the factors $A(x)$ are relabelled.  Dually, the modulation $A(x)\mapsto(-1)^{x\cdot s}A(x)$
preserves the law and sends $\widetilde A(y)\mapsto\widetilde A(y+s)$; pairing it with the
relabelling $T_A(y)\mapsto T_A(y+s)$ sends $U(y)\mapsto U(y+s)$, while the factors $A(x)$, $x\in X$, pick up
$(-1)^{s\cdot\Sigma_X}$.

Part~\ref{it:selection} is~\ref{it:covariance} read at a period of $X$: the left-hand side is
$\Gamma(X,Y)$ and the right-hand side is $-\Gamma(X,Y)$.
\end{proof}

We use Part~\ref{it:selection} mainly at $X=\varnothing$.  Every $x_0$ is a period of the empty set,
so $\Gamma(\varnothing,Y)=0$ whenever $\Sigma_Y\ne0$; in particular
\begin{equation}\label{eq:pairs-uncorrelated}
\Gamma(\varnothing,\{y_1,y_2\})=0\qquad\text{for all distinct }y_1,y_2,
\end{equation}
so the signs of two distinct Fourier coefficients are exactly uncorrelated, not merely nearly so.

Beyond the symmetries, the following bound is what we can prove about $\Gamma$ without further
assumptions.  It needs no expansion of the signs: they depend on $A$ only through a few aggregate
sums, and conditioned on those sums each value of $A$ is nearly unbiased.

\begin{theorem}[Unconditional decay in $\abs X$]\label{thm:revealed-decay}
Let $X\subseteq\bits^n$ and $Y\subseteq\bits^n\setminus\{0\}$ be finite and nonempty, let $V$ be the
linear span of $Y$ in $\Ftwo^n$, and let $\rho=\dim V$.  The $2^{\rho}$ cosets of $V^{\perp}$
partition $\bits^n$; let $j_1,\dots,j_{2^{\rho}}$ be the numbers of points of $X$ in each coset, so
that $\sum_hj_h=\abs X$.  Then
\[
\abs{\Gamma(X,Y)}
\ \le\
\prod_{h=1}^{2^{\rho}}\binom{N/2^{\rho}}{j_h}^{-1/2},
\]
and if moreover $\abs X\le N/2^{\rho}$, then the right side is at most
\[
\binom{N/2^{\rho}}{\abs X}^{-1/2}
\ \le\
\left(\frac{2^{\rho}\abs X}{N}\right)^{\abs X/2}.
\]
\end{theorem}

\begin{proof}
The proof has three steps.  First, the signs are functions of the coset sums alone.  Second,
conditioned on the coset sums, the product of the values $A(x)$ over $x\in X$ has second moment equal to the
reciprocal of a product of binomial coefficients.  Third, Cauchy--Schwarz within each coset, with
independence across cosets, gives the bound.

\emph{Step 1: the signs depend only on the coset sums.}
Each character $x\mapsto(-1)^{x\cdot y}$ with $y\in Y$ is constant on every coset $F$ of
$V^{\perp}$, since $x\cdot y$ is unchanged when $x$ moves by an element of $V^{\perp}$.  Writing
$\chi_F(y)$ for its value on $F$ and $W_F=\sum_{x\in F}A(x)$ for the coset sums,
\[
\sqrt N\,\widetilde A(y)=\sum_{F}\chi_F(y)\,W_F
\qquad\text{for every }y\in Y.
\]
Each $U(y)=\sgnstar(\widetilde A(y);T_A(y))$, and hence the product $\prod_{y\in Y}U(y)$, is
therefore a function of the coset sums $W=(W_F)_F$ and the tie-breaking signs $T_A$, and $T_A$ is
independent of $A$.

\emph{Step 2: the conditional expectation of the product over $X$.}
Conditioned on $W$, the restrictions $A|_F$ are independent across cosets, and within a coset $F$ of
size $\nu:=N/2^{\rho}$ the restriction is uniform on the sign patterns with sum $W_F$.  Writing
$X_F=X\cap F$ and $j_F=\abs{X_F}$,
\[
\E\Bigl[\prod_{x\in X}A(x)\Bigm|W\Bigr]=\prod_F\kappa_F(W_F),
\qquad
\kappa_F(W_F):=\E\Bigl[\prod_{x\in X_F}A(x)\Bigm|W_F\Bigr].
\]
The conditional law of $A|_F$ is exchangeable, so $\E[\prod_{x\in S}A(x)\mid W_F]$ is the same for
every $S\subseteq F$ with $\abs S=j_F$; summing over all such $S$,
\[
\binom{\nu}{j_F}\,\kappa_F=\E[e_{j_F}\mid W_F]=e_{j_F},
\qquad
e_j:=\sum_{S\subseteq F,\ \abs S=j}\ \prod_{x\in S}A(x),
\]
the last equality because $e_{j}$ is a symmetric function of the values $A|_F$ and hence a function
of $W_F$.  Squaring and taking expectations, the monomials $\prod_{x\in S}A(x)$ are orthonormal, so
$\E[e_{j}^{2}]=\binom{\nu}{j}$ and
\[
\E\bigl[\kappa_F^{2}\bigr]=\binom{\nu}{j_F}^{-1}.
\]

\emph{Step 3: Cauchy--Schwarz.}
Conditioning on $(W,T_A)$ and using Step~1 and Step~2,
\[
\abs{\Gamma(X,Y)}
=\Bigl|\E\Bigl[\prod_F\kappa_F(W_F)\cdot\prod_{y\in Y}U(y)\Bigr]\Bigr|
\le\E\Bigl[\prod_F\abs{\kappa_F}\Bigr]
=\prod_F\E\abs{\kappa_F}
\le\prod_F\binom{\nu}{j_F}^{-1/2},
\]
using $\abs{\prod U}\le1$, independence of the $W_F$, and Cauchy--Schwarz in each factor.  For the
second inequality in the statement, $\binom{\nu}{a}\binom{\nu}{b}\ge\binom{\nu}{a+b}$ for
$a+b\le\nu$, since their ratio is $\binom{a+b}{a}\,(\nu)_a(\nu)_b/(\nu)_{a+b}\ge1$, where
$(\nu)_a$ is the falling factorial; iterating merges the cosets.  The third uses
$\binom{\nu}{\abs X}\ge(\nu/\abs X)^{\abs X}$ at $\nu=N/2^\rho$.
\end{proof}

Write $d=\abs X+\abs Y$.  The theorem and the symmetries combine as follows.

\begin{corollary}[Unconditional decay in the range $\abs X\ge\abs Y$]\label{cor:revealed-decay-range}
For all finite $X,Y\subseteq\bits^n$ with $\abs X\ge\abs Y$, writing $d=\abs X+\abs Y$,
\[
\abs{\Gamma(X,Y)}\le2^{O(d^{2})}\,N^{-\frac{\max(\abs X,\abs Y)+\abs X}{4}}.
\]
\end{corollary}

\begin{proof}
If $Y=\varnothing$ then $\abs X\ge\abs Y$ and $\Gamma(X,\varnothing)=\mathbf 1[X=\varnothing]$ give
the bound directly, so assume $\abs X\ge\abs Y\ge1$, whence $d\ge2$ and $\log_2d\ge1$.  If
$2^{\abs Y}\abs X>N$, then $n<\abs Y+\log_2\abs X\le d+\log_2d$, so
$N^{\abs X/2}\le2^{(d+\log_2d)d/2}$ and the trivial bound $\abs{\Gamma(X,Y)}\le1$ is already of the
displayed form.  Otherwise $2^{\abs Y}\le N$, so $\abs Y\le n<N$ and $Y$ is a proper subset of
$\bits^n$; if $0\in Y$, replace $Y$ by $Y+s$ for some $s\notin Y$, which changes neither $\abs Y$ nor
$\abs{\Gamma}$ by Lemma~\ref{lem:gamma-symmetries}(\ref{it:covariance}).  Then
Theorem~\ref{thm:revealed-decay} applies, since $\abs X\le N/2^{\abs Y}\le N/2^{\rho}$, and gives
$\abs{\Gamma(X,Y)}\le(2^{\abs Y}\abs X)^{\abs X/2}N^{-\abs X/2}$ with
$(2^{\abs Y}\abs X)^{\abs X/2}\le2^{(d+\log_2d)d/2}$.
\end{proof}

What Theorem~\ref{thm:revealed-decay} cannot capture is decay in $\abs Y$: at $X=\varnothing$ it
gives only a vacuous bound, and the range $\abs X<\abs Y$, with the coefficients
$\widehat{g_Y}(\varnothing)=\E[g_Y]$ at its extreme, is what remains open.

\subsection{The conjectured decay rate and its consequences}\label{subsec:appendix-open}

We now reduce the $\FH_2$ lower bound for the sign ensemble to a single statement about $\Gamma$.

By the simulation lemma (Lemma~\ref{lem:round-embedding}), an $\FH_2$ circuit making $q$ phase
queries is a non-adaptive quantum query algorithm making $q$ parallel queries.  Its acceptance
probability $f:\pmone^{3N}\to[0,1]$ is a bounded function of the three oracles, and by the
non-adaptive growth bound of Corollary~\ref{cor:fh2-growth} its level-$\ell$ Fourier weight satisfies
\begin{equation}\label{eq:appendix-growth}
L_{1,\ell}(f)\ \le\ (\ell+1)\,\bigl(3Nq\bigr)^{\ell/4}.
\end{equation}
Write a Fourier character over the $3N$ oracle bits as a reduced monomial $M_{X,Y,Z}$ with
$X,Y,Z\subseteq\bits^n$ indexing the queried points of $F_0,F_1,F_2$ and set $\ell=\abs{X}+\abs{Y}+\abs{Z}$.
Since $\E_{\No_n}[M_{X,Y,Z}]=0$ for $(X,Y,Z)\neq(\varnothing,\varnothing,\varnothing)$ and, by
Proposition~\ref{prop:structured-factorization-rev},
$\E_{\Yes_n}[M_{X,Y,Z}]=\Gamma(X,Y)\,\Gamma(Z,Y)$, the distinguishing advantage of $f$ is
\begin{equation}\label{eq:appendix-advantage}
\bigl|\E_{\Yes_n}[f]-\E_{\No_n}[f]\bigr|
\ \le\
\sum_{\ell\ge 1} L_{1,\ell}(f)\cdot
\max_{\abs{X}+\abs{Y}+\abs{Z}=\ell}\bigl|\Gamma(X,Y)\,\Gamma(Z,Y)\bigr|.
\end{equation}
Every nonzero moment has $Y\neq\varnothing$: if $Y=\varnothing$ then
$\Gamma(X,\varnothing)=\E\prod_{x\in X}A(x)=\mathbf 1[X=\varnothing]$, so a nonzero moment with
$(X,Y,Z)\neq(\varnothing,\varnothing,\varnothing)$ must have $Y\neq\varnothing$.  The advantage is
therefore controlled entirely by the correlations $\Gamma(X,Y)$ with $Y\neq\varnothing$, and the
only remaining question is how fast they decay in $\abs{X}$ and $\abs{Y}$.

The exponent to conjecture is predicted by the following heuristic count of the leading correction.  The linear forms
$\widetilde A(y)$, $y\in Y$, are orthonormal, so a jointly Gaussian vector with their covariance has
independent coordinates and $\E\prod_y\sgn(G_y)=0$; the correlation is entirely a non-Gaussian
effect.  Signs and Gaussians have the same first three moments, so the first correction is of fourth
order, and for distinct $y_1,\dots,y_4$ a direct expansion gives
\begin{equation}\label{eq:fourth-cumulant}
\E\bigl[\widetilde A(y_1)\widetilde A(y_2)\widetilde A(y_3)\widetilde A(y_4)\bigr]
=-\frac2N\,\mathbf 1\bigl[y_1\oplus y_2\oplus y_3\oplus y_4=0\bigr],
\end{equation}
the two terms being one coincidence of the four summation points, contributing $+N^{-1}$, against
three pairings contributing $-N^{-1}$ each.  The leading non-Gaussian correction to the Hadamard coefficients is therefore a fourth cumulant of size
$N^{-1}$, supported on quadruples of vanishing sum; the higher cumulants are of lower order in $N$
and are ignored in this count.  Correlating $\abs Y$
signs therefore requires at least $\abs Y/4$ such corrections, at $N^{-1}$ each, for a rate
$N^{-\abs Y/4}$.  A point of $X$ is more expensive, at $N^{-1/2}$: it is fixed rather than summed over, so no
character cancellation occurs.  We conjecture the rate that this count predicts.

\begin{conjecture}[Low-degree decay of $\widehat{g_Y}$]\label{conj:sharp-decay}
For every constant $\alpha>0$ there is a constant $C=C(\alpha)\ge1$ such that for all
$X,Y\subseteq\bits^n$ with $Y\neq\varnothing$ and $d=\abs{X}+\abs{Y}\le n^{\alpha}$,
\[
\abs{\Gamma(X,Y)}\ \le\ d^{\,Cd}\,N^{-\frac{\max(\abs{X},\abs{Y})+\abs{X}}{4}}.
\]
\end{conjecture}

Since $\max(\abs X,\abs Y)+\abs X\ge\abs X+\abs Y$, Conjecture~\ref{conj:sharp-decay} implies the
weaker rate $d^{\,Cd}N^{-d/4}$, with which the argument below also runs, at the cost of a smaller
query threshold.  In the range $\abs X\ge\abs Y$ the exponent of the conjecture is a theorem: there it equals
$\abs X/2$, and Corollary~\ref{cor:revealed-decay-range} proves this rate unconditionally, though
with the prefactor $2^{O(d^{2})}$ rather than $d^{Cd}$.  The exponent of $N$ is therefore new only in the range
$\abs X<\abs Y$, while the prefactor $d^{Cd}$ is conjectural in both ranges;
Remark~\ref{rem:what-remains-conditional} records where each is used.

An invariance principle compares $g_Y$ with its Gaussian analogue, for which
$\E\prod_y\sgn(G_y)=0$ by the orthonormality of the linear forms, and returns a single additive
error in which the number of threshold functions enters only through the prefactor.  The conjecture asserts
something of a different type: each further element of $Y$ contributes a further factor $N^{-1/4}$,
so that the error compounds rather than accumulates, as the count following
\eqref{eq:fourth-cumulant} predicts.  Two nearby results do not give it.  Fixed-dimensional lattice
Edgeworth expansions~\cite{bhattacharya1976normal} fix both the dimension and the order of the
expansion, and do not control the constants as $d$ grows.  The Fourier concentration theorem for
functions of $m$ halfspaces~\cite{klivans2002learning} bounds the total Fourier weight above a given
degree, whereas $\Gamma(X,Y)$ is one coefficient at a fixed degree, and that theorem carries no
decay in $N$.  We regard the conjecture as the main analytic question raised by this ensemble.

The refined exponent enters through the following elementary minimization.

\begin{lemma}[The worst configuration at a fixed level]\label{lem:level-exponent}
Let $\ell\ge1$ and let $a,b,c$ be nonnegative integers with $a+b+c=\ell$ and $b\ge1$.  Then
\[
\frac{\max(a,b)+a}{4}+\frac{\max(c,b)+c}{4}\ \ge\ \frac{\ell}{3},
\]
with equality when $a=b=c$.
\end{lemma}

\begin{proof}
Write $4E$ for four times the left-hand side and distinguish four cases.  If $a\ge b$ and $c\ge b$,
then $4E=2(a+c)=2(\ell-b)$, and $\ell\ge3b$ gives $4E\ge\tfrac43\ell$.  If $a<b$ and $c<b$, then
$4E=2b+a+c=\ell+b$, and $\ell<3b$ gives $4E>\tfrac43\ell$.  If $a\ge b$ and $c<b$, then
$4E=2a+b+c$, and $4E\ge\tfrac43(a+b+c)$ is equivalent to $2a\ge b+c$, which holds because $a\ge b$
and $a>c$.  The case $a<b$, $c\ge b$ is symmetric.  At $a=b=c$ both maxima equal $a$ and
$4E=\tfrac43\ell$.
\end{proof}

\begin{proposition}[Conditional $\FH_2$ lower bound]\label{prop:conditional-lower-bound}
Assume Conjecture~\ref{conj:sharp-decay}, and fix a constant $\gamma>0$.  Then for all large $n$,
every $\FH_2$ circuit making $q\le n^{\gamma}$ phase queries distinguishes $\Yes_n$ from $\No_n$ with
advantage at most $2^{-\Omega(n)}$, the constant in the exponent depending on $\gamma$.  Combined with
Theorem~\ref{thm:FH3-pointwise}, which decides the promise problem of
Definition~\ref{def:sign-promise-sets} at Fourier depth three, this would make the sign ensemble an
explicitly defined hard distribution for the oracle separation $\FH_2\subsetneq\FH_3$
(Remark~\ref{rem:sign-or}).
\end{proposition}

\begin{proof}
If $q=0$ the acceptance probability does not depend on the oracle and the advantage is $0$, so
assume $q\ge1$.  On a nonzero term of \eqref{eq:appendix-advantage} we have $Y\neq\varnothing$, and since $f$ has
degree at most $2q$ only levels $\ell=\abs{X}+\abs{Y}+\abs{Z}\le2q$ contribute.  Both factor degrees satisfy
$\abs{X}+\abs{Y},\,\abs{Z}+\abs{Y}\le\ell\le2q\le n^{\gamma+1}$, so Conjecture~\ref{conj:sharp-decay} with $\alpha=\gamma+1$
applies to each, with $C=C(\gamma+1)$, and Lemma~\ref{lem:level-exponent} with $(a,b,c)=(\abs{X},\abs{Y},\abs{Z})$ gives
\[
\bigl|\Gamma(X,Y)\,\Gamma(Z,Y)\bigr|
\ \le\
\ell^{2C\ell}\,
N^{-\frac{\max(\abs{X},\abs{Y})+\abs{X}}{4}-\frac{\max(\abs{Z},\abs{Y})+\abs{Z}}{4}}
\ \le\
\ell^{2C\ell}\,N^{-\ell/3}.
\]
Substituting this and \eqref{eq:appendix-growth} into \eqref{eq:appendix-advantage},
\[
\bigl|\E_{\Yes_n}[f]-\E_{\No_n}[f]\bigr|
\ \le\
\sum_{\ell=1}^{2q} (\ell+1)\,(3Nq)^{\ell/4}\cdot\ell^{2C\ell}N^{-\ell/3}
\ =\
\sum_{\ell=1}^{2q} (\ell+1)\bigl(3q\bigr)^{\ell/4}\ell^{2C\ell}\,N^{-\ell/12}.
\]
Unlike the corresponding sum for the rate $N^{-d/4}$, this one decays in $\ell$: the conjectured
exponent exceeds the growth exponent $\ell/4$ by $\ell/12$ at every level.  Taking logarithms, the
$\ell$-th term is $2^{\,O(\ell\log(q\ell))-n\ell/12}$, and $\ell\le2q\le2n^{\gamma}$ gives
$\log(q\ell)=O(\log n)$, so the exponent is $O(\ell\log n)-n\ell/12\le-n\ell/24$ for large $n$.
Every term is therefore at most $2^{-n\ell/24}$ and the sum is at most $2^{-n/24+1}$.
\end{proof}

\begin{remark}[The remaining conditional content]\label{rem:what-remains-conditional}
The hypothesis can be weakened, though not for the whole range of query counts.  In the proof, a
factor $\Gamma(X,Y)$ with $\abs X\ge\abs Y$ need not appeal to the conjecture:
Corollary~\ref{cor:revealed-decay-range} gives the same exponent
$N^{-\abs X/2}=N^{-(\max(\abs X,\abs Y)+\abs X)/4}$ unconditionally, but with prefactor
$2^{O(\ell^{2})}$ rather than $\ell^{2C\ell}$.  Substituting it changes the $\ell$-th exponent in the
sum above to $O(\ell^{2})+O(\ell\log n)-n\ell/12$, and the quadratic term is absorbed only when
$\ell=o(n)$.  Since $\ell\le2q$, the substitution is therefore legitimate whenever $q=o(n)$, and
in that range Conjecture~\ref{conj:sharp-decay} is needed only for the pairs with $\abs X<\abs Y$.
For larger polynomial query counts the conjectured prefactor $d^{Cd}$ is
used on both factors, and we do not know how to avoid it.

Even in the restricted range, this locates the conditional content of the lower bound: it lies in the
correlations with $\abs Y>\abs X$, and at the extreme $X=\varnothing$ in the correlations
$\Gamma(\varnothing,Y)$ of the signs alone.
\end{remark}

Two features of the argument remove distinct obstacles.  The first is the shape of the conjectured
rate.  Under the weaker rate $d^{Cd}N^{-d/4}$ the product bound
is only $\ell^{2C\ell}N^{-(\ell+\abs Y)/4}$, and its worst case $\abs Y=1$ leaves no decay in $\ell$
at all; the level sum then collapses to $N^{-1/4}2^{O(q\log q)}$, which is $o(1)$ only for
$q=o(n/\log n)$.  The rate of Conjecture~\ref{conj:sharp-decay} is designed to remove exactly that
case, and Lemma~\ref{lem:level-exponent} converts the improvement into a uniform factor $N^{-\ell/12}$.
The second is the growth bound.  Using the general Corollary~\ref{cor:FHk-growth} instead of the
non-adaptive one would replace the factor $(3Nq)^{\ell/4}$ by $9^{\ell}q^{\ell}(3N)^{\ell/4}$, a
larger query exponent that costs a constant factor in the admissible range of $q$.  Within the range
of degrees the conjecture covers, namely $d\le n^{\alpha}$, the conclusion available is hardness against
every polynomial query bound in $n$, as
Proposition~\ref{prop:conditional-lower-bound} states; and by
Remark~\ref{rem:what-remains-conditional}, for $q=o(n)$ only the range $\abs X<\abs Y$ is at issue, while
for larger polynomial $q$ the conjectured prefactor is needed in both ranges.

\section{A single-exponential constant in the Fourier-growth bound}\label{app:gstw-growth-constant}

Our phase-query lower bound uses the Fourier growth theorem of Girish, Sinha, Tal, and Wu through
Theorem~\ref{thm:gstw-growth}, which is their Corollary~4.2, derived from
\cite[Theorem~4.1]{girish2024power}.  At a fixed Fourier depth the constant in that bound is
absorbed and plays no role.  The growing-depth statements are different: there the depth grows with
the input length, and the constant has to be singly exponential in $r\ell$ for the argument to
survive.  This appendix proves the theorem in that form.

The quantity to be evaluated is an assignment count on the Boolean lattice.  In the step that
rewrites $h(s)$ as $\sum_{s'}P[s,s']\,g(s')$, the coefficient $P[s,s']$ counts the admissible
families $(J^{(b)})_{b\in B}$ for a query tuple of size profile $s'$.\footnote{Cross-references are
to the full version, arXiv:2311.16057v2, where this step appears in the proof of Theorem~4.1; the
published proceedings version~\cite{girish2024power} gives the proof overview and defers the
detailed argument to the full version.}  We evaluate it in closed form by M\"obius inversion,
which gives the factor $(2^{2r}-1)^{2\ell}=2^{O(r\ell)}$, uniformly over all levels $\ell\ge1$.
Theorem~\ref{thm:gstw-growth} is stated with this constant, and it is what the growing-depth
separation of Corollary~\ref{cor:growing-k} uses.

We retain the notation of the proof of~\cite[Theorem~4.1]{girish2024power}.  Thus $d\ge2$ is the
number of oracle factors in the product defining $f$, $\ell$ is the Fourier level, and
\[
B:=\{b\in\bits^d: \Norm{b}_1\equiv1\bmod2\}.
\]
For a tuple of query multisets $(I_1,\ldots,I_d)$, let $I^{(b)}$ be the set of coordinates lying in
exactly those $\oplus I_i$ with $b_i=1$, as in~\cite{girish2024power}.  We consider only tuples for
which $\abs{\oplus_i I_i}=\ell$, so that the $I^{(b)}$, $b\in B$, partition the symmetric difference,
with $\sum_{b\in B}\abs{I^{(b)}}=\ell$.  Following~\cite{girish2024power} we record the sizes by
$s'=(s'^{(b)})_{b\in B}\in\mathbb N^B$ with $s'^{(b)}=\abs{I^{(b)}}$ and $\Norm{s'}_1=\ell$, and let
$g(s')$ denote the contribution of all such tuples with $\abs{I^{(b)}}=s'^{(b)}$ for every $b$,
exactly as in the cited proof.  For a second vector $s\in\mathbb N^B$ with $\Norm{s}_1=\ell$, the
quantity $h(s)$ is the corresponding sum over the query tuples and over the families
$(J^{(b)})_{b\in B}$ of pairwise disjoint sets satisfying
\[
 \abs{J^{(b)}}=s^{(b)}
 \qquad\text{and}\qquad
 J^{(b)}\subseteq\bigcup_{b'\ge b}I^{(b')} .
\]

The coefficient is an assignment count, and it is convenient to compute it over a slightly
larger index set.  Let
\[
\mathcal C:=\bits^d\setminus\{0^d\},
\qquad
K:=\abs{\mathcal C}=2^d-1,
\]
ordered coordinatewise, and let $\mathcal S_\ell$ be the set of vectors $s\in\mathbb N^{\mathcal C}$ with
$\Norm{s}_1=\ell$.  Extend $g$ by zero to $\mathcal S_\ell$, and for $s\in\mathcal S_\ell$ let $h_{\mathcal C}(s)$
be defined exactly as $h(s)$, except that the pairwise disjoint sets $J^{(c)}$ are indexed by all
$c\in\mathcal C$, subject to the same constraints $\abs{J^{(c)}}=s^{(c)}$ and
$J^{(c)}\subseteq\bigcup_{b\in B,\,b\ge c}I^{(b)}$.  For $s$ supported on $B$ this is $h(s)$.

\begin{lemma}[The assignment count in closed form]\label{lem:gstw-exact-count}
For $s,s'\in\mathcal S_\ell$, let $P_{\mathcal C}[s,s']$ be the number of maps
\[
\varphi:\{(\gamma,j):\gamma\in\mathcal C,\ 1\le j\le s'^{(\gamma)}\}\to\mathcal C
\qquad\text{with}\qquad
\varphi(\gamma,j)\le\gamma
\quad\text{and}\quad
\abs{\varphi^{-1}(c)}=s^{(c)}\ \text{for all }c .
\]
Then $h_{\mathcal C}(s)=\sum_{s'\in\mathcal S_\ell}P_{\mathcal C}[s,s']\,g(s')$; moreover $P_{\mathcal C}$ is
invertible, and
\[
\Norm{P_{\mathcal C}^{-1}}_1\ \le\ K^{\ell}=(2^d-1)^{\ell}.
\]
\end{lemma}

\begin{proof}
For a fixed query tuple, the sets $I^{(b)}$, $b\in B$, are disjoint and their union, the symmetric
difference, has cardinality $\ell$.  The sets $J^{(c)}$ are pairwise disjoint subsets of that union
of total size $\ell$, so they partition it.  Assigning each element of $I^{(b)}$ to the unique
$J^{(c)}$ containing it, which is permitted precisely when $c\le b$, is a bijection between the
admissible families $(J^{(c)})_{c\in\mathcal C}$ and the maps counted by $P_{\mathcal C}[s,s']$, and the count
depends only on the size profile $s'$.  Summing the weights over the query tuples gives the
identity.

For the inversion, $P_{\mathcal C}$ is the assignment matrix of Lemma~\ref{lem:standard-cert-count} for
the lattice $\mathcal C$: the generating identity
\[
\sum_{s\in\mathcal S_\ell}P_{\mathcal C}[s,s']\prod_{c\in\mathcal C}y_c^{\,s^{(c)}}
=\prod_{\gamma\in\mathcal C}\Bigl(\sum_{c\le\gamma}y_c\Bigr)^{s'^{(\gamma)}}
\]
exhibits it as the matrix of the substitution $y_\gamma\mapsto\sum_{c\le\gamma}y_c$ on homogeneous
polynomials of degree $\ell$, which M\"obius inversion on the Boolean
lattice~\cite[the M\"obius inversion formula, Prop.~3.7.1, and the Boolean lattice,
Ex.~3.8.3]{stanley2012enumerative} inverts explicitly; the
inverse substitutes $\sum_{c\le\gamma}(-1)^{\Norm{\gamma}_1-\Norm{c}_1}y_c$ for $y_\gamma$.  A
column of $P_{\mathcal C}^{-1}$ lists the coefficients of a product of $\ell$ linear forms, each with at
most $K$ terms of unit modulus, so every absolute column sum is at most $K^{\ell}$.
\end{proof}

\begin{lemma}[The entrywise bound over the enlarged index set]\label{lem:gstw-entrywise-extended}
In the setting of~\cite[Lemma~4.3]{girish2024power}, for every $s\in\mathcal S_\ell$,
\[
\abs{h_{\mathcal C}(s)}
\ \le\
t^{\ell}\,
\max\Bigl\{1,\ \bigl(\widetilde n/t\bigr)^{\frac12\lfloor(d-1)\ell/d\rfloor}\Bigr\},
\]
where $t$ is the number of queries per factor and $\widetilde n$ the number of free variables.
\end{lemma}

\begin{proof}
\emph{Re-indexing the certificate registers by $\mathcal C$.}
In the operator factorization proving the entrywise bound on $h(s)$
in~\cite[Lemma~4.3]{girish2024power}, index the certificate sets $J^{(c)}$ and the matrices
$Q_i,Q_i'$ and $W$ by $\mathcal C$ in place of $B$; the matrices $R_i,R_i'$ are unchanged.

\emph{The norm estimates survive the re-indexing.}
They use only that every $c\in\mathcal C$ has a first and a last nonzero coordinate.  For a
split at position $r$ they give $\Norm{Q_i}\le\prod\binom{t}{s^{(c)}}^{1/2}$, over the $c$ whose
first nonzero coordinate is $i$, for $i\le r$; $\Norm{Q_i'}\le\prod\binom{t}{s^{(c)}}^{1/2}$, over
the $c$ whose last nonzero coordinate is $i$, for $i\ge r$; and
$\Norm{W}\le\prod\binom{\widetilde n}{s^{(c)}}^{1/2}$, over the $c$ vanishing on the first $r$
coordinates together with the $c$ vanishing on the last $d-r+1$.  Each $c$ therefore contributes one
factor for the initial segment and one for the final segment: a $\binom{t}{s^{(c)}}^{1/2}$ on a side
where it is present, and a $\binom{\widetilde n}{s^{(c)}}^{1/2}$ on a side where it is absent.
Writing $e_r$ for the total size of the $c$ absent from one of the two sides, and using
$\binom{t}{s^{(c)}}\le t^{s^{(c)}}$ and $\binom{\widetilde n}{s^{(c)}}\le\widetilde n^{s^{(c)}}$,
the product of the norms is at most
$t^{(2\ell-e_r)/2}\widetilde n^{e_r/2}=t^{\ell}(\widetilde n/t)^{e_r/2}$.

\emph{Averaging over the split position.}
Writing $z(c)$ for the number of leading zeros of $c$
and $z'(c)-1$ for the position of its last nonzero coordinate,
\[
\sum_{r=1}^{d}e_r=\sum_{c\in\mathcal C}s^{(c)}\bigl(z(c)+d+1-z'(c)\bigr)\le(d-1)\ell,
\]
since every nonzero $c$ has $z'(c)-z(c)\ge2$; some $r$ therefore has
$e_r\le\lfloor(d-1)\ell/d\rfloor$.

\emph{The modified constraints define the same quantity.}
They are
$\abs{J^{(c)}}=s^{(c)}$, $J^{(c)}\subseteq\bigcap_{i:c_i=1}\oplus I_i$, and
$\bigcup_c J^{(c)}=\oplus I_1\oplus\cdots\oplus I_d$.  Since $\sum_c s^{(c)}=\ell$ equals the
cardinality of the symmetric difference, the union condition forces the $J^{(c)}$ to be pairwise
disjoint and every element to lie in some $I^{(b)}$ with $b\in B$ and $b\ge c$, so these constraints
are equivalent to the ones defining $h_{\mathcal C}(s)$, and the remainder of the cited proof is
unchanged.
\end{proof}

\begin{corollary}[The growth theorem with a single-exponential constant]\label{cor:gstw-growth-constant}
For every $d\ge2$, every $\ell\ge1$, and every restriction leaving $\widetilde n$ free variables,
\[
L_{1,\ell}(f|_\rho)
\ \le\
(2^d-1)^{2\ell}\,t^{\ell}\,
\max\Bigl\{1,\ \bigl(\widetilde n/t\bigr)^{\frac12\lfloor(d-1)\ell/d\rfloor}\Bigr\}.
\]
With $d=2r$ this is Theorem~\ref{thm:gstw-growth}.
\end{corollary}

\begin{proof}
As in~\cite{girish2024power}, $L_{1,\ell}(f|_\rho)\le\sum_{s'}\abs{g(s')}$.  By
Lemma~\ref{lem:gstw-exact-count},
\[
\sum_{s'\in\mathcal S_\ell}\abs{g(s')}
\le\Norm{P_{\mathcal C}^{-1}}_1\sum_{s\in\mathcal S_\ell}\abs{h_{\mathcal C}(s)}
\le K^{\ell}\cdot\abs{\mathcal S_\ell}\cdot\max_{s}\abs{h_{\mathcal C}(s)} ,
\]
and $\abs{\mathcal S_\ell}=\binom{\ell+K-1}{K-1}\le K^{\ell}$, since a vector in $\mathcal S_\ell$ is a
multiset of $\ell$ elements of $\mathcal C$.  Combining with
Lemma~\ref{lem:gstw-entrywise-extended} gives the stated bound.

The form without the maximum, in which Theorem~\ref{thm:gstw-growth} is stated, follows for every
$t\ge1$.  Write $a=\frac12\lfloor(d-1)\ell/d\rfloor$.  If $\widetilde n\ge t$ the maximum equals
$(\widetilde n/t)^{a}$.  If $1\le\widetilde n<t$, then $\abs f\le1$ gives
$\abs{\widehat{f|_\rho}(S)}\le1$ for every $S$, so
$L_{1,\ell}(f|_\rho)\le\binom{\widetilde n}{\ell}\le\widetilde n^{\ell}\le t^{\ell-a}\widetilde n^{a}
=t^{\ell}(\widetilde n/t)^{a}$, using $\widetilde n\le t$ and $a\le\ell$; and if $\widetilde n=0$ then
$L_{1,\ell}(f|_\rho)=0$ for $\ell\ge1$.
\end{proof}

A single-exponential constant is what permits the Fourier depth to grow in
Corollary~\ref{cor:growing-k}.  A factor whose logarithm were itself exponential in $r\ell$ would
support only $k=O((\log n)^{1/3})$ there, whereas $2^{O(r\ell)}$ supports $k=O(n^{1/3})$, the
binding constraint then becoming the $2^{\Theta(k)}$ queries of the deciding circuit.

\section{Proof of the non-adaptive growth bound}\label{app:nonadaptive-proof}

We prove Theorem~\ref{thm:nonadaptive-growth}.  The bound is stated by Girish, Sinha, Tal, and
Wu~\cite[Remark~1.6]{girish2024power}, and the argument below is theirs: it follows the technical
proof overview in the full version for the case of two oracle factors, where the level-$\ell$ weight
is grouped by the pair of sizes $(s_1,s_2)$.  Their Section~4 proves the general form
(Theorem~\ref{thm:gstw-growth}) rather than this sharper non-adaptive one.  We write the proof out
because our application depends on how the constant grows with $\ell$, and because the resulting
bound is stable under restrictions, which we use in Appendix~\ref{sec:appendix-sign}; this case also
does not involve the assignment count of Appendix~\ref{app:gstw-growth-constant}.  We
claim no novelty for the method.

\begin{proof}
The proof has five steps: (i) absorb the restriction into the boundary vectors; (ii) split the
level-$\ell$ weight according to the shape of the symmetric difference $\oplus I\triangle\oplus J$;
(iii) write each piece as a product of two matrices in two different ways; (iv) bound the norms of
the four factors; (v) take the geometric mean of the two resulting bounds.

\emph{Step 1: the restriction.}  Since $O_x^{\otimes t}\ket I=\bigl(\prod_{p\in[t]}x_{I(p)}\bigr)\ket I$
and $x_i^2=1$, writing $\oplus I\subseteq[M]$ for the set of indices occurring an odd number of times
in the tuple $I$, the index $0$ being ignored, we have
\[
f(x)=\sum_{(I,a),(J,b)\in A}\overline{u[I,a]}\;\mathsf B[(I,a),(J,b)]\;v[J,b]\;
x^{\oplus I}\,x^{\oplus J},
\qquad x^{S}:=\prod_{i\in S}x_i.
\]
Let $F\subseteq[M]$ be the free coordinates of $\rho$, so $\abs F=\widetilde M$.  Splitting each
product over fixed and free coordinates gives $x^{\oplus I}=\varepsilon_\rho(I)\,x^{\oplus I\cap F}$
with $\varepsilon_\rho(I)\in\pmone$ determined by $\rho$.  Hence $f_\rho$ has the same form with $u$
and $v$ replaced by their entrywise products with $\varepsilon_\rho$, which leaves their norms
unchanged, and with $\oplus I$ replaced by $\oplus I\cap F$.  As $\oplus I\cap F$ is a subset of $F$
containing at most $t$ elements, $\abs{\oplus I\cap F}\le\tau$.  We may therefore assume that $\rho$
is empty, that $F=[M]$ with $\widetilde M=M$, and that $\abs{\oplus I}\le\tau$ for every $I$.

\emph{Step 2: splitting by shape.}  Choose unimodular phases $\alpha(S)$ with
$L_{1,\ell}(f)=\sum_{\abs S=\ell}\alpha(S)\widehat f(S)$.  By the identity above,
$\widehat f(S)$ is the sum of $\overline{u[I,a]}\mathsf B[(I,a),(J,b)]v[J,b]$ over the pairs with
$\oplus I\triangle\oplus J=S$.  For $s=(s_1,s_2)$ let $\mathbf 1_{s}(S_1,S_2)$ be the indicator of
$\abs{S_1\setminus S_2}=s_1$ and $\abs{S_2\setminus S_1}=s_2$, and set
\[
g(s):=\sum_{(I,a),(J,b)}
\mathbf 1_{s}(\oplus I,\oplus J)\,\alpha(\oplus I\triangle\oplus J)\,
\overline{u[I,a]}\,\mathsf B[(I,a),(J,b)]\,v[J,b].
\]
The $\ell+1$ pairs $s$ with $s_1+s_2=\ell$ partition the constraint
$\abs{\oplus I\triangle\oplus J}=\ell$, so $L_{1,\ell}(f)=\sum_{s_1+s_2=\ell}g(s)$.

\emph{Step 3: two factorizations.}  Introduce a register indexed by the subsets $T\subseteq[M]$ and
define
\[
\begin{aligned}
W\bigl[(I,a),((I',a'),T)\bigr]
 &=\mathbf 1[(I,a)=(I',a')]\,\mathbf 1_{s}(\oplus I,T)\,\alpha(\oplus I\triangle T), \\
R'\bigl[((I',a'),T),(J,b)\bigr]
 &=\mathbf 1[\oplus J=T]\,\mathsf B[(I',a'),(J,b)], \\
R\bigl[(I,a),((J',b'),T)\bigr]
 &=\mathbf 1[\oplus I=T]\,\mathsf B[(I,a),(J',b')], \\
W'\bigl[((J',b'),T),(J,b)\bigr]
 &=\mathbf 1[(J,b)=(J',b')]\,\mathbf 1_{s}(T,\oplus J)\,\alpha(T\triangle\oplus J).
\end{aligned}
\]
Carrying out both products, the entries $(WR')[(I,a),(J,b)]$ and $(RW')[(I,a),(J,b)]$ are both equal
to
\[
\mathbf 1_{s}(\oplus I,\oplus J)\,\alpha(\oplus I\triangle\oplus J)\,\mathsf B[(I,a),(J,b)],
\qquad\text{so}\qquad
g(s)=u^{\dagger}WR'v=u^{\dagger}RW'v.
\]
We use two such factorizations because $W$ imposes the size constraint while propagating $\oplus J$
forward and $W'$ imposes it while propagating $\oplus I$ backward, so that the two binomial factors
appear in the opposite order in the two bounds.

\emph{Step 4: norms.}  Group the rows of $R'$ by their $T$ component and group its columns by
$\oplus J$.
A row with component $T$ is nonzero only in columns with $\oplus J=T$, so the matrix is block
diagonal, and the block at $T$ is the submatrix of $\mathsf B$ on the columns $\{(J,b):\oplus J=T\}$;
hence $\Norm{R'}\le\Norm{\mathsf B}\le1$.  Grouping the rows of $R$ by $\oplus I$ and its columns by
$T$ gives $\Norm R\le1$ in the same way.

For $W$ use $\Norm W\le\sqrt{\Norm W_1\Norm W_\infty}$, with $\Norm\cdot_1$ the largest absolute
column sum and $\Norm\cdot_\infty$ the largest absolute row sum.  Each column of $W$ has at most one
nonzero entry, of modulus one, so $\Norm W_1\le1$.  The row $(I,a)$ is nonzero only in the columns
$((I,a),T)$ with $\mathbf 1_{s}(\oplus I,T)=1$, that is, with $T$ obtained from $\oplus I$ by
deleting $s_1$ of its elements and adjoining $s_2$ elements from outside it; there are
$\binom{\abs{\oplus I}}{s_1}\binom{M-\abs{\oplus I}}{s_2}\le\binom{\tau}{s_1}\binom{M}{s_2}$ such
$T$.  Symmetrically each row of $W'$ has at most one nonzero entry, while the column $(J,b)$ is
nonzero only in the rows $((J,b),T)$ with $T$ obtained from $\oplus J$ by deleting $s_2$ elements
and adjoining $s_1$.  Hence
\[
\Norm W\le\sqrt{\tbinom{\tau}{s_1}\tbinom{M}{s_2}},
\qquad
\Norm{W'}\le\sqrt{\tbinom{\tau}{s_2}\tbinom{M}{s_1}}.
\]

\emph{Step 5: the geometric mean.}  Both factorizations bound $\abs{g(s)}$ by a product of norms, so
$\abs{g(s)}$ is at most the minimum of the two bounds, hence at most their geometric mean:
\[
\abs{g(s)}
\ \le\
\Bigl[\tbinom{\tau}{s_1}\tbinom{M}{s_2}\tbinom{\tau}{s_2}\tbinom{M}{s_1}\Bigr]^{1/4}
\ \le\
\bigl[\tau^{s_1+s_2}M^{s_1+s_2}\bigr]^{1/4}
=(M\tau)^{\ell/4}.
\]
Summing over the $\ell+1$ values of $s$ gives $L_{1,\ell}(f)\le(\ell+1)(M\tau)^{\ell/4}$, which is
the claim.  For the final assertion, a non-adaptive algorithm has the form
$U_1(O_x^{\otimes t}\otimes I_m)U_0$, so its acceptance probability is
$\bra{\psi}(O_x^{\otimes t}\otimes I_m)U_1^{\dagger}\Pi U_1(O_x^{\otimes t}\otimes I_m)\ket{\psi}$
with $\ket\psi=U_0\ket{\psi_0}$ a unit vector, using that $O_x$ is Hermitian.
\end{proof}

\end{document}